\documentclass[11pt, oneside]{article}
\usepackage{geometry}                		
\usepackage{graphicx}				
\usepackage{array}
\usepackage{wrapfig}
\usepackage{xcolor}
\usepackage{amssymb}
\usepackage{amsthm}
\usepackage{amsmath}
\usepackage{physics}
\usepackage{braket}
\usepackage{multirow}
\usepackage{mathtools}
\usepackage{cancel}
\usepackage{tikz}
\usetikzlibrary{decorations.pathreplacing, decorations.pathmorphing}
\usetikzlibrary{positioning}
\usetikzlibrary{shapes,patterns}
\usepackage{listings}
\usepackage{tabu}
\usepackage[toc,page]{appendix}
\usepackage[margin=.2cm,font=footnotesize]{caption}
\definecolor{DarkGreen}{rgb}{0.1,0.5,0.1}
\definecolor{DarkRed}{rgb}{0.5,0.1,0.1}
\definecolor{DarkBlue}{rgb}{0.1,0.1,0.5}
\usepackage{hyperref}
\hypersetup{
    unicode=false,         
    pdftoolbar=true,       
    pdfmenubar=true,       
    pdffitwindow=false,     
    pdfnewwindow=true,     
    colorlinks=true,      
    linkcolor=DarkBlue,         
    citecolor=DarkGreen,       
    filecolor=DarkGreen,     
    urlcolor=DarkBlue,         
    pdftitle={},
    pdfauthor={},
}

\def\BBLMTU{\ensuremath{\mathcal{T}}}
\newtheorem{thm}{Theorem}[section]
\newtheorem{conj}[thm]{Conjecture}
\newtheorem*{thm*}{Theorem}
\newtheorem{remark}[thm]{Remark}
\newtheorem{rmk}[thm]{Remark}
\newtheorem{cor}[thm]{Corollary}
\newtheorem{lem}[thm]{Lemma}
\newtheorem*{lem*}{Lemma}
\newtheorem{define}[thm]{Definition}
\newtheorem{proposition}[thm]{Proposition}
\newtheorem*{proposition*}{Proposition}
\newtheorem{observation}[thm]{Observation}

\newtheorem{claim}[thm]{Claim}
\newtheorem*{claim*}{Claim}

\DeclareMathOperator{\AND}{AND}
\DeclareMathOperator{\conv}{conv}
\DeclareMathOperator{\NAND}{NAND}

\DeclareMathOperator{\XOR}{XOR}

\DeclareMathOperator{\Maj}{Maj}

\DeclareMathOperator{\dual}{dual}
\DeclareMathOperator{\Amp}{Amp}
\DeclareMathOperator{\Ber}{Ber}

\renewcommand{\epsilon}{\varepsilon}
\newcommand{\eps}{\epsilon}

\newcommand{\R}{\mathbb{R}}
\newcommand{\F}{\mathbb{F}}
\renewcommand{\P}{\mathbb{P}}
\newcommand{\PR}[1]{\mathbb{P}\left[#1\right]}

\newcommand{\cC}{\mathcal{C}}
\newcommand{\etal}{\textit{et al. }}

\usepackage[normalem]{ulem}

\allowdisplaybreaks

\begin{document}

\title{Tight Limits on Nonlocality from Nontrivial Communication Complexity; a.k.a. Reliable Computation with Asymmetric Gate Noise}

\author{
    Noah Shutty \footnote{Email: \texttt{\{noaj, marykw, phayden\}@stanford.edu}} \thanks{Department of Physics, Stanford University.} \thanks{N.S. was supported in part by NSF DGE-1656518.}
    \and Mary Wootters \footnotemark[2] \thanks{Departments of Computer Science and Electrical Engineering, Stanford University.  M.W. was supported in part by NSF CAREER CCF-1844628.}
    \and Patrick Hayden \footnotemark[2] \footnotemark[3] \thanks{P.H. was supported by AFOSR (FA9550-16-1-0082), CIFAR and the Simons Foundation}
}

\maketitle

\begin{abstract}
It has long been known that the existence of certain superquantum nonlocal correlations would cause communication complexity to collapse. The absurdity of a world in which any nonlocal binary function could be evaluated with a constant amount of communication in turn provides a tantalizing way to distinguish quantum mechanics from incorrect theories of physics; the statement ``communication complexity is nontrivial’’ has even been conjectured to be a concise information-theoretic axiom for characterizing quantum mechanics. We directly address the viability of that perspective with two results. First, we exhibit a nonlocal game such that communication complexity collapses in any physical theory whose maximal winning probability exceeds the quantum value. Second, we consider the venerable CHSH game that initiated this line of inquiry. In that case, the quantum value is about 0.85 but it is known that a winning probability of approximately 0.91 would collapse communication complexity. We provide evidence that the 0.91 result is the best possible using a large class of proof strategies, suggesting that the communication complexity axiom is insufficient for characterizing CHSH correlations. Both results build on new insights about reliable classical computation. The first exploits our formalization of an equivalence between amplification and reliable computation, while the second  follows from an upper bound on the threshold for reliable computation with formulas of noisy XOR and AND gates.

\end{abstract}

\thispagestyle{empty}
\setcounter{page}{0}
\newpage

\section{Introduction}\label{sec:intro}

Quantum mechanics is mysterious, so it is appealing to look for a concise information-theoretic principle that explains quantum mechanical phenomena.  One such principle might be ``communication complexity is nontrivial.'' That is, two parties with inputs $x \in \{0,1\}^n$ and $y \in \{0,1\}^n$ respectively should not be able to compute arbitrary functions $f(x,y)$ with high probability, using only a constant amount of communication (independent of $n$). 

It is known that the axiom ``communication complexity is nontrivial'' does in fact rule out some superquantum phenomena, specifically superquantum success at certain \em nonlocal games. \em  
For example, consider the famous \em CHSH game\em.\footnote{This game is named after
 Clauser, Horne, Shimony, and Holt and was introduced implicitly in their paper~\cite{chsh}.}
The two players, Alice and Bob, cannot communicate.  Alice and Bob receive independent random bits $x$ and $y$ respectively.  Their goal is to output bits $a$ and $b$, respectively, so that $a \oplus b = x \land y$.

In a classical world, Alice and Bob can win the CHSH game with probability $3/4$ (\em e.g. \em by outputting $a,b= 0$) and cannot do any better; thus the classical value of the CHSH game is 
$ \omega_C(CHSH) = \frac{3}{4}.$
If Alice and Bob have access to any \em nonsignalling \em correlation---that is, they can produce correlated bits $a$ and $b$ in any way they like as long as they do not gain the ability to communicate---then they can win the CHSH game with probability $1$; we say that the nonsignalling value of the CHSH game is
$ \omega_{NS}(CHSH) = 1. $
If instead Alice and Bob share quantum entanglement, they can do something in between $\omega_C$ and $\omega_{NS}$: it turns out that the quantum value of the CSHS game is \cite{cirel1980quantum}
\[ \omega_Q(CHSH) = \frac{1}{2} + \frac{1}{\sqrt{8}} \approx 0.8536.\]

Work of van Dam~\cite{implausibleConsequencesSuperstrongNonlocality} showed that if Alice and Bob could win the CHSH game with probability $1$, then communication complexity would become trivial.  This was extended by
Brassard \etal \cite{Brassard}, who showed that if Alice and Bob could win the CHSH game with probability
greater than
$\frac{1}{2} + \frac{1}{\sqrt{6}} \approx 0.908$ then communication complexity would become trivial.  Thus, the axiom ``communication complexity is nontrivial'' in some sense explains why $\omega_Q(CHSH) < 0.908$.  
Other works have extended the set of nonlocal correlations known to collapse communication complexity~\cite{FWW09,distillMaxPRbox,hoyer2010optimal}. 

However, so far the axiom ``communication complexity is nontrivial'' had not pinned down the exact quantum value for any nonlocal game. For example, in the CHSH game, there is a gap between the threshold of approximately $0.908$ that Brassard \etal obtain and the true quantum value $\omega_Q(CHSH) \approx 0.856$.

In this paper, we address this question: can the axiom ``communication complexity is nontrivial'' be used to explain the quantum value of certain nonlocal games?
Along the way, we formalize a connection to 
the theory of reliable computation for (classical) circuits with noisy gates, and our results for nonlocal games correspond to new results for reliable classical computation.
We outline our contributions in both areas below.

\subsection{Contributions}\label{sec:contributions}
First, we address the extent to which the axiom ``communication complexity is nontrivial'' can explain the quantum value of nonlocal games.
\begin{itemize}
\item[(1)] We exhibit a nonlocal game $G$, for which 
\[ \omega_C(G) < \omega_Q(G) < \omega_{NS}(G), \]
and for which the axiom ``communication complexity is not trivial'' precisely pins down the value $\omega_Q(G)$.  
Our game $G$ is fundamental, in the sense that if communication complexity is trivial in any superquantum theory $S$, 
then there is (a version of) our game $G$ so that $\omega_S(G) > \omega_Q(G)$.  That is, a superquantum advantage at the game $G$ makes communication complexity trivial, and meanwhile any universe in which communication complexity is trivial offers a superquantum advantage at the game~$G$.

\item[(2)] We provide evidence that the axiom ``communication complexity is nontrivial'' is in fact \em not \em sufficient to pin down the quantum value of the CHSH game itself.  
In more detail, in \cite{Brassard}, Brassard \etal essentially use the ability to succeed at the CHSH game as a noisy $\AND$ gate.  They show that reliable computation is possible when circuits are built from these noisy $\AND$ gates along with noiseless $\XOR$ gates (which correspond to certain local operations for Alice and Bob).  This leads to protocols that collapse communication complexity.
We derive an upper bound on the noise threshold for reliable computation by {\it formulas} of noisy $\AND$ and noisy $\XOR$ gates.
Assuming a Conjecture~\ref{conj:AppliesToCircuits}, we are able to use this result about reliable computation to show that the strategy of \cite{Brassard} cannot be pursued further: the threshold of $0.908$ is tight for this model of computation.  While this result is only a barrier against one line of attack, it does suggest that the axiom ``communication complexity is nontrivial'' may not suffice to explain $\omega_Q(CHSH)$.
\end{itemize}
As alluded to in our contribution (2) above, there is a connection to reliable computation with noisy gates.  In that area, we make the following contributions.

\begin{itemize}
\item[(3)] Our contribution (2) above can be seen as a result about reliable computation.
Consider the following circuit model with noisy gates.
Let $\land_\eps$ denote a $2$-input $\AND$ gate which produces an incorrect answer with probability $\eps$, and let $\oplus_\tau$ denote a $2$-input $\XOR$ gate which produces an incorrect answer with probability $\tau$.\footnote{Here and in the rest of the paper, for a gate $g$, $g_\eps$ refers to a version of $g$ which fails with probability $\eps$.
Let $\mathcal{F}_{\eps, \tau}$ be the collection of formulas\footnote{A \em formula \em is a circuit where every gate has fan-out $1$ (that is, the graph underlying the circuit is a tree and each input variable may appear at one or more leaves of this tree).} defined on the gate set $\{\land_\eps, \oplus_\tau\}$, where the noise in each $\land_\eps$ and $\oplus_\tau$ gate is independent. Analogously, let $\mathcal{C}_{\eps, \tau}$ be the collection of general circuits defined on the same gate set, and note that $\mathcal{C}_{\eps, \tau} \supset \mathcal{F}_{\eps, \tau}$.}

Our main technical result is that for all $\tau>0$, for all $\eps > 1/6$, reliable computation is impossible in $\mathcal{F}_{\eps, \tau}$.  Note that for $\eps < 1/6$, it is possible to compute any function using a circuit in $\mathcal{C}_{\epsilon, 0}$ with error probability bounded away from $1/2$.  On the other hand, for any $\tau > 0$, for any $\eps \geq 1/6$, there is some function for which this task is impossible with formulas in $\mathcal{F}_{\eps, \tau}$. We make a conjecture (Conjecture~\ref{conj:AppliesToCircuits}) that our upper bound applies to $\mathcal{C}_{\eps, \tau}$ as well, then show that a topological result (Theorem~\ref{thm:noiseThresholdsAreStrict}) can be used to extend the bound to the case of noise-free $\XOR$ gates applicable to the construction of \cite{Brassard}.

There has been a great deal of work on pinning down noise thresholds for reliable computation, which we survey in Section~\ref{sec:related}.  However, most 
prior work has focused on \em symmetric noise, \em where the noise rate is the same across all gate types.
As we discuss below in Section~\ref{sec:overview}, extending these results to asymmetric noise---and in particular to include noiseless gates---raises several challenges relative to previous work.
Figure~\ref{fig:params} depicts how our work fits into existing work, which is summarized in Section~\ref{sec:related}.

Beyond our primary motivation in quantum mechanics, we believe that the case of asymmetric gate noise is independently interesting from the perspective of fault-tolerant computation.
We hope that our techniques and results may spur future research in this direction. 

\item[(4)] We formalize an equivalence between reliable computation by circuits of noisy gates and \em amplification. \em   Informally, an amplifier is a function $f:\{0,1\}^d \to \{0,1\}$ so that when $f$ is fed in random bits $x \in \{0,1\}^d$ with a slight bias away from $1/2$, the output $f(x)$ amplifies that bias.  While a relationship between reliable computation and amplification had been present in prior work, nailing down an equivalence is a bit subtle, and requires considering the \em convex hull \em of circuit classes; to the best of our knowledge ours is the first work to do this.

Our equivalence between reliable computation and amplification is required in conjunction with Conjecture~\ref{conj:AppliesToCircuits} to establish the threshold in our contribution (3) above.  Further, it leads to the definition and analysis of  
our game $G$ from contribution (1) whose quantum value is pinned down by the nontriviality of communication complexity.  

\end{itemize}

\subsection{Organization}
In Section~\ref{sec:overview}, we state our results in more detail, and give an overview of our proof techniques.
In Section~\ref{sec:related}, we survey related work.
In Section~\ref{sec:preliminaries}, we state some additional formal definitions that we need for our proofs.

In the Sections~\ref{sec:proofOneSixthFormulas}--\ref{sec:newNLGame}, we prove our results.  Because the quantum results build on our results in reliable computation, we begin with those. 
In Section~\ref{sec:proofOneSixthFormulas}, we prove Lemma~\ref{claim:mainClaim}, which upper bounds the threshold for reliable computation in the class $\mathcal{F}_{\eps, \tau}$. We also make Conjecture~\ref{conj:AppliesToCircuits}, and show that it implies Theorem~\ref{thm:mainOneSixthFormulas}, which gives a sharp noise threshold for reliable computation in the class $\mathcal{C}_{\eps, 0}$.

Our proof of Theorem~\ref{thm:mainOneSixthFormulas} relies on Theorem~\ref{thm:noiseThresholdsAreStrict}, which states that for any class of noisy circuits, the region where reliable computation is impossible is closed; we use Theorem~\ref{thm:noiseThresholdsAreStrict} as a black box in our proof of Theorem~\ref{thm:mainOneSixthFormulas} and return to it later.
In Section~\ref{sec:pfEquiv}, we prove Theorem~\ref{thm:AmplifierEquivToFTCC_andOverhead}, which shows a formal equivalence between reliable computation by circuits of noisy gates and amplification.  
In Section~\ref{section:proofNoiseThresholdsStrict}, we use Theorem~\ref{thm:AmplifierEquivToFTCC_andOverhead} to prove Theorem~\ref{thm:noiseThresholdsAreStrict}, which along with Conjecture~\ref{conj:AppliesToCircuits} allows us to prove Theorem~\ref{thm:mainOneSixthFormulas} from Lemma~\ref{claim:mainClaim}.  
Finally, in Section~\ref{sec:newNLGame}, we prove Theorem~\ref{thm:newNonlocalGame}, which constructs the game $G_k$ so that $\omega_Q(G_k)$ is pinned down by the axiom ``communication complexity is nontrivial.''

We conclude in Section~\ref{sec:conclusion} with some discussion and future directions.

\section{Results and Technical Overview}\label{sec:overview}

In this section we state our results more precisely, and give a brief overview of how we achieve them. 

First, in Section~\ref{sec:overview_rel}, we introduce a few necessary definitions and discuss the relationship between nonlocality and reliable computation.
Then we discuss each of the contributions from Section~\ref{sec:contributions} in more detail.
We discuss our results in quantum nonlocality in Sections~\ref{sec:overview_magicG} and \ref{sec:overview_CHSH}, explaining how they would follow from our results on reliable computation.
Then we formally state our results on reliable computation in Sections~\ref{sec:overview_thresh} and \ref{sec:overview_equiv}, respectively, and give a high-level overview of our proof techniques.  

We note that our results and techniques in Sections~\ref{sec:overview_thresh} and \ref{sec:overview_equiv} are purely classical, and can be read without any background in quantum mechanics.  In particular, the reader interested only in our results in classical reliable computation can skip to Section~\ref{sec:overview_thresh}.

\subsection{Relationship between nontrivial communication complexity and reliable computation}\label{sec:overview_rel}

In our study of nonlocal games, we will consider players Alice and Bob who have joint access to different sets of \em bipartite correlations \em (Definition~\ref{define:NSCor}).  A bipartite correlation can be thought of as a box 
that a spatially separated Alice and Bob can use to process distributed inputs, without providing them the ability to communicate.
Alice inputs $x$, Bob inputs $y$, and the box outputs $a$ for Alice and $b$ for Bob according to some distribution $\mathbb{P}[a,b | x,y]$.  Sets of interest include $C$, the set of all bipartite correlations that are possible classically; $Q$, the set of all bipartite correlations that are possible if Alice and Bob share quantum entanglement; and $NS$, the set of all nonsignalling bipartite correlations.

In this paper we will consider sets $S$ of bipartite correlations that are closed under all of the operations that Alice and Bob might want to do to combine elements of $S$ with each other (for example, composing or taking probabilistic mixtures of correlations); following \cite{closedCorrelationSets}, we say that such sets are \em closed under wirings. \em  The sets $C, Q$, and $NS$ are all closed under wirings.
A set $S$ of bipartite correlations naturally gives rise to a circuit class by thinking about how these correlations act on \em distributed bits.  \em  We say that a bit $z$ is distributed as $z = x \oplus y$ if Alice holds $x$ and Bob holds $y$, where $x$ is uniformly random and $y = z \oplus x$. 
Then we can think of a bipartite correlation acting on inputs $x$ and $y$ as a gate acting on input $z$.

 In more detail, suppose that $c \in S$ is a bipartite correlation that stochastically maps inputs $\mathbf{x}, \mathbf{y} \in \{0,1\}^t$ for Alice and Bob respectively to bits $a,b \in \{0,1\}$.  We can define a randomized gate $\BBLMTU(c): \{0,1\}^t \to \{0,1\}$ as follows.
The gate $\BBLMTU(c)$ takes as input $\mathbf{z} \in \{0,1\}^t$.  Each coordinate $z_i$ of $\mathbf{z}$ is distributed between Alice and Bob as $z_i = x_i \oplus y_i$.  Then $\BBLMTU(c)$ outputs $a \oplus b$, where $a,b$ are the output of $c$ acting on $\mathbf{x}$ and $\mathbf{y}$.
This process is depicted in Figure~\ref{fig:nonlocalcircuit}. 

\begin{figure}
\centering
\begin{tikzpicture}[xscale=1.3,yscale=.9]
\begin{scope}
    \clip (-4,0) rectangle (4,4);
    \draw[thick] (0,0) circle(4);
    \draw[thick] (-4,0) -- (4,0);
\end{scope}
\node (z1) at (-2,-.75) {$z_1$};
\node (z2) at (2,-.75) {$z_2$};

\coordinate (zz1) at (-2,0);
\coordinate (zz2) at (2,0);

\draw (z1) to (zz1);
\draw (z2) to (zz2);

\draw[blue] (-2,2) rectangle (-1,3);
\node[blue] at (-1.5,2.5) {Alice};
\node[blue] at (1.5,2.5) {Bob};
\draw[blue] (1,2) rectangle (2,3);
\draw[blue,decorate,decoration=snake] (-1,2.5) -- (1,2.5);

\node(x1) at (-1.75,1.25) {$x_1$};
\node(x2) at (-1.25,1.25) {$x_2$};
\node(y1) at (1.25,1.25) {$y_1$};
\node(y2) at (1.75,1.25) {$y_2$};

\draw[blue] (x1) to (-1.75,2);
\draw[blue] (y1) to (1.25,2);
\draw[blue] (x2) to (-1.25,2);
\draw[blue] (y2) to (1.75,2);

\node(a) at (-1.5, 3.5) {$a$};
\node(b) at (1.5, 3.5) {$b$};

\node(out) at (0,4.75) {$\BBLMTU(c)(z_1,z_2)$};
\draw (out) to (0,4);
\node(oplus) at (0,3.75) {$\oplus$};
\draw[dashed] (0,4) to (oplus);

\draw[dashed] (a) to (oplus);
\draw[dashed] (b) to (oplus);

\draw[blue] (a) to (-1.5,3);
\draw[blue] (b) to (1.5,3);

\draw[dashed] (zz1) to (x1);
\draw[dashed] (zz1) to (y1);
\draw[dashed] (zz2) to (x2);
\draw[dashed] (zz2) to (y2);

\end{tikzpicture}
\caption{Defining a randomized circuit $\BBLMTU(c)$ from a nonlocal correlation $c \in S$.}\label{fig:nonlocalcircuit}
\end{figure}
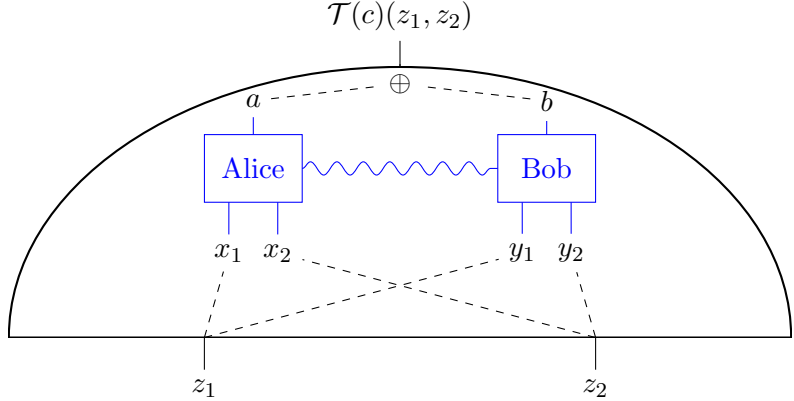

Given a convex set $S$ of bipartite correlations that is closed under wirings, we can define\footnote{Technically, Definition~\ref{define:BBLMTU} takes the convex hull of circuits comprised of $\{\BBLMTU(c)\,:\,c \in S\}$ gates; if $S$ is closed under wirings then this distinction does not matter, as per Proposition~\ref{proposition:closedUnderWiringsImpliesBBLMTUisElementwise}.} 
the set $\BBLMTU(S)$ to be the set of circuits one can make out of the gates $\{ \BBLMTU(c) \,:\, c \in S \}$
(see Definition~\ref{define:BBLMTU}).

Our goal is to understand which sets $S$ of bipartite correlations cause communication complexity to be trivial.  
We will do so by studying when the (noisy) circuit model $\BBLMTU(S)$ supports reliable computation.
\begin{define}\label{define:reliableComputationDefinition}
A (noisy) circuit model $\mathcal{C}$  {\bf supports reliable computation with advantage $\delta_0 > 0$} if for all $n > 0$, for all Boolean functions $f: \F_2^{n}\rightarrow \F_2$, there exists a circuit $c \in \mathcal{C}$ such that for each possible input $\mathbf{x} \in \F_2^n$,
\begin{equation}\label{eq:advantage}
(-1)^{f(\mathbf{x})} \left( \mathbb{P}[c(\mathbf{x}) = 0] - \mathbb{P}[c(\mathbf{x}) = 1] \right) \geq \delta_0,
\end{equation}
where the probability is over the randomness in $c$.

We say that $\mathcal{C}$ {\bf supports reliable computation} if there exists a $\delta_0 > 0$ so that $\mathcal{C}$ supports reliable computation with advantage $\delta_0$.
\end{define}

We will say that a set $S$ of bipartite correlations causes communication complexity to become trivial if there is some way for Alice and Bob to use the correlations in $S$, along with shared randomness and arbitrary local computation, to compute any function with high probability with constant communication complexity.  That is, there are some constants $\eps > 0$ and $t \geq 1$ so that for any $n$ and for any $f:\{0,1\}^n \times \{0,1\}^n \to \{0,1\}$, the following holds.  There is some protocol $\Pi_S$ for Alice and Bob so that $\Pi_S$ uses $t$ bits of communication (in either direction), and so that for any inputs $\mathbf{x},\mathbf{y} \in \{0,1\}^n$ for Alice and Bob respectively,
\[ \mathbb{P}[ \Pi_S(\mathbf{x}, \mathbf{y}) = f(\mathbf{x}, \mathbf{y}) ] \geq 1/2 + \eps. \]

Our starting point is the observation, implicit in \cite{Brassard}, that if $\BBLMTU(S)$ supports reliable computation, then $S$ causes communication complexity to become trivial.  In fact, it is not hard to see that the converse is true as well.  Thus, we have the following proposition.
\begin{proposition}\label{prop:thepoint}
Suppose that $C \subseteq S \subseteq NS$ and that $S$ is closed under wirings.  Then
 $S$ causes probabilistic communication complexity to become trivial (in the sense described above) if and only if $\BBLMTU(S)$ supports reliable computation. 
\end{proposition}
The proof of Proposition~\ref{prop:thepoint} follows similar logic to \cite{Brassard}.  For completness, we include a proof in Appendix~\ref{app:qm}.  The basic idea behind the connection is as follows.  If $\BBLMTU(S)$ supports reliable computation, then in particular $\BBLMTU(S)$ contains a circuit $\mathrm{Amp}: \{0,1\}^t \to \{0,1\}$ that acts as an \em amplifier \em (Definition~\ref{define:amplifier}).  That is, given independent random bits $z_1, \ldots, z_t$ with bias $p>1/2$ (resp. $p<1/2$), $\mathrm{Amp}(z_1, \ldots, z_t)$ outputs a bit $a$ that is very likely to be $1$ (resp. $0$).\footnote{Indeed, $\mathrm{Amp}$ is simply the circuit that implements the Majority function.}  Suppose that Alice and Bob want to compute some function $f:\{0,1\}^n \times \{0,1\}^n \to \{0,1\}$.  It turns out that Alice and Bob can, using only classical techniques and without communicating, obtain bits $x_1, \ldots, x_t$ and $y_1, \ldots, y_t$ so that for each $i$, $x_i \oplus y_i$ has a very slight bias towards the correct answer.  However, this bias shrinks as $n$ grows.  To amplify their success so that this bias is a constant, Alice and Bob use $S$, as shown in Figure~\ref{fig:nonlocalcircuit}, to obtain bits $a$ and $b$ respectively so that $a \oplus b = \mathrm{Amp}( x_1 \oplus y_1, \ldots, x_t \oplus y_t )$.  Then $a \oplus b$  is very likely to be equal to the correct value of $f$.    Finally, Alice sends the single bit $a$ to Bob, who outputs $a \oplus b$.

Due to Proposition~\ref{prop:thepoint}, for the rest of the paper we will in fact take ``$\BBLMTU(S)$ supports reliable computation'' as the \em definition \em of trivial probabilistic communication complexity (Definition~\ref{define:TPCC}).

With the connection between trivial communication complexity and reliable computation established, we continue with an overview of our main results in both nonlocality and in reliable communication.

\subsection{A nonlocal game whose quantum value is the threshold for nontrivial communication complexity}\label{sec:overview_magicG}

Our main result in this section is the following.

\begin{thm}[A game whose quantum value is the threshold for nontrivial communication complexity]\label{thm:newNonlocalGame}
There exists a sequence of 2-player nonlocal games $G_k$ for $k \geq 1$ that satisfies properties (1-3) below, in which $S$ is any set of bipartite nonsignalling correlations closed under wirings and such that $S\supseteq Q$.
\begin{enumerate}
\item For all $k\geq 1$, $\omega_C(G_k) < \omega_Q(G_k) < 1.$
\item Fix any $k\geq 1$. If $\omega_S(G_k) > \omega_Q(G_k)$, then $S$ has trivial probabilistic communication complexity.
\item If $S$ has trivial probabilistic communication complexity, then there exists some $k\geq 1$ such that $\omega_S(G_k) > \omega_Q(G_k)$.
\end{enumerate}

\end{thm}

The proof of Theorem~\ref{thm:newNonlocalGame} is in Section~\ref{sec:newNLGame}, and we sketch the intuition below.
In Section~\ref{sec:overview_equiv} below, we state
Theorem~\ref{thm:AmplifierEquivToFTCC_andOverhead}, which roughly says that reliable computation is equivalent to containing an amplifier.  Inspired by this, we define the \em amplification game \em $\mathrm{Amp}_k$ as follows.  Alice and Bob get $\mathbf{x}, \mathbf{y} \in \{0,1\}^{2k+1}$ respectively, and their goal is to output $a,b$ so that
\[ a \oplus b = \mathrm{Maj}( \mathbf{x} \oplus \mathbf{y} ), \]
where $\oplus$ is applied coordinate-wise.  The inputs $\mathbf{x}$ and $\mathbf{y}$ are drawn from a distribution so that success at the amplification game using the correlations $S$ translates into an amplifier in $\BBLMTU(S)$.  Since, by Theorem~\ref{thm:AmplifierEquivToFTCC_andOverhead}, amplification is equivalent to reliable computation, this translates to reliable computation for $\BBLMTU(S)$, which in turn, via Proposition~\ref{prop:thepoint}, translates into trivial probabilistic communication complexity.   Formalizing these connections imply that the family $\mathrm{Amp}_k$ satisfies properties 2 and 3 of Theorem~\ref{thm:newNonlocalGame}.

However, it turns out that property 1 of Theorem~\ref{thm:newNonlocalGame} is not satisfied: $\omega_C(\mathrm{Amp}_k) = \omega_Q(\mathrm{Amp}_k)$.   This is disappointing if the goal is to use the axiom ``communication complexity is nontrivial'' to pin down $\omega_Q(\mathrm{Amp}_k)$, because it also pins down $\omega_C(\mathrm{Amp}_k)$.    To obtain our game $G_k$ as in Theorem~\ref{thm:newNonlocalGame}, we use $\mathrm{Amp}_k$ along with the \em Mermin-Peres magic square game \em \cite{merminMagicSquare, peresMagicSquare} in order to make a game which retains properties 2 and 3, but which also has a gap between $\omega_C(G_k)$ and $\omega_Q(G_k)$.

\subsection{The approach of Brassard \etal cannot be improved}\label{sec:overview_CHSH}

Our next result is that the approach of Brassard \etal in \cite{Brassard} cannot be improved.  Recall from the introduction that \cite{Brassard} shows that, if $\omega_S(CHSH) > 0.908$ for some set $S$ of bipartite correlations, then communication complexity is trivial in any world where $S$ is allowed.  The hope would be to extend this result to replace $0.908$ with $\omega_Q(CHSH) \approx 0.854$. If this were the case, then the axiom ``communication complexity is nontrivial'' could pin down the quantum value of the CHSH game.

Unfortunately, we provide evidence that the approach of \cite{Brassard} cannot be improved.  As per Proposition~\ref{prop:thepoint}, \cite{Brassard} show that $\BBLMTU(S)$ supports reliable computation for any $S$ that allows Alice and Bob to win the CHSH game with probability greater than $0.908$.  Their approach is to show that for any such $S$, $\BBLMTU(S)$ contains the gates $\{\land_\eps, \oplus_0\}$ for $\eps < 1/6$.  Then they show how to build an amplifier as a formula on these gates.

The hope to improve the result of \cite{Brassard}---to replace the threshold $0.908$ with a smaller number---was to make an amplifier out of $\{\land_\eps, \oplus_0\}$ for $\eps \geq 1/6$.  However, our main technical result, Theorem~\ref{thm:main}, shows that assuming a certain conjecture, this task is impossible.  Specifically, we show that the class $\mathcal{F}_{\eps,\tau}$ of formulas on $\{\land_\eps, \oplus_\tau\}$ gates does not support reliable computation for any $\eps > 1/6, \tau>0$, and conjecture (Conjecture~\ref{conj:AppliesToCircuits}) that this bound applies to circuits as well. In particular (using Theorem~\ref{thm:AmplifierEquivToFTCC_andOverhead} about the equivalence between amplification and reliable computation), this would imply $\mathcal{C}_{\eps, \tau}$ does not contain an amplifier, and due to Theorem~\ref{thm:noiseThresholdsAreStrict}, these thresholds can be sharpened to include $\eps \geq 1/6, \tau \geq 0$. 

Theorem~\ref{thm:main}, assuming Conjecture~\ref{conj:AppliesToCircuits},  rules out the approach of \cite{Brassard}, but there are still two avenues open.  First, Conjecture~\ref{conj:AppliesToCircuits} may be false.  Second, one could hope to use the more of the class $\BBLMTU(S)$ than just $\{\land_\eps, \oplus_0\}$ gates.  However, there are reasons to be pessimistic about both of these avenues. 
First, Conjecture~\ref{conj:AppliesToCircuits} is directly analogous to conjectures made by authors who have used the same general technique to upper bound the threshold for formulas \cite{maxGateNoise,pippengerFirstUpperBound,FalkUngerMinorBoundImprovements}. To the best of our knowledge, there are no classes of gates known for which the true threshold for circuits is known to lie above the bound that this technique gives for formulas.
Second, although we cannot currently rule it out, we would find it surprising if there were a more efficient way of using the ability to succeed at the CHSH game than to create noisy AND gates.
As one example of work in this direction, a nontrivial ``adaptive" protocol was introduced by \cite{pawlowski2009information} and used to show that the ability to win the CHSH game better than quantum mechanics violates a principle they termed {\it information causality}. This same protocol was later employed by  \cite{mori2016three}, who applied it to violating non-trivial communication complexity but could not improve on the threshold value found by \cite{Brassard}.
Thus, our results suggest that the axiom ``communication complexity is nontrivial'' may not pin down $\omega_Q(CHSH)$.

\subsection{Sharp thresholds for reliable computation in $\cC_\eps$}\label{sec:overview_thresh}
Having explained the implications of our results on reliable computation for nonlocality, we now explain these results themselves.  We begin with our main technical result, which is that the noise threshold for reliable computation using formulas on $\{\land_\eps, \oplus_0\}$ is $\eps = 1/6$.

In fact, we show something stronger,
in that we allow \em probabilistic mixtures \em of formulas.
That is, for a class of probabilistic circuits $\mathcal{C}$, we define $\conv \mathcal{C}$ to be the set of probabilistic circuits obtained as distributions on elements of $\mathcal{C}$.  
With this notation, our main theorem in this section is as follows.

\begin{thm}[Sharp threshold for reliable computation]\label{thm:mainOneSixthFormulas}\label{thm:main}
Let $\eps \in [1/6,5/6]$.
Let $\mathcal{C}_\eps$ be the class of circuits on $\{\land_\eps, \oplus_0\}$.
Assuming Conjecture~\ref{conj:AppliesToCircuits}, $\conv \mathcal{C}_\eps$ does not support reliable computation.
\end{thm}

The work of \cite{Brassard} implies that $\mathcal{C}_\eps$ supports reliable computation for all $\epsilon < 1/6$.  Thus, Theorem~\ref{thm:main} is tight.
The proof of Theorem~\ref{thm:mainOneSixthFormulas} is given in Section~\ref{sec:proofOneSixthFormulas}.
\begin{remark}[NOT gates]\label{not1}
A noise-free $\oplus_0$ gate may be used to construct a noise-free unary $\neg$ (NOT) gate, by setting one of the input wires to $1$.  Thus, $\mathcal{C}_{\eps}$ also includes $\neg_0$.

In fact, our entire proof (including Lemma~\ref{claim:mainClaim} below which does not include noiseless $\oplus_0$ gates) goes through in the presence of $\neg_0$ gates, and implies the slightly stronger statement that, defining the circuit model $\cC_{\eps, \tau, 0}$ of circuits from the gate set $\{\land_\eps, \oplus_\tau, \neg_0\}$, $\conv \cC_{\eps, \tau, 0}$ does not support reliable computation for any $\eps \in [1/6, 5/6]$, $\tau \in [0,1]$.  See Remark~\ref{rmk:NOTsComeForFree}.
\end{remark}

The reason we need to consider convex hulls is for the connection to nonlocality, described in Section~\ref{sec:overview_CHSH}.  
A noisy circuit corresponds to a strategy for the CHSH game, for which Alice and Bob are allowed shared randomness and hence can execute probabilistic mixtures of strategies.

The main ingredient in the proof of Theorem~\ref{thm:mainOneSixthFormulas} is the following lemma. 
\begin{lem}\label{claim:mainClaim}
Let $\mathcal{F}_{\eps, \tau}$ be the class of formulas on $\{\land_\eps, \oplus_\tau\}$, and suppose that $\eps \in (1/6, 5/6)$, and $\tau \in (0, 1)$.
Fix $\Delta >0$ and let $f: \{0,1\}^n \to \{0,1\}$ be a function that is computable with probability at least $1/2 + \Delta$ by functions in $\conv \mathcal{F}_{\eps, \tau}$.
Then $f$ depends on at most a constant number of inputs.
\end{lem}
The proof of Lemma~\ref{claim:mainClaim} is given in Section~\ref{sec:proofOneSixthFormulas}.
Our proof may be viewed a probabilistic analogue of an argument first presented by Pippenger, which reduces the problem of formulas reliably computing functions that depend on many arguments to the problem of deep formulas computing a function of a single argument \cite{pippengerFirstUpperBound}.
The proof idea is as follows.  Let $f$ be some function.  We show that for any distribution on formulas $C \in \conv \mathcal{F}_{\eps,\tau}$, there is some variable $X_i$ that $f$ depends on, which appears reasonably deep, on average, in the formulas in the support of $C$.
This means that $X_i$ must pass through many noisy gates before reaching the output, which implies that $C$ cannot compute $f$ too accurately. 
While the basic idea is similar to the argument of \cite{pippengerFirstUpperBound}, since we consider distributions on formulas and also allow for arbitrarily small $\tau > 0$, new ideas are required to establish Lemma~\ref{claim:mainClaim}.

We make Conjecture~\ref{conj:AppliesToCircuits}, which states that Lemma~\ref{claim:mainClaim} applies to circuits as well as formulas.
Conjecture~\ref{conj:AppliesToCircuits} and Lemma~\ref{claim:mainClaim} come close to establishing Theorem~\ref{thm:mainOneSixthFormulas}. Indeed, since there are functions which depend on more than a constant number of inputs (for example, the AND of $n$ bits), Conjecture~\ref{conj:AppliesToCircuits} and Lemma~\ref{claim:mainClaim} imply that such functions cannot be computed in $\conv \mathcal{C}_{\eps,\tau}$ with any constant probability larger than $1/2$, provided that $\eps \in (1/6, 5/6)$ and $\tau \in (0,1)$.  The final step to the proof of Theorem~\ref{thm:mainOneSixthFormulas} is to handle the case of $(\eps, \tau) \in \{1/6, 5/6\}\times\{0,1\}$.  We do this by showing that the set of $(\eps, \tau)$ for which $\conv \mathcal{C}_{\eps, \tau}$ does \em not \em support reliable computation is closed.

\begin{remark}[The difference between Theorem~\ref{thm:main} and Lemma~\ref{claim:mainClaim}]
While Lemma~\ref{claim:mainClaim} and Conjecture~\ref{conj:AppliesToCircuits} are a key step to Theorem~\ref{thm:main}, there is still a big gap.
The proof of Lemma~\ref{claim:mainClaim} crucially relies on being able to take $\eps > 1/6$ and $\tau > 0$, where the inequalities are strict.  In particular, Lemma~\ref{claim:mainClaim} is clearly false if we take $\tau = 0$, as the example of the parity function shows.
Since $\tau=0$ is the setting we care about for our application, the step from Lemma~\ref{claim:mainClaim} to Theorem~\ref{thm:main} is important.
\end{remark}

Our proof that the ``non-reliable computation region'' is closed uses a
characterization---which may be of independent interest---of those circuit models $\mathcal{C}$ whose convex hulls $\conv \mathcal{C}$ support reliable computation.  More precisely, we formalize  
the relationship between { amplification} and classical fault-tolerant computation.  We discuss this formalization more in the next section.

\subsection{Equivalence between reliable computation and amplification}\label{sec:overview_equiv}
We say that a function $f:\{0,1\}^n \to \{0,1\}$ is an \em amplifier \em (Definition~\ref{define:amplifier}) if it amplifies the probability of a $1$ (resp. $0$) when given as input i.i.d. bits which are slightly biased towards $1$ (resp. $0$). 
The  relationship between amplifiers and reliable computation has been implicitly exploited in previous work.
However, making this relationship explicit (which turns out to be somewhat involved), is helpful in proving Theorem~\ref{thm:mainOneSixthFormulas}.  Moreover, it has applications to nonlocal games, as discussed in Sections~\ref{sec:overview_magicG} and \ref{sec:overview_CHSH}.  We hope that this formalization will be useful for other questions in fault-tolerant computation.

We establish the relationship between reliable computation and amplification with the following theorem.  Note that this theorem applies to arbitrary circuits, not just formulas. 

\begin{thm}[Equivalence between reliable computation and amplification]\label{thm:AmplifierEquivToFTCC_andOverhead}
Let $\mathcal{C}$ denote a circuit model closed under composition. Then $\conv\mathcal{C}$ supports reliable computation if and only if $\conv\mathcal{C}$ contains both an amplifier and a $\neg_\kappa$ gate for $\kappa < 1/2$.

Further,
given a circuit model $\mathcal{C}$ such that $\conv \mathcal{C}$ supports reliable computation, there exists a constant $s$ such that for any function $f:\{0,1\}^n \to \{0,1\}$ computable by a depth-$d$ circuit of noiseless $\NAND$ gates, $f$ can be computed by a depth-$(s\cdot d)$ circuit in $\conv \mathcal{C}$ with failure probability bounded away from $1/2$.
 \end{thm}
Theorem~\ref{thm:AmplifierEquivToFTCC_andOverhead} may be viewed as a generalization of a number of results of a similar flavor which have been proven, explicitly or implicitly, for specific circuit models \cite{Maj3VN,evansSchulmanNoisyCircuitBound,HW91, maximumTolerableNoiseKinputGates,maximumTolerableNoiseKinputGates,maxGateNoise,FalkUngerMinorBoundImprovements,Brassard}.
However, to the best of our knowledge no equivalence in this generality has been stated before; perhaps this is because previous work has not explicitly considered convex hulls of circuit models.

The depth statement in Theorem~\ref{thm:AmplifierEquivToFTCC_andOverhead} (the ``Further'' clause) is closely related to many classical results on the depth and size complexity overhead for fault-tolerance, for example \cite{pippenger1985networks,dobrushin1977upper}.  Our result in Theorem~\ref{thm:AmplifierEquivToFTCC_andOverhead} differs from previous work in that it holds whenever reliable computation is possible (as opposed to for some fixed noise level or gate set). 
To the best of our knowledge, the depth statement in Theorem~\ref{thm:AmplifierEquivToFTCC_andOverhead} is not an immediate consequence of prior work.

One direction of Theorem~\ref{thm:AmplifierEquivToFTCC_andOverhead} is straightforward. Supposing $\conv \mathcal{C}$ supports reliable computation, it can reliably compute the majority function on $k$ variables $\Maj^{(k)}$; for sufficiently large $k$ this leads to an amplifier in $\conv \mathcal{C}$. 

The other direction is more involved.  If $\conv \mathcal{C}$ contains an amplifier and a $\neg_\kappa$ gate for $\kappa < 1/2$, then one can construct a map that behaves similarly to a $\NAND$ gate, and use this map to reliably compute an arbitrary Boolean function. Since the depth of this approximate $\NAND$ gate is constant (for fixed $\kappa$), this means that the depth of any fault-tolerant formula is only a constant factor larger than the depth of the noise-free formula. We prove Theorem~\ref{thm:AmplifierEquivToFTCC_andOverhead} in Section~\ref{sec:pfEquiv}.

For our purposes, Theorem~\ref{thm:AmplifierEquivToFTCC_andOverhead} is useful for two reasons.
First, it is useful in formalizing the connection to nonlocal games, as discussed below.
Second, Theorem~\ref{thm:AmplifierEquivToFTCC_andOverhead}
makes it easy to prove the following Theorem~\ref{thm:noiseThresholdsAreStrict}, which provides the last part of the proof of Theorem~\ref{thm:mainOneSixthFormulas} that we outlined above in Section~\ref{sec:overview_thresh}.
\begin{thm}\label{thm:noiseThresholdsAreStrict}
Let $\mathcal{C}_\epsilon$ denote a circuit model on a gate set $\mathcal{G}$ which includes a noisy gate $g_\eps$.
Let $I \subseteq [0, 1]$ denote the set of $\epsilon$ for which $\conv \mathcal{C}_\epsilon$ does not support reliable computation (varying the noise on $g_\eps$ and keeping all other gates in $\mathcal{G}$ fixed). Then $I$ is closed.
\end{thm}
Notice that Theorem~\ref{thm:noiseThresholdsAreStrict} directly implies that there exists a nonzero noise threshold for any circuit model that is closed under convex combinations and based on a functionally complete set of logic gates.

The proof of Theorem~\ref{thm:noiseThresholdsAreStrict} is a simple consequence of Theorem~\ref{thm:AmplifierEquivToFTCC_andOverhead}. Suppose that $\mathcal{C}_\eps = \conv \mathcal{C}_\eps$ supports reliable computation.  Then Theorem~\ref{thm:AmplifierEquivToFTCC_andOverhead} implies there is an amplifier and a $\neg_\kappa$ gate for $\kappa < 1/2$ in $\mathcal{C}_\eps$.  
Using some elementary analytical lemmas, one can show that these finite circuits retain their nature despite a sufficiently small ``nudge" in the noise rate $\eps$, and hence applying Theorem~\ref{thm:noiseThresholdsAreStrict} again, the resulting circuit model still supports reliable computation.
We prove Theorem~\ref{thm:noiseThresholdsAreStrict} in Section~\ref{section:proofNoiseThresholdsStrict}.

\section{Related Work}\label{sec:related}

In this section, we briefly review related work.  We begin with related work in quantum mechanics, and then discuss related work in classical fault-tolerant computation.

\subsection{Axiomatization of Quantum Mechanics}\label{ssection:relatedWorkQuantumPhysics}
Features of quantum mechanics like the uncertainty principle and quantum entanglement have perplexed scientists since its early days, ultimately requiring a wholesale reconsideration of information theory and the limits of computation. At the same time, quantum mechanics lacks the equivalent of the clear and concise physical principles from which Einstein derived special relativity. Instead, it is usually presented as a highly effective mathematical framework without prior or deeper justification. Given its radical implications for the definition and behavior of information, there have been several proposals for sets of information-theoretic axioms that can be used to derive quantum mechanics. Examples include those of Hardy~\cite{hardy2001quantum} as well as Mueller and Masanes~\cite{mueller2016information}. While those efforts are enlightening in many ways, they don't directly address one of the most profound features of quantum mechanics, quantum nonlocality. However, the concise and uncontroversial requirement that ``communication complexity is not trivial'' is known to place stringent constraints on that nonlocality, so it is intriguing to consider whether the requirement could function as an axiom precisely delineating the limits of quantum mechanics~\cite{implausibleConsequencesSuperstrongNonlocality, Brassard, distillMaxPRbox}.

The work of van Dam~\cite{implausibleConsequencesSuperstrongNonlocality} established that the ability to win the CHSH game with probability $1$ (that is, access to a so-called \em Popescu-Rohrlich \em (PR) box) causes communication complexity to be trivial.  As discussed above, the work of Brassard, Buhrman, Linden, M\'ethot, Tapp, and Unger~\cite{Brassard} extended this result to apply to success probability greater than $0.908$.  However, there is still a gap between this value and $\omega_Q(CHSH) \approx 0.8536$.

The work of Forster, Winkler and Wolf~\cite{FWW09} shows that certain superquantum correlations on the boundary of the NS polytope can be distilled into perfect PR boxes, and thus also collapse communication complexity.%
\footnote{
Distillation protocols are descriptions of wirings that combine multiple ``weaker'' nonlocal correlations such that the new correlation is more useful (e.g., for playing the CHSH game). 
}
Brunner and Skrzypczyk~\cite{distillMaxPRbox} considered adding noise to these correlations and extended this set of superquantum correlations that collapse communication complexity (which we will call the {\it distillable set}). Allcock, Brunner, Linden, Popescu, Skrzypczyk, and V\'ertesi~\cite{closedCorrelationSets} introduced the notion of a set of correlations remaining closed under wirings; they exhibited convex sets of correlations without this property.

The distillable set has points arbitrarily close to a vertex of the polytope of classical correlations $C$. Unfortunately this does not produce a nonlocal game whose quantum value is exactly limited by the requirement that communication complexity be nontrivial.
Geometrically, this is because supporting hyperplanes of $Q$ at points of intersection between $Q$ and the boundary of the distillable set are supporting hyperplanes of $NS$ itself. Therefore superquantum advantage at such games is not possible in any (possibly superquantum) theory.

Other works have computed the optimality of distillation protocols, shown impossibility results within restricted settings (e.g., for nonadaptive procedures), and exhibited closed sets of superquantum correlations~\cite{noDistill1, noDistill2, noDistill3,hoyer2010optimal, zoo}.\footnote{\cite{hoyer2010optimal} also extended the distillable set.} One might have hoped to improve the construction of Brassard \etal \cite{Brassard} by first distilling slightly superquantum noisy PR boxes to obtain better ones, which would then collapse communication complexity.  However, prior work has yet to discover a distillation procedure for noisy PR boxes, or to rigorously rule out that one exists.

Navascu\'es, Guryanova, Hoban, and Ac\'in~\cite{almostQuantum}  introduced the set $\tilde{Q}$ of ``almost-quantum" correlations, which strictly contains $Q$ and has nontrivial communication complexity. This result implies there are many superquantum correlations which do not collapse communication complexity, and that other principles beyond ``communication complexity is nontrivial'' are required to discriminate points between $Q$ and $NS\setminus Q$. On the other hand, since $\omega_{\tilde{Q}}(CHSH) = \omega_Q(CHSH)$, this left open the possibility that $\omega_Q(CHSH)$ is the maximum value consistent with nontrivial communication complexity.

\begin{figure}
\begin{center}
\definecolor{cbvermillion}{RGB}{213,94,0}
\definecolor{cbskyblue}{RGB}{86,180,233}
\definecolor{cbbluishgreen}{RGB}{0,158,115}
\definecolor{cborange}{RGB}{230,159,0}
\definecolor{cbyellow}{RGB}{240,228,66}

\def\legendx{-10*1/30}
\def\legendy{-10*3/4}
\def\legendh{.5}
\def\legendw{.7}
\def\axistickl{.2}
\def\thicklinewidth{3}

\begin{tikzpicture}[scale=1]

\draw [very thick] (0,0) -- (5+1/70,0);
\draw [very thick] (0,0) -- (0,5+1/70);

\draw [->,very thick] (0,0) -- (5+1/2,0);
\draw [->,very thick] (0,0) -- (0,5+1/2);

\draw [thick] (0.8856217223385232, 0) -- (0.8856217223385232, -\axistickl);
\node [below] at (0.8856217223385232, -\axistickl) {$\frac{3-\sqrt{7}}{4}$};
\draw [thick] (0, 0.8856217223385232) -- (-\axistickl, 0.8856217223385232);
\node [left] at (-\axistickl, 0.8856217223385232) {$\frac{3-\sqrt{7}}{4}$};

\draw [thick] (10*0.1464466094067262, 0) -- (10*0.1464466094067262, - 5*\axistickl);
\node [below] at (10*0.1464466094067262, -5*\axistickl) {$\frac{1}{2}-\frac{1}{\sqrt{8}}$};

\draw [thick] (0, 10*0.1464466094067262) -- (- 1*\axistickl, 10*0.1464466094067262);
\node [left] at (-1*\axistickl, 10*0.1464466094067262) {$\frac{1}{2}-\frac{1}{\sqrt{8}}$};

\draw [thick] (0, 0) -- (0, -\axistickl);
\node [below] at (0, -\axistickl) {$0$};

\draw [thick] (0, 0) -- (-\axistickl, 0);
\node [left] at (-\axistickl, 0) {$0$};

\draw [thick] (10/6, 0) -- (10/6, -\axistickl);
\node [below] at (10/6, -\axistickl) {$\frac{1}{6}$};

\draw [thick] (10/2, 0) -- (10/2, -\axistickl);
\node [below] at (10/2, -\axistickl) {$\frac{1}{2}$};
\draw [thick] (0, 10/2) -- (-\axistickl, 10/2);
\node [left] at (-\axistickl, 10/2) {$\frac{1}{2}$};

\draw [thick] (10/4, 0) -- (10/4, -\axistickl);
\node [below] at (10/4, -\axistickl) {$\frac{1}{4}$};
\draw [thick] (0, 10/4) -- (-\axistickl, 10/4);
\node [left] at (-\axistickl, 10/4) {$\frac{1}{4}$};

\node [below] at (4/4*1/2*10, -5*\axistickl) {$\epsilon$ (NAND noise)};
\node [left] at (-2*\axistickl, 3/4*1/2*10) {$\tau$ (XOR noise)};

\node at (\legendx+3, \legendy-0.8) {RC = ``reliable computation"};

\draw [cbvermillion, pattern=crosshatch dots, pattern color=cbvermillion] (0,0) -- (0,5) -- (0.8856217223385232,5) -- (0.8856217223385232,0) -- (0,0);
\draw [cbvermillion, pattern=crosshatch dots, pattern color=cbvermillion] (\legendx, \legendy + 7*1.3*\legendh) -- (\legendx, \legendy + \legendh + 7*1.3*\legendh) -- (\legendx + \legendw, \legendy + \legendh + 7*1.3*\legendh) -- (\legendx + \legendw, \legendy + 7*1.3*\legendh) -- (\legendx, \legendy + 7*1.3*\legendh);
\node [right] at (\legendx + \legendw, \legendy + 7*1.3*\legendh + \legendh/2) {RC by formulas, \cite{maxGateNoise}};

\draw [cbyellow, line width=\thicklinewidth, line cap=round] (0, 0) -- (10/6, 0);
\draw [cbyellow, fill=cbyellow] (\legendx, \legendy + 6*1.3*\legendh) -- (\legendx, \legendy + 6*1.3*\legendh + \legendh) -- (\legendx + \legendw, \legendy + 6*1.3*\legendh + \legendh) -- (\legendx + \legendw, \legendy + 6*1.3*\legendh) -- (\legendx, \legendy + 6*1.3*\legendh);
\node [right] at (\legendx + \legendw, \legendy + 6*1.3*\legendh + \legendh/2) {RC by circuits, \cite{Brassard}};

\begin{scope}[yshift=1.2cm]

\begin{scope}[yshift=-1.2cm]
\draw [cborange, pattern=north west lines, pattern color=cborange] (10/6, 0) -- (10/6, 10/2) -- (10/2, 10/2) -- (10/2, 0) -- (10/6, 0);
\end{scope}
\draw [cborange, pattern=north west lines, pattern color=cborange] (\legendx, \legendy + 1.3*\legendh) -- (\legendx, \legendy + 1.3*\legendh + \legendh) -- (\legendx + \legendw, \legendy + 1.3*\legendh + \legendh) -- (\legendx + \legendw, \legendy + 1.3*\legendh) -- (\legendx, \legendy + 1.3*\legendh);
\node [right] at (\legendx + \legendw, \legendy + 1.3*\legendh + \legendh/2) {no RC by formulas, this work};

\begin{scope}[yshift=-1.2cm]
\draw [gray, pattern=north east lines, pattern color=gray] (10/4, 0) -- (10/4, 10/2) -- (10/2, 10/2) -- (10/2, 0) -- (10/4, 0);
\end{scope}
\draw [gray, pattern=north east lines, pattern color=gray] (\legendx, \legendy + 3*1.3*\legendh) -- (\legendx, \legendy + 3*1.3*\legendh + \legendh) -- (\legendx + \legendw, \legendy + 3*1.3*\legendh + \legendh) -- (\legendx + \legendw, \legendy + 3*1.3*\legendh) -- (\legendx, \legendy + 3*1.3*\legendh);
\node [right] at (\legendx + \legendw, \legendy + \legendh/2 + 3*1.3*\legendh) {no RC by circuits $({\dag})$};

\begin{scope}[yshift=-1.2cm]
\draw [cbbluishgreen, line cap=round, line width=\thicklinewidth] (0.8856217223385232, 0.8856217223385232) -- (1/6*10, 1/6*10);
\end{scope}
\draw [cbbluishgreen, fill=cbbluishgreen] (\legendx, \legendy + 2*1.3*\legendh) -- (\legendx, \legendy + 2*1.3*\legendh + \legendh) -- (\legendx + \legendw, \legendy + 2*1.3*\legendh + \legendh) -- (\legendx + \legendw, \legendy + 2*1.3*\legendh) -- (\legendx, \legendy + 2*1.3*\legendh);
\node [right] at (\legendx + \legendw, \legendy + 2*1.3*\legendh + \legendh/2) {no RC by formulas, \cite{maxGateNoise,FalkUngerMinorBoundImprovements}};

\begin{scope}[yshift=-1.2cm]
\draw [cbskyblue, line width=\thicklinewidth, line cap=round] (10*0.1464466094067262, 10*0.1464466094067262) -- (10/2, 10/2);
\end{scope}
\draw [cbskyblue, fill=cbskyblue] (\legendx, \legendy + 0*1.3*\legendh) -- (\legendx, \legendy + 0*1.3*\legendh + \legendh) -- (\legendx + \legendw, \legendy + 0*1.3*\legendh + \legendh) -- (\legendx + \legendw, \legendy + 0*1.3*\legendh) -- (\legendx, \legendy + 0*1.3*\legendh);
\node [right] at (\legendx + \legendw, \legendy + 0*1.3*\legendh + \legendh/2) {no RC by circuits, \cite{evansSchulmanNoisyCircuitBound}};

\draw [decorate,decoration={brace,amplitude=5,raise=4}] (\legendx, \legendy + 0*1.3*\legendh) -- (\legendx, \legendy + 3*1.3*\legendh + \legendh) node [black,midway,left,xshift=-10] {negative results};

\end{scope}
\begin{scope}[yshift=-2.75cm]

\begin{scope}[yshift=2.75cm]
\draw [cbbluishgreen, dotted, line width=\thicklinewidth/2, rounded corners] (0.8856217223385232, 0.8856217223385232) -- (0.8856217223385232, 10/2) -- (10/2, 10/2) -- (10/2, 0.8856217223385232) -- (0.8856217223385232, 0.8856217223385232);
\end{scope}

\draw [cbbluishgreen, dotted, line width=\thicklinewidth/2, rounded corners] (\legendx, \legendy + 4*1.3*\legendh) -- (\legendx, \legendy + 4*1.3*\legendh + \legendh) -- (\legendx + \legendw, \legendy + 4*1.3*\legendh + \legendh) -- (\legendx + \legendw, \legendy + 4*1.3*\legendh) -- (\legendx, \legendy + 4*1.3*\legendh);
\node [right] at (\legendx + \legendw, \legendy + 4*1.3*\legendh + \legendh/2) {RC by formulas seems unlikely \cite{maxGateNoise,FalkUngerMinorBoundImprovements}};

\draw [decorate,decoration={brace,amplitude=5,raise=4}] (\legendx, \legendy + 4*1.3*\legendh) -- (\legendx, \legendy + 5*1.3*\legendh + \legendh) node [black,midway,left,xshift=-10] {likely negative results};

\begin{scope}[yshift=2.75cm]
\draw [cbskyblue, dotted, line width=\thicklinewidth/2, rounded corners] (10*0.1464466094067262, 10*0.1464466094067262) -- (10*0.1464466094067262, 10/2) -- (10/2, 10/2) -- (10/2, 10*0.1464466094067262) -- (10*0.1464466094067262, 10*0.1464466094067262);
\end{scope}
\draw [cbskyblue, dotted, line width=\thicklinewidth/2, rounded corners] (\legendx, \legendy + 5*1.3*\legendh) -- (\legendx, \legendy + 5*1.3*\legendh + \legendh) -- (\legendx + \legendw, \legendy + 5*1.3*\legendh + \legendh) -- (\legendx + \legendw, \legendy + 5*1.3*\legendh) -- (\legendx, \legendy + 5*1.3*\legendh);
\node [right] at (\legendx + \legendw, \legendy + 5*1.3*\legendh + \legendh/2) {RC by circuits seems unlikely \cite{evansSchulmanNoisyCircuitBound}};

\end{scope}

\draw [decorate,decoration={brace,amplitude=5,raise=4}] (\legendx, \legendy + 6*1.3*\legendh) -- (\legendx, \legendy + 7*1.3*\legendh + \legendh) node [black,midway,left,xshift=-10] {positive results};

\end{tikzpicture}
\caption{Map of $(\epsilon, \tau)$ parameter space indicating regions which have been shown to support reliable computation or not (positive and negative results respectively).
The negative result ${(\dag)}$ for $\epsilon\geq 1/4$ is a consequence of the fact that there exist functions for which communication complexity is nontrivial.  If one believes the (reasonable) hypothesis that increasing noise will not enable reliable computation, then impossibility results for $\eps = \tau = \alpha_0 < 1/2$ also imply impossibility results for $\alpha_0 \leq \eps,\tau \leq 1/2$, and these are depicted by the dashed boxes labeled ``likely negative results.''
Some prior work has focused on the functionally complete NAND gate in noisy circuit models, which is why the $\epsilon$ axis above corresponds to noise on the $\NAND$ gate. For our results, as explained in Remark~\ref{not1}, the distinction between $\AND$ and $\NAND$ does not matter.
}\label{fig:params}
\end{center}
\end{figure}

\subsection{Fault-tolerant Computation from Noisy Gates}
Fault-tolerant computation by circuits has been studied extensively since von Neumann's work in the 1950's.
A central question in this area is how noisy the gates can get before reliable computation is impossible.
In general, stronger bounds on the noise threshold have been obtained for formulas rather than general circuits; following this line of work, Theorem~\ref{thm:main} holds only for formulas, although we conjecture that a similar result holds for circuits as well.
Almost all work\footnote{We note that one exception to the
symmetric noise paradigm is~\cite{falkUngerBetterGatesBreakFaultTol} which shows that if an adversary gets to decrease the noise heterogeneously from gate to gate, fault tolerant computation actually becomes harder.}
 that we are aware of
 in fault-tolerant computation
focuses on symmetric noise. 

Modern work in  the symmetric case goes back to
the work of von Neumann in 1956,
who showed that reliable computation is possible using noisy 3-majority gates which fail independently with probability $\eps \leq 0.0073$~\cite{Maj3VN}.  Since then, there has been a great deal of work; we summarize the best results in this setting in Table~\ref{tab:litreview}.

\begin{table}[h!]
\centering
\tabulinesep=1mm
\begin{tabu}to\linewidth{|X[.9]|X[.7]|X[c]|X[c]|}
\hline
Noise Model & Source & Circuit model & Bounds on threshold $\eps_0$\\
\hline \hline
\multirow{3}{*}{Symmetric Noise}
& \cite{evansSchulmanNoisyCircuitBound}
	& All circuits of $\eps$-noisy gates of fan-in $k$
	& $\eps_0 \leq \frac{1}{2} - \frac{1}{2\sqrt{k}}$ \\
\cline{2-4}
& \cite{HW91, maximumTolerableNoiseKinputGates}
	& Formulas of $\eps$-noisy gates of odd fan-in $k$
	& $\eps_0 = \frac{1}{2} - \frac{2^{k-1}}{k {k-1 \choose k/2 - 1/2} }$ \\
\cline{2-4}
& \cite{maxGateNoise,FalkUngerMinorBoundImprovements}
	& Formulas of $\eps$-noisy gates of fan-in $2$
	& $\eps_0 = \frac{ 3 - \sqrt{7} }{4} \approx 0.08856$ \\
\hline
\multirow{2}{*}{Asymmetric Noise}
& \cite{Brassard}
	& Formulas of $\{\land_\epsilon, \oplus_0\}$ gates
	& $\eps_0 \geq 1/6$ \\
\cline{2-4}
& This work
	& Formulas of $\{\land_\epsilon, \oplus_\tau\}$ gates
	& $\forall \tau > 0, \eps_0 \leq 1/6$ \\
\hline
\end{tabu}
\vspace{.5cm}
\caption{Summary of best results on thresholds in both the symmetric and asymmetric case.  Above, $\eps_0$ represents the noise threshold so that if $\eps < \eps_0$ then reliable computation is possible, but if $\eps \geq \eps_0$ then it is impossible.
}
\label{tab:litreview}
\end{table}

To gain some intuition for these results, it is helpful to understand {amplification}, which we discussed in Sections~\ref{sec:intro} and \ref{sec:overview} and which we define formally in Section~\ref{sec:preliminaries}.
All of the positive results that we are aware of go through amplifiers.
That is, these works construct an amplifier out of the target gate set and then use that, perhaps along with other gates, to establish a method for reliable computation.
For example, von Neumann~\cite{Maj3VN} used a noisy 3-input majority gate $\Maj^{(3)}_\eps$ as an amplifier;
both Hajek and Weller~\cite{HW91} as well as Evans and Schulman~\cite{maximumTolerableNoiseKinputGates} also used $\Maj^{(k)}_\eps$ as an amplifier, and used noisy $\mathrm{XNAND}_\eps$ gates along with this amplifier to improve von Neumann's result to give a sharp threshold for $k$-input gates for odd $k$.
Evans and Pippenger~\cite{maxGateNoise} and Unger~\cite{FalkUngerMinorBoundImprovements} used the amplifier
\begin{equation}\label{eq:formulaNANDsquared}
\NAND_\eps(\NAND_\eps(X_0, X_1), \NAND_\eps(X_2, X_3)),
\end{equation}
along with more $\NAND_\eps$ gates to establish reliable computation for any $\eps < \eps_0 = \frac{ 3 - \sqrt{7 }}{4}$.  The work of \cite{maxGateNoise} showed a matching upper bound for reliable computation by formulas of noisy $\NAND_\eps$ gates, assuming noisy inputs, showing that reliable computation is impossible when $\eps > \eps_0$ under these assumptions.
Finally \cite{FalkUngerMinorBoundImprovements} extended the impossibility result to also include the case where $\eps = \eps_0$, removed the assumption that the inputs are noisy, and generalized the result to computation by the formulas of all $2$-input $\eps$-noisy gates.
As we explain further in Section~\ref{ssection:relatedWorkQuantumPhysics}, the limit on nonlocality from nontrivial communication complexity is derived in \cite{Brassard} using the following amplifier:
\begin{equation}\label{eq:brassardAmplifierFormula}
((X_0\oplus X_2)\land_\eps (X_0 \oplus X_1)) \oplus X_0
\end{equation}

Because all of the positive results go through amplifiers, it is natural to wonder whether there is a deeper connection between the amplifiers and reliable computation in a circuit model, and this is what we show in Theorem~\ref{thm:AmplifierEquivToFTCC_andOverhead}.
Although such a connection is implicit in prior work, to the best of our knowledge it has not been made rigorous.  This may be because the equivalence is easier to formulate and prove using $\conv\cC$ rather than $\cC$ itself.

\section{Formal Definitions}\label{sec:preliminaries}
In this section we formally define a few notions that we will need to prove our main results. 
\subsection{Fault-Tolerant Computation by Circuits of Noisy Gates}\label{ssection:ftccDefs}
\begin{define}[Formula]\label{define:formula}
Let $G = (V,E)$ be a directed tree, so that every node has out-degree at most $1$ and in-degree either $2$ or $0$.  
Let $\mathbf{X} = \{X_i \,:\, i \in [n]\}$ be a set of variables.
Let $\mathcal{G}$ be a set of (possibly noisy) two-input binary gates.
Let $L: V \to \{0,1\} \cup \mathbf{X} \cup \mathcal{G}$ be a labeling function so that $L(v) \in \mathbf{X} \cup \{0,1\}$ if and only if $v$ is a leaf.  Otherwise, $L(v) \in \mathcal{G}$.

A {\bf formula} on the gate set $\mathcal{G}$ acting on input variables $\mathbf{X}$ is given by a tuple $(G,L)$ for such a labeling function $L$.  We use $r(F)$ to denote the root vertex of $G$, which corresponds to the output gate of the formula $F$.

\end{define}

\begin{remark}[Restriction to two-input gates]
In this work we only consider gate sets $\mathcal{G}$ which contain only two-input gates, and so for convenience our definition of a formula reflects this restriction.  We note that a unary $\neg$ gate can be included in this definition by including appropriately modified versions of two-input gates in $\mathcal{G}$.
\end{remark}

Note that $\mathcal{G}$ may include noisy gates, in which case a formula built with gates from $\mathcal{G}$ induces a stochastic map:
a {\bf stochastic map} from $\{0,1\}^n$ to $\{0,1\}$ is a function $f: \{0,1\}^n \to [0,1]$, where we interpret $f(\mathbf{x}) = p$ as ``$f(\mathbf{x})$ outputs $1$ with probability $p$.'' For notational convenience, we will often identify a  formula $F$ with the stochastic map it induces.

\begin{define}[Mixture of Formulas]
A {\bf mixture of formulas} $C$ is a probability distribution over a $N$ formulas.  We may write the stochastic map of $C$ as
$$C = \sum_{j=1}^{N}{p_j F_j}$$
where $F_j$ is the stochastic map corresponding the the $j$'th formula, which is chosen with probability $p_j$.
\end{define}

\begin{define}[Depth]\label{def:depth}
Given a formula $F = (G, L)$ and a vertex $v \in V(G)$, the {\bf depth $d(v)$ of the vertex} $v$ is the length of the path from $v$ to the root of $G$. 
The {\bf depth $d(F)$ of the formula} $F$ is defined as $\max_{v \in V(G)} d(v)$.
For a  mixture of formulas $C = \sum_{i=1}^{N}{p_i F_i}$, we define the {\bf depth $d(C)$ of the mixture of formulas} $C$ as
$d(C) = \max_{i\in[N]}{d(F_i)}.$
\end{define}

\begin{define}[Dependence in Boolean functions]\label{define:boolDepends}
Given a Boolean function $f: \{0, 1\}^n \rightarrow \{0, 1\}$, we say that {\bf $f$ depends on $X_i$} if there exist constants $c_1, ..., c_{i-1}, c_{i+1}, ..., c_n \in \{0, 1\}$ such that
\begin{equation}\label{eq:dependsDef}
f(c_1, ..., c_{i-1}, 0, c_{i+1}..., c_n) \neq f(c_1, ..., c_{i-1}, 1, c_{i+1}..., c_n).
\end{equation}

Let $S \subseteq [n]$ be the largest set such that $f$ depends on $X_i$ for each $i \in S$.  If $|S| = k$, we say that $f$ {\bf depends on } $k$ inputs.
\end{define}

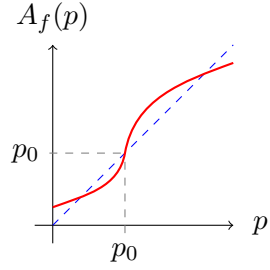
\begin{figure}
\centering
\begin{tikzpicture}[xscale=1.2, yscale=1.2]
\draw[->] (-.2, 0) to (2, 0) ;
\draw[->] (0,-.2) to (0, 2);
\node at (2.3, 0) {$p$};
\node at (0,2.3) {$A_f(p)$};
\node(l1) at (-.3, 0.8) {$p_0$};
\node(l2) at (0.8,-.3) {$p_0$};
\draw[gray, dashed] (l1) to (0.8, 0.8);
\draw[gray, dashed] (l2) to (0.8, 0.8);
\draw[dashed,blue] (0,0) to (2,2);
\draw[thick, red] (0,.2) to[out=20,in=260] (.8,.8) to[out=-180+260,in=200] (2,1.8);
\end{tikzpicture}
\caption{
    $f$ is an amplifier away from $p_0$.
}\label{fig:amp}
\end{figure}

\begin{define}[Amplifiers]\label{define:amplificationFunction}\label{define:amplifier}
The {\bf amplification function} of a (possibly stochastic) map $f: \{0,1\}^n \to \{0,1\}$ is the polynomial\footnote{Notice that $A_f(p) = \sum_{\mathbf{x} \in \{0,1\}^n} p^{|\mathbf{x}|}(1-p)^{n-|\mathbf{x}|} \mathbb{P}_f[f(\mathbf{x}) = 1]$ is a polynomial in $p$ of degree at most $n$.}
\begin{equation}\label{eq:AmplificationFunctionDefinition}
A_f(p) = \P_{X_i \sim \text{Ber}(p)}[f(X_1, ..., X_n) = 1]
\end{equation}
where the probability is over both the inputs $X_1,\ldots,X_n$, which are i.i.d. Bernoulli($p$) random variables, and the function $f$.
We say $f$ is an {\bf amplifier} if $\exists p_0 \in (0, 1)$ such that
$$A_f(p_0) = p_0 \qquad \text{and} \qquad A_f'(p_0) > 1.$$
In this case, we say that $f$ {\bf amplifies away from $p_0$}.
\end{define}

An amplification function $A_f$ for an amplifier $f$ is shown in Figure~\ref{fig:amp}.

\subsection{Bipartite Correlations and Nonlocal Games}\label{section:defsTwoParty}

As discussed in the introduction, we will allow Alice and Bob the ability to sample from \em nonsignalling correlations. \em 

\begin{define}\label{define:NSCor}
A {\bf bipartite nonsignalling correlation} (also called a box or strategy) $(A,B,X,Y,p)$ consists of finite input and output alphabets $A, B, X, Y$ along with a function $p: A\times B\times X\times Y \rightarrow [0,1]$ which defines a probability distribution over $a, b$ conditioned on $x,y$,
$$\P[a, b | x, y] = p(a, b, x, y),$$
and which satisfies the no-signalling condition:
$$\forall x \in X, \forall y, y' \in Y, \forall a \in A, \sum_{b\in B}{p(a,b, x,y)} = \sum_{b\in B}{p(a,b, x,y')}$$
$$\forall x, x' \in X, \forall y \in Y, \forall b \in B, \sum_{a\in A}{p(a,b, x,y)} = \sum_{a\in A}{p(a,b, x',y)}.$$

The set of all bipartite nonsignalling correlations is denoted {\bf NS}.
\end{define}
This no-signalling conditions enforce that neither party may alter the distribution observed by the other.

Two important subsets of $NS$ are $C$, the set of bipartite classical correlations, and $Q$, the set of bipartite quantum correlations.
\begin{define}
The set $C$ of bipartite classical correlations achievable with shared randomness is the set
$$C = \conv\left\{(A, B, X, Y, p) : p(a,b,x,y) = \begin{cases}
1&q(x)=a \text{ and } r(y) = b\\
0 & \text{otherwise}
\end{cases}, q: X\rightarrow A, r: Y\rightarrow B \right\}.$$
The set $Q$ of bipartite quantum correlations is the set 
$$Q = \bigcup_{n, k} \left\{(p(a,b,x,y)) : p(a,b,x,y) = \braket{\psi |A^x_a \otimes B^y_b|\psi}, \ket{\psi}\in \mathcal{H}_A\otimes \mathcal{H}_B, \forall xy, \{A^x_a\}_a, \{B^y_b\}_b \text{ POVM} \right\}.$$
\end{define}
Above, POVM stands for ``Positive Operator-Valued Measure.'' POVMs represent general quantum measurements. In our paper, we will not use make use of the technical definition of $Q$ directly, but it is worth mentioning that this definition corresponds to what have been called the {\it quantum spatial correlations}, in contrast to the more general {\it commuting-operator model}. In fact, our results hold in either model because communication complexity is nontrivial in both.
We note that $C$ includes correlations that give Alice and Bob access to shared randomness. 

We have the inclusions $C \subsetneq Q \subsetneq NS$, and for fixed finite $A, B, X, Y$, the sets $NS$ and $C$ are polytopes \cite{nonlocalCorrelationsCharacterization}.

Given a set $C \subseteq S \subseteq NS$ of bipartite correlations, Alice and Bob may use correlations $c \in S$ as steps in a larger computation.  In particular, they may compose different correlations in $S$ along with local computations, shared randomness, and so on.  Allcock \etal \cite{closedCorrelationSets} defined any correlation that Alice and Bob can make by ``wiring'' together correlations in $S$ to be the set \textbf{wirings}$(S)$.  If $S = \mathrm{wirings}(S)$, $S$ is said to be \textbf{closed under wirings.}  The sets $S, Q,$ and $NS$ are all closed under wirings.

We are interested in when Alice and Bob can use correlations in $S$ to collapse communication complexity.
As discussed in Section~\ref{sec:overview_rel}, given a set $S$ that is closed under wirings, Alice and Bob can use it to collapse communication complexity if and only if $\BBLMTU(S)$ supports reliable computation.  
The transformation $\BBLMTU$ is depicted in Figure~\ref{fig:nonlocalcircuit}.
This transformation was also used in \cite{Brassard}, and we formally define it as follows.

\begin{define}\label{define:BBLMTU}
Let $n\geq 1$. Let $c = (A,B,X,Y,p)$ be a bipartite nonsignalling correlation with output alphabets $A=B=\{0,1\}^m$ and input alphabets $X=Y=\{0,1\}^n$. Then $\BBLMTU(c)$ is the stochastic map $d: \{0,1\}^n\rightarrow \{0,1\}^m$
defined by
\begin{equation}\label{eq:BBLMTUdist}
\mathbb{P}[d(\mathbf{z}) = \mathbf{w}] = \mathbb{P}_{\mathbf{u}\sim\{0,1\}^n}\left[\XOR(c(\mathbf{u}, \mathbf{u}\oplus\mathbf{z})) = \mathbf{w}\right],
\end{equation}
That is, $\mathbf{w}\sim d(\mathbf{z})$ is the random variable 
in which $\mathbf{u}$ is sampled uniformly from $\{0,1\}^n$, then $\mathbf{a}, \mathbf{b} \in \{0,1\}^m$ are sampled from $c(\mathbf{u}, \mathbf{u}\oplus\mathbf{z})$, and then $\mathbf{w}$ is set to $\mathbf{a}\oplus \mathbf{b}$, in which $\oplus$ acts elementwise.

If $S$ is a set of bipartite nonsignalling correlations, we define the following circuit model:
$$\BBLMTU(S) = \conv\{\text{circuits from gates in }\{\BBLMTU(c) : c \in S \}\}.$$
\end{define}

Above, we defined $\BBLMTU(S)$ as the convex hull of circuits formed out of gates $\BBLMTU(c)$ so that $c \in S$.  This description is necessary to talk about $\BBLMTU(S)$ as a circuit model, but fortunately, when $S$ is closed under wirings, $\BBLMTU(S)$ has a simpler description.

\begin{proposition}\label{proposition:closedUnderWiringsImpliesBBLMTUisElementwise}
Let $S$ be a set of bipartite nonsignalling correlations so that $C \subseteq S$ and $S$ is closed under wirings. Then 
\begin{equation}\label{eq:BBLMTUcircuitsEqBBLTMUmap}
\BBLMTU(S)  = \{\BBLMTU(c) : c \in S\}.
\end{equation}
\end{proposition}
\begin{proof}
First, we observe that $\{\BBLMTU(c) : c \in S\} \subseteq \BBLMTU(S)$ trivially.  For the other direction, suppose that $d$ is a circuit with gates from $\{\BBLMTU(c) : c \in S\}$.  Thus, $d$ is obtainable by wiring together correlations in $S$, and so $d \in \mathrm{wirings}(S)$.  Since $S$ is closed under wirings, $d \in S$.\footnote{We note that there is a slight subtlety here, which is that when when two gates $\BBLMTU(c)$ and $\BBLMTU(c')$ are composed, this is not the same as directly wiring together the correlations $c$ and $c'$.  In more detail, directly wiring the outputs of $c$ to the inputs of $c'$ would mean that if $c$ outputs $x,y$ so that $x \oplus y = z$, then $c'$ would take as input $x$ and $y$.  However, composing $\BBLMTU(c)$ and $\BBLMTU(c')$ and then translating back to the correlations $c$ and $c'$ means that the outputs $x,y$ of $c$ would be reshared as $x' \oplus y' = z$, and then $x',y'$ would be the inputs to $c'$.  Fortunately, Alice and Bob can simulate this using shared randomness (which they have since $C \subseteq S$ and $S$ is closed under wirings), by defining $x' = x \oplus r$ and $y' = y \oplus r$ for a uniformly random bit $r$.  Thus the composition of $\BBLMTU(c)$ and $\BBLMTU(c')$ is of the form $\BBLMTU(d)$ for some  $d \in \mathrm{wirings}(S) = S$.} 
This establishes that $\{ \text{circuits from gates } \{ \BBLMTU(c) : c \in S \} \} \subseteq \{ \BBLMTU(c) : c \in S \}$.
Finally, since $C \subseteq S$ and $S$ is closed under wirings, $S$ is also closed under probabilistic mixtures, which implies that $\BBLMTU(S) \subseteq \{ \BBLMTU(c) : c \in S \}$.
\end{proof}

As per Proposition~\ref{prop:thepoint}, $S$ collapses communication complexity (in the sense described in Section~\ref{sec:overview_rel}) if and only if $\BBLMTU(S)$ supports reliable computation.  To that end, we will actually define trivial probabilistic communication complexity in this language.
\begin{define}\label{define:TPCC}
For a set $S$ of bipartite nonsignalling correlations which is closed under wirings and such that $C \subseteq S$, $S$ has {\bf trivial probabilistic communication complexity} if and only if the circuit model $\BBLMTU(S)$ supports reliable computation.
\end{define}

Alice and Bob will utilize nonsignalling correlations to play nonlocal games.  Formally, we define a nonlocal game as follows.

\begin{define}\label{define:nonlocalGame}
A {\bf two-player nonlocal game} $G = (X,Y,A,B,\pi,D)$ consists of finite sets of possible questions $X$ and $Y$ for the two players, finite sets $A$ and $B$ of possible answers, a probability distribution $\pi:X\times Y\rightarrow [0,1]$ over the questions, and a predicate $D: X\times Y \times A \times B \rightarrow \{0,1\}$. Given a set of correlations $S$, we define the {\bf $S$-value} of the game $G$ by
$$\omega_{S}(G) := \sup_{c\in S}{\P_{x,y\sim \pi}[D(x, y, c(x,y))=1]},$$
in which the supremum is over $c\in S$ for which the input and output alphabets match those of $G$, and in which the probability is also over the randomness of $c$.
That is, $\omega_{S}(G)$ is the optimal success probability when the game is played with access to correlations in $S$.
\end{define}

We will need the following ways to combine two nonlocal games.
The first is the \em conjunction \em of two games.
Informally, the conjuction $G_1 \land G_2$ is the game in which Alice and Bob must play both $G_1$ and $G_2$ at the same time, and win a round only if they answer correctly for both $G_1$ and $G_2$.
\begin{define}\label{define:NLConjunction}
For two nonlocal games $G_1, G_2$ with $G_i = (X_i,Y_i,A_i,B_i,\pi_i,D_i)$, the {\bf conjunction} $G_1\land G_2$ is the nonlocal game
$$G_1\land G_2 = \left( X_1\times X_2, Y_1\times Y_2, A_1\times A_2, B_1\times B_2, \pi_1\otimes  \pi_2, D_1\otimes D_2 \right).$$
Here, $(\pi_1\otimes \pi_2)((x_1, x_2), (y_1, y_2)) = \pi_1(x_1, y_1)\pi_2(x_2, y_2)$, and similarly for $D_1\otimes D_2$.
\end{define}

The second is the \em mixture \em of two games.  Informally, the mixture $qG_1 + (1-q)G_2$  is the game in which with probability $q$, Alice and Bob must play $G_1$ and with probability $1-q$, they must play $G_2$. For each round, Alice and Bob are told which game they must play.

\begin{define}\label{define:NLMixture}

For two nonlocal games $G_1, G_2$ with $G_i = (X_i,Y_i,A_i,B_i,\pi_i,D_i)$, for $q\in [0,1]$, the {\bf mixture} $qG_1 + (1-q)G_2$ is the nonlocal game
$$qG_1 + (1-q)G_2 = (X_1\sqcup X_2, Y_1\sqcup Y_2, A_1\sqcup A_2, B_1\sqcup B_2, \pi, D)$$
in which 
$$\pi(x,y) = \begin{cases}
q\pi_1(x,y) & x\in X_1, y\in Y_1\\
(1-q)\pi_2(x,y) & x\in X_2, y\in Y_2\\
0 & \text{otherwise}
\end{cases},$$
and similarly
$$D(x,y, a, b) = \begin{cases}
D_1(x,y,a,b) & x\in X_1, y\in Y_1, a\in A_1, b\in B_1\\
D_2(x,y,a,b) & \text{otherwise}
\end{cases}.$$
\end{define}

Finally, we define the \em trivial game. \em 
\begin{define}\label{define:trivialNLGame}
The {\bf trivial nonlocal game $G_T$} has $X=Y=A=B=\{\perp\}$ for some unique symbol $\perp$, and $\pi(\perp,\perp) = D(\perp,\perp,\perp,\perp) = 1$.
\end{define}

\section{Proof of Theorem~\ref{thm:mainOneSixthFormulas}: Sharp threshold for reliable computation}\label{sec:proofOneSixthFormulas}
In this section we prove Theorem~\ref{thm:main}.  We begin by proving Lemma~\ref{claim:mainClaim}, which we restate below.

\begin{lem*}[Lemma~\ref{claim:mainClaim}, restated]
Let $\mathcal{F}_{\eps, \tau}$ be the class of formulas on $\{\land_\eps, \oplus_\tau\}$, and suppose that $\eps \in (1/6, 5/6)$, and $\tau \in (0, 1)$.
Fix $\Delta >0$ and let $f: \{0,1\}^n \to \{0,1\}$ be a function that is computable with probability at least $1/2 + \Delta$ by functions in $\conv \mathcal{F}_{\eps, \tau}$.
Then $f$ depends on at most a constant number of inputs.
\end{lem*}

To prove Lemma~\ref{claim:mainClaim}, we begin by showing in Lemma~\ref{lem:depth} below that in any probabilistic mixture of noisy formulas, there is some variable which is on average at high depth.  Next we show in Lemma~\ref{lem:weight} that a certain quantity related to the bias on a wire decays by a constant factor at each gate.  Finally we put these two lemmas together in Section~\ref{sec:pfMainClaim} to prove Lemma~\ref{claim:mainClaim}.  In Section~\ref{sec:pfMainThm} we show how to use Lemma~\ref{claim:mainClaim} and Theorem~\ref{thm:noiseThresholdsAreStrict} to prove Theorem~\ref{thm:main}.

\subsection{The ``Depth'' of Variables in Mixtures of Formulas}
In this section, we prove Lemma~\ref{lem:depth}, stated below, which roughly says that in any mixture of formulas, there is some variable $X_i$ with large depth.  Before we state the lemma, we introduce some notation.

Let $C = \sum_{j=1}^N p_j F_j$ be a mixture of formulas $F_j$.  Let $\mathbf{X} = \{X_1, \ldots, X_n\}$ be the set of inputs to $C$.  Thus, the set of inputs to each $F_j$ is some subset of the variables in $\mathbf{X}$.

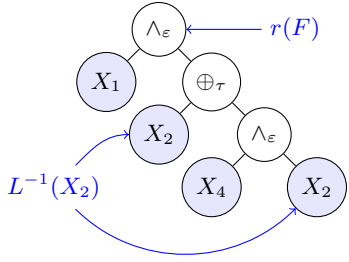
\begin{figure}
\begin{center}
\begin{tikzpicture}[scale=.7]
\footnotesize

\node[draw,circle](r) at (0,0) {$\land_\eps$};
\node[right=1cm of r,blue](rr) {$r(F)$};
\draw[->,blue] (rr) to (r);
\node[draw,circle,fill=blue!10](a) at (-1,-1) {$X_1$};
\node[draw,circle](b) at (1,-1) {$\oplus_\tau$};
\draw (r) -- (a);
\draw (r) -- (b);
\node[draw,circle, fill=blue!10](c) at (0, -2) {$X_2$};
\node[draw,circle](d) at (2,-2) {$\land_\eps$};
\draw (b) -- (c);
\draw (b) -- (d);
\node[draw,circle, fill=blue!10](e) at (3, -3) {$X_2$};
\node[draw,circle, fill=blue!10](f) at (1,-3) {$X_4$};
\draw (d) -- (e);
\draw (d) -- (f);
\node[blue](lbl) at (-2,-3) {$L^{-1}(X_2)$};
\draw[->,blue] (lbl) to[out=45,in=180] (c);
\draw[->,blue] (lbl) to[out=-45,in=-135] (e);
\end{tikzpicture}
\end{center}
\caption{An example of a formula $F$ in $\mathcal{F}_{\eps, \tau}$.  In this example, we have $\rho(F,X_2) = 2$, because the first occurence of $X_2$ is at depth $2$; $\rho(F,X_3) = \infty$, since $X_3$ never appears; and $\rho(F,X_4) = 3$.
} \label{fig:notation}
\end{figure}

For a formula $F = (G, L)$ and $X_i \in \mathbf{X}$,
Define the {\bf rank} $\rho(F, X_i)$ by
\[ \rho(F, X_i) =  
\min_{v\in L^{-1}(X_i)} d (v) 
\]
where by convention $\rho(F, X_i) = \infty$ if the set is empty.
Above, recall that $d (v)$ is the depth of vertex $v$ in the graph $G$, and $L^{-1}(X_i)$ is the set of leaves in $G$ that are labeled with the input variable $X_i$.  An example of the notation is given in Figure~\ref{fig:notation}.

\begin{lem}\label{lem:depth}
There is a function $f(\theta,k)$ so that
\[ \lim_{k\to\infty} f(\theta,k) = 0 \ \ \ \  \forall \theta \in (0,1) \]
and so that the following holds.
Fix any $\theta \in (0, 1)$.
Let $C = \sum_{j=1}^N p_j F_j$ be mixture of $N$ formulas, so that each $F_j$ is composed of two-input gates.
Suppose that $C$ takes as input the variables $\mathbf{X} = \{X_i : i \in [n]\}$.  
Choose any nonempty subset $S \subset \mathbf{X}$, and let $k = |S|$.
Then there exists a variable $X_i \in S$ for which
$$\sum_{j=1}^{N}{p_j}{\theta^{\rho(F_j, X_i)}} \leq f(\theta, k)$$
\end{lem}
\begin{proof}
Both \cite{pippengerFirstUpperBound} and \cite{maxGateNoise} proved simpler statements of a similar nature for formulas, rather than for mixtures of formulas. For that purpose, the upper bound of $\sum_{i=1}^{k}\sum_{v\in L^{-1}(X_i)}{2^{-\rho(F, X_i)}} \leq 1$ sufficed. Since we allow convex combinations of formulas, and also may have $\theta > 1/2$, we require a new argument to obtain an $o(1)$ bound. 

Fix a formula $F$ and a nonempty subset $S \subset \mathbf{X}$, and let $k = |S|$. We define the sum $W_\theta$ as follows:
\begin{equation}\label{eq:defWsumThetaRanks}
W_\theta(F) = \sum_{X\in S}{\theta^{\rho(F, X)}}.
\end{equation}
We will upper bound $W_\theta(F)$.
Let $d = d(F)$ be the depth of $F$.
We may write $W_\theta(F)$ alternatively as
$$W_\theta(F) = \sum_{s=0}^{d}{n_s \theta^{i}}$$
in which
$n_s$ is the number of variables $X_i \in S$ so that $\rho(F,X_i) = s$.
Because there are at most $2^s$ variables at the depth $s$ level of the tree, $n_s \leq 2^s$.
Consider the optimization problem
\begin{align*}
\text{maximize }&\sum_{s=0}^{d}{n_s \theta^s}\\
\text{subject to }&\sum_{s=0}^{d}{n_s} = k\\
\text{and to }&0 \leq n_s \leq 2^s  \, \forall s \in \{0,...,d\}.
\end{align*}
Clearly $W_\theta(F)$ is bounded above by the optimal value of this problem. Moreover, it is not hard to see that the optimal value of this problem is attained by concentrating all weight of the $n_s$ on the lowest levels of $s$, subject to the constraint that $n_s\leq 2^s$. This implies that
\begin{align}
W_\theta(F)&\leq \begin{cases}
\frac{(2\theta)^{\log_2(k+1)} - 1}{2\theta - 1} & \theta \neq 1/2\\
\log_2(k+1) & \theta = 1/2
\end{cases}.
\label{eq:upperBoundWthetaF}
\end{align}

Now consider the convex combination of formulas $C = \sum_{j=1}^{N}{p_j F_j}$. We extend the formula for $W_\theta$ in the natural way:
$$W_\theta(C) = \sum_{j=1}^{N}{p_j \sum_{i=1}^{k}{\theta^{\rho(F_j, X_i)}}} = \sum_{X_i \in S}\sum_{j=1}^{N}{p_j}{\theta^{\rho(F_j, X_i)}}.$$
Since $\sum_{j=1}^{N}{p_j} = 1$, our upper bound still applies:
$$W_\theta(C) = \sum_{X_i \in S}\sum_{j=1}^{N}{p_j}{\theta^{\rho(F_j, X_i)}} \leq \begin{cases}
\frac{(2\theta)^{\log_2(k+1)} - 1}{2\theta - 1} & \theta \neq 1/2\\
\log_2(k+1) & \theta = 1/2
\end{cases}$$
which implies that for some $X_i \in S$,
$$\sum_{j=1}^{N}{p_j \theta^{\rho(F_j, X_i)}} \leq \begin{cases}
\frac{(2\theta)^{\log_2(k+1)} - 1}{k(2\theta - 1)} & \theta \neq 1/2\\
\frac{1}{k}\log_2(k+1) & \theta = 1/2
\end{cases}$$
Clearly the bound in the $\theta = 1/2$ case is $o(1)$.
The bound in the $\theta \neq 1/2$ case may be rearranged as
$$\frac{(2\theta)^{\log_2(k+1)} - 1}{k(2\theta - 1)} = \frac{(k+1)^{\log_2(2\theta)} - 1}{k(2\theta - 1)}.$$
Since $\theta \in (0, 1)$, $2\theta < 2$, $\log_2(2\theta) < 1$, this bound is $o(1)$ as well.

\end{proof}

\subsection{Bias Reduction by Deep Formulas}

In this section, we prove Lemma~\ref{lem:weight} below, which says roughly that the bias (that is, the amount of signal) on each wire decays at every noisy gate.

Fix a formula $F = (L, G) \in \mathcal{F}_{\epsilon, \tau}$ that is {\bf univariate}, meaning its input wires are labeled with only constants or a single variable $X_i$. That is, $L(V) \subseteq \{X_i, 0, 1\}\cup \{\land_\eps, \oplus_\tau\}$.

Let $v \in G$ denote some internal vertex of $F$, let $g = L(v)$ be the gate that occurs at $v$, let $A, B$ denote the input wires to $g$, and let $C$ denote its output wire. For $(w, W) \in \{(a, A), (b, B), (c, C)\}$, define
$$w(F, X_i):= \frac{1}{2} \mathbb{E}[W | X_i = 1] + \frac{1}{2} \mathbb{E}[W  | X_i = 0]$$
and
$$\delta_w(F, X_i) := \mathbb{E}[W  | X_i = 1] - \mathbb{E}[W | X_i = 0]$$
where the probabilities are over the randomness of the gates in the subformula below $W$ only. Notice that this makes the probabilities for $A, B$ independent since $F$ is a formula.
This notation is illustrated in Figure~\ref{fig:notation2}.
For notational clarity, we will  omit the arguments of all $\delta_w$ and $w$ when they are clear from context.

\begin{figure}
\centering
\begin{tikzpicture}[xscale=1.7]

\node[blue](l1) at (-3.7, .5) {$a = a(F,X_i) = \frac{1}{2} \mathbb{E}[A | X_i = 1] + \frac{1}{2} \mathbb{E}[A  | X_i = 0] $ };
\node[blue](l2) at (-4, -.3) {$\delta_a = \delta_a(F,X_i) = \mathbb{E}[A  | X_i = 1] - \mathbb{E}[A | X_i = 0] $ };

\node[draw,circle](g) at (0,0) {$g$};
\draw (g) to node[above,pos=0.5](A) {$A$} (-1,-1);
\draw (g) to node[above,pos=0.5] {$B$} (1,-1);
\draw (g) to node[left,pos=0.5] {$C$} (0,1);
\draw (-1,-1) -- (-1.8,-2) -- (-.2,-2) -- cycle;
\draw (1,-1) -- (1.8,-2) -- (.2,-2) -- cycle;
\draw[dashed] (0,2) -- (-2.5,-2.2) -- (2.5,-2.2) -- cycle;
\node at (1,.7) {$F$};
\draw[->,blue](l1) to[out=0,in=120] (A);
\draw[->,blue](l2) to (A);

\node[draw,circle,fill=blue!10](\i)(a) at (-1.6,-2.9) {\small $X_i$};
\node[draw,circle,fill=green!10](\i)(b) at (-1,-2.9) { $1$};
\node[draw,circle,fill=blue!10](\i)(c) at (-.4,-2.9) {\small $X_i$};

\draw (a) -- (-1.6, -2);
\draw (b) -- (-1, -2);
\draw (c) -- (-.4, -2);

\begin{scope}[xshift=2cm]
\node[draw,circle,fill=green!10](\i)(a) at (-1.6,-2.9) {$1$};
\node[draw,circle,fill=green!10](\i)(b) at (-1,-2.9) { $0$};
\node[draw,circle,fill=blue!10](\i)(c) at (-.4,-2.9) {\small $X_i$};

\draw (a) -- (-1.6, -2);
\draw (b) -- (-1, -2);
\draw (c) -- (-.4, -2);
\end{scope}

\end{tikzpicture}
\caption{The notation used for Lemma~\ref{lem:weight}.}\label{fig:notation2}
\end{figure}
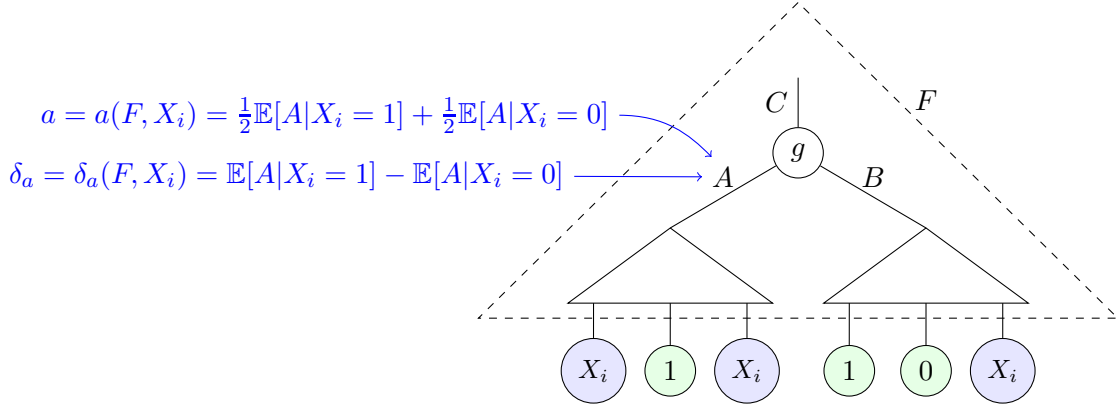

We shall now consider expressions for the quantities $\delta_c, c$ in terms of the quantities $\delta_a, a, \delta_b, b$. To simplify notation, we will rename $c$ to $d$ if $g = \land_\eps$ and rename $c$ to $e$ if $g = \oplus_\tau$. 
The following expressions (which hold for any $F,X_i$) are not hard to derive and follow from the independence of the noise on the subformulas beneath $A$ and $B$:
$$d = (1-2\epsilon)\left[ab + \frac{1}{4}\delta_a\delta_b\right] + \epsilon$$
$$\delta_d = (1-2\epsilon)(b\delta_a + a\delta_b)$$
$$e = (1-2\tau)\left[-2 a b  - \frac{1}{2}\delta_{a} \delta_{b} +  a  + b\right] + \tau$$
$$\delta_e = (1-2\tau)((1-2a) \delta_b+ (1- 2b)\delta_a ).$$

With this notation out of the way, we state the main lemma in this section, which implies that the bias on a wire decays with the depth of that wire in the circuit.  In particular, this will imply that inputs that are too deep cannot have a very big effect on the output of the circuit.

\begin{lem}[Weight Decay]\label{lem:weight}
For all $(\epsilon, \tau) \in (1/6,5/6)\times(0,1)$, there exists $\theta \in [0, 1)$ such that for all $a, b \in (0, 1), \delta_a, \delta_b \in [-1, 1], |\delta_a| + |2a-1| \leq 1, |\delta_b| + |2b-1| \leq 1$,
\begin{equation}\label{eq:weightDecay}
    \max\left\{\frac{|\delta_d|}{1-|2d-1|},  \frac{|\delta_e|}{1-|2e-1|}\right\}  \leq \theta \max\left\{ \frac{|\delta_a|}{1-|2a-1|}, \frac{|\delta_b|}{1-|2b-1|} \right\}.
\end{equation}
\end{lem}
\begin{proof}

We prove Lemma~\ref{lem:weight} in two cases, depending on whether the maximum on the left hand side is attained by $d$ (an $\land_\eps$ gate) or by $e$ (a $\oplus_\tau$ gate). We assume that $(\epsilon, \tau) \in (1/6,1/2] \times (0,1/2]$; this is without loss of generality as explained in Remark~\ref{rmk:NOTsComeForFree}.

\paragraph{The $\land_\eps$ case.}
If $\eps = 1/2$, then $\delta_d = 0$ and we are done.  Thus, we assume that $\eps \in (1/6, 1/2)$.
Defining $\sigma := \frac{1}{1-2\epsilon}$, and noting that for any $\epsilon \in (1/6, 1/2)$ we have $\frac{1}{1-2\epsilon} > \frac{3}{2}$, it suffices to show that for all $\sigma > 3/2$,
\[
\sup_{\frac{|\delta_b|}{1-|2b-1|} \leq \frac{|\delta_a|}{1-|2a-1|} \leq 1 } \frac{|b\delta_a + a\delta_b|}{\sigma - |2 ab + \frac{1}{2}\delta_a\delta_b - 1 |}  -  \frac{|\delta_a|}{1-|2a-1|} < 0. \]
This holds if and only if
\[ \sup_{\frac{|\delta_b|}{1-|2b-1|} \leq \frac{|\delta_a|}{1-|2a-1|} \leq 1 } \frac{1-|2a-1|}{|\delta_a|}|b\delta_a + a\delta_b| + \left|2 ab + \frac{1}{2}\delta_a\delta_b - 1 \right|  < \sigma \]
which in turn holds if and only if
\begin{equation} \sup_{\frac{|\delta_b|}{1-|2b-1|} \leq \frac{|\delta_a|}{1-|2a-1|} \leq 1 } \max_{s_0, s_1 \in \{-1, 1\}}\left\{(1-|2a-1|)s_0\left(b + a\frac{\delta_b}{\delta_a}\right) + s_1\left(2 ab + \frac{1}{2}\delta_a\delta_b - 1 \right)  \right\} \leq \frac{3}{2}.\label{eq:optimizationANDepsRefined}
\end{equation}
For fixed $s_0, s_1, \delta_a, a$, this is an affine function in $b, \delta_b$ whose extrema are obtained on the feasible set boundary, which is the surface defined by
$$\frac{|\delta_b|}{1-|2b-1|} = \frac{|\delta_a|}{1-|2a-1|},$$
and so for some $s_2 \in \{-1, 1\}$,
the optimal point of \eqref{eq:optimizationANDepsRefined} has
$$\delta_b = s_2 \frac{1-|2b-1|}{1-|2a-1|}\delta_a.$$
Plugging this back in to \eqref{eq:optimizationANDepsRefined} we must show
\begin{multline}
\sup_{\frac{|\delta_b|}{1-|2b-1|} \leq \frac{|\delta_a|}{1-|2a-1|} \leq 1 } \max_{s_0, s_1, s_2 \in \{-1, 1\}}\left\{(1-|2a-1|)s_0\left(b + as_2 \frac{1-|2b-1|}{1-|2a-1|}\right) + \right. \\
\left. s_1\left(2 ab + \frac{1}{2}\delta_a^2 s_2 \frac{1-|2b-1|}{1-|2a-1|} - 1 \right)  \right\} \leq \frac{3}{2}.
\label{eq:toshow}
\end{multline}
The criterion \eqref{eq:toshow} can be established numerically by checking many cases.  More precisely,
we notice that for fixed $a, \delta_a, s_0, s_1, s_2$, this is an affine function of $b, |2b-1|$ whose extrema with respect to $b$ must occur at either a critical point (which only occurs at $2b-1=0$) or an endpoint of the interval $b \in [0, 1]$. The function depends on $\delta_a$ only with a quadratic term, so its maximum with respect to $\delta_a$ must occur at either a critical point (which only occurs at $\delta_a=0$) or an endpoint of the interval $\delta_a\in [-(1-|2a-1|), 1-|2a-1|]$. Evaluating the 72 expressions obtained by substituting
$$(b, \delta_a, s_0, s_1, s_2)\in \{ 0, 1/2,  1\} \times \{-(1-|2a-1|), 0, 1-|2a-1|\} \times \{-1, +1\}^3,$$
we find 7 distinct functions of $a$ (up to overall sign):
\begin{align}
\pm \left( -2 \, a + \frac{1}{2} \, {\left| 2 \, a - 1 \right|} + \frac{1}{2} \right) \label{eq:first}\\
\pm 1 \\
\pm \left( 2 \, a - {\left| 2 \, a - 1 \right|} \right) \\
\pm \left( -2 \, a - \frac{1}{2} \, {\left| 2 \, a - 1 \right|} + \frac{3}{2} \right) \\
\pm \left( 2 \, a + {\left| 2 \, a - 1 \right|} - 2 \right) \\
\pm \left( \frac{1}{2} \, {\left| 2 \, a - 1 \right|} - \frac{3}{2} \right) \\
\pm \left( \frac{1}{2} \, {\left| 2 \, a - 1 \right|} + \frac{1}{2} \right)\label{eq:nth}
\end{align}
Since these are all affine in $a$ and $|2a-1|$, the maximum with respect to $a$ must occur at either a critical point (which only occurs at $2a-1=0$) or at an endpoint of the interval $a \in [0, 1]$. Checking the 21 cases resulting from substituting $a\in \{0, 1/2, 1\}$ into the seven equations \eqref{eq:first}-\eqref{eq:nth}, we have that the maximum is $\frac{3}{2}$, when $a=1/2$. Therefore \eqref{eq:optimizationANDepsRefined} holds, and we are done with this case.

\paragraph{The $\oplus_\tau$ case.}
If $\tau = 1/2$, then $\delta_e = 0$ and we are done.  Thus we assume that $\tau \in (0,1/2)$.
Defining $\mu := \frac{1}{1-2\tau}$, and noting that $\forall \tau \in (0, 1/2), \frac{1}{1-2\tau} > 1$, it suffices to show that for all $\mu > 1$,
\[ \sup_{\frac{|\delta_b|}{1-|2b-1|} \leq \frac{|\delta_a|}{1-|2a-1|} \leq 1 } \frac{|(2a-1) \delta_b+ (2b-1) \delta_a |}{\mu-|\delta_a \delta_b + (2a-1)(2b-1) |}  - \frac{|\delta_a|}{1-|2a-1|} < 0. \]
Reparameterizing with $x_a := 2a - 1$ and $x_b := 2b - 1$, this becomes:
\[ \sup_{\frac{|\delta_b|}{1-|x_b|} \leq \frac{|\delta_a|}{1-|x_a|} \leq 1 } \frac{|x_a \delta_b+ x_b \delta_a |}{\mu-|\delta_a \delta_b + x_a x_b |}  - \frac{|\delta_a|}{1-|x_a|} < 0. \]
Using the triangle inequality, along with the fact that $|\delta_a\delta_b| + |x_ax_b| \leq 1 \leq \mu$, it suffices to show
\[ \sup_{\frac{|\delta_b|}{1-|x_b|} \leq \frac{|\delta_a|}{1-|x_a|} \leq 1 } \frac{|x_a \delta_b|+ |x_b \delta_a |}{\mu-|\delta_a \delta_b| - |x_a x_b |}  - \frac{|\delta_a|}{1-|x_a|} < 0 \]
which is equivalent to
\begin{align}
\sup_{x_a, \delta_a, x_b, \delta_b \in [0, 1], \frac{\delta_b}{1-x_b} \leq \frac{\delta_a}{1-x_a} \leq 1 } \frac{x_a \delta_b+ x_b \delta_a }{\mu - \delta_a \delta_b - x_a x_b }  - \frac{\delta_a}{1 - x_a} &< 0 \notag\\
 \sup_{x_a, \delta_a, x_b, \delta_b \in [0, 1], \frac{\delta_b}{1-x_b} \leq \frac{\delta_a}{1-x_a} \leq 1 } (x_a \delta_b+ x_b \delta_a)\frac{1-x_a}{\delta_a} + \delta_a \delta_b +  x_a x_b &<  \mu \label{eq:optimizationXORkappaReducedTriangleAbsfree1} 
\end{align}
The expression
 \eqref{eq:optimizationXORkappaReducedTriangleAbsfree1} is affine in $\delta_b$, so the extrema must occur for $\delta_b$ on the boundary of the feasible set; that is, the extrema occur at $\delta_b \in \left\{0,  \delta_a\frac{1-x_b}{1-x_a}\right\}$.  The case where $\delta_b = 0$ simplifies the expression to $x_b \leq 1$, which holds trivially by the constraints.  Therefore, it suffices to show that
\[ \sup_{x_a, \delta_a, x_b \in [0, 1], \frac{\delta_a}{1-x_a} \leq 1 } x_a (1-x_b)+ x_b(1-x_a) + \delta_a^2 \frac{1-x_b}{1-x_a} + x_a x_b \leq 1 \]
Since $\frac{1 - x_b}{1 - x_a} \geq 0$, the expression above is maximized when $\delta_a = 1 - x_a$ is as large as possible, which means that it suffices to show that:
\[ \sup_{x_a, x_b \in [0, 1] } x_a (1-x_b)+ x_b(1-x_a) + (1-x_a)^2 \frac{1-x_b}{1-x_a} + x_a x_b \leq 1.\]
Finally, this last expression simplifies to read
\[ \sup_{x_a, x_b \in [0, 1] } 1 \leq 1 \]
which is true.

This finishes the proof of this case, and of the lemma.
\end{proof}

\begin{rmk}\label{rmk:NOTsComeForFree}
Notice that Equation~\eqref{eq:weightDecay} of Lemma~\ref{lem:weight} is invariant under mapping $(w, \delta_w) \mapsto  (1-w, -\delta_w)$, for any $w\in \{a,b,c, d,e\}$.
This implies that the proof goes through even if we are allowed to apply unary noise-free $\neg$ gates.
In particular, this implies that checking the case of $(\epsilon, \tau) \in (1/6,1/2] \times (0,1/2]$ suffices to complete the proof of Lemma~\ref{lem:weight}.
\end{rmk}

\subsection{Proof of Lemma~\ref{claim:mainClaim}}\label{sec:pfMainClaim}
Now we put Lemmas~\ref{lem:depth} and \ref{lem:weight} together to prove Lemma~\ref{claim:mainClaim}.
\begin{proof}[Proof of Lemma~\ref{claim:mainClaim}]
Fix $\Delta > 0$.
Fix $\epsilon  \in (1/6, 1/2], \tau \in (0,1/2]$.
For $n=2,3,\ldots,$ let $f^{(n)}: \{0,1\}^n\rightarrow \{0,1\}$ denote a sequence of Boolean functions so that $f^{(n)}$ depends on all $n$ inputs.
Let $C^{(n)} \in \conv \mathcal{F}_{\epsilon, \tau}$ denote a sequence of mixtures over $N^{(n)}$ formulas on $n$ inputs in $\mathcal{F}_{\epsilon, \tau}$.
Now fix $n$ and write
$$C^{(n)} = C = \sum_{j=1}^{N}{p_j F_{j}},$$
where as above we shall drop the dependence on $n$ for clarity.
By Lemma~\ref{lem:depth}, there exists a choice of input $X_i$ such that $\forall \theta \in [0,1)$,
\begin{equation}\label{eq:FjThetao1Bd}
\sum_{j}{p_j}{\theta^{\rho(F_j, X_i)}} \leq o(1).
\end{equation}
Fix this $i$.
Let $\mathbf{q}\in\{0,1\}^{n-1}$ 
denote some bitstring such that
$$f\left(q_1, \ldots, q_{i-1}, 0, q_{i+1}, \ldots, q_{n}\right) \neq f\left(q_1, \ldots, q_{i-1}, 1, q_{i+1}, \ldots, q_{n}\right).$$
Note that such a bitstring $\mathbf{q}$ exists since $f$ depends on all of its inputs.
Let
$$Q = C\left(q_1, \ldots, q_{i-1}, X_i, q_{i+1}, \ldots, q_{n}\right).$$
Let $E_j = F_j\left(q_1, \ldots, q_{i-1}, X_i, q_{i+1}, \ldots, q_{n}\right)$ so that
$$Q = \sum_{j=1}^{N}{p_j E_j}.$$
Define
$$\delta(E_j, X_i) = \mathbb{E}[E_j|X_i=1] - \mathbb{E}[E_j|X_i=0].$$
Now, since there are $\rho(E_j,X_i)$ gates between each input labeled $X_i$ and the output $r(E_j)$, Lemma~\ref{lem:weight} implies
that there exists $\theta \in [0, 1)$ so that
\begin{equation}\label{eq:biasEjUpperBd}
|\delta(E_j, X_i)| \leq \theta^{\rho(E_j,X_i)}.
\end{equation}
Now 
\begin{align*}
\left|\delta(Q, X_i)\right| &= \left|\sum_{j=1}^{N}{p_j\delta(E_j, X_i)}\right|\\
&\leq \sum_{j=1}^{N}{p_j|\delta(E_j,X_i)|}\\
&\leq \sum_{j=1}^{N}{p_j\theta^{\rho(E_j,X_i)}}\\
&= \sum_{j=1}^{N}{p_j\theta^{\rho(F_j,X_i)}}\\
&\leq o(1),
\end{align*}
where above we have used the triangle inequality, Equation~\eqref{eq:biasEjUpperBd}, the fact that $\rho(E_j, X_i) = \rho(F_j, X_i)$, and Equation~\eqref{eq:FjThetao1Bd}.

By the definition of $\delta$ and $Q$ (and un-fixing $n$ and $i$), this implies that for sufficiently large $n$, there is some $i$ and $\mathbf{q}$ so that
$$\left|\mathbb{E}\left[C^{(n)}\left(q_1, \ldots, q_{i-1}, 1, q_{i+1}, \ldots, q_{n}\right)\right] - \mathbb{E}\left[C^{(n)}\left(q_1, \ldots, q_{i-1}, 0, q_{i+1}, \ldots, q_{n}\right)\right]\right| < \Delta,$$
which implies that $f^{(n)}$ is not reliably computed by these $C^{(n)}$ with advantage $\Delta$.

Finally, we conclude Lemma~\ref{claim:mainClaim}.  Indeed, suppose that $f^{(n)}$ is a sequence of functions which depends on any super-constant number of inputs, and let $g^{(n)}$ denote the restriction of $f^{(n)}$ to the inputs on which it depends.  Then $g^{(n)}$ is a family of functions that depends on all of its inputs, and the argument above applies.
Therefore, any sequence $f^{(n)}$ of functions which is reliably computed by a sequence of formula mixtures $C^{(n)} \in \conv \mathcal{F}_{\epsilon, \tau}$ depends on at most a constant number of inputs.
\end{proof}

\subsection{Proof of Theorem~\ref{thm:main}}\label{sec:pfMainThm}
\begin{conj}\label{conj:AppliesToCircuits}
The bound of Lemma~\ref{claim:mainClaim} applies to circuits as well as formulas. That is, letting $\mathcal{C}_{\eps, \tau}$ denote the class of circuits on $\{\land_\eps, \oplus_\tau\}$, suppose that $\eps \in (1/6, 5/6)$ and $\tau \in (0, 1)$.
Fix $\Delta >0$ and let $f: \{0,1\}^n \to \{0,1\}$ be a function that is computable with probability at least $1/2 + \Delta$ by functions in $\conv \mathcal{C}_{\eps, \tau}$.
Then $f$ depends on at most a constant number of inputs.
\end{conj}
Finally we prove Theorem~\ref{thm:main}, assuming Conjecture~\ref{conj:AppliesToCircuits} and Theorem~\ref{thm:noiseThresholdsAreStrict}, which we prove in Section~\ref{section:proofNoiseThresholdsStrict}.  

Lemma~\ref{claim:mainClaim} and Conjecture~\ref{conj:AppliesToCircuits} imply that $\conv \mathcal{C}_{\eps,\tau}$ does not support reliable computation for any $(\eps, \tau) \in (1/6,5/6) \times (0,1)$.
Now for all $\tau\in(0,1)$, Theorem~\ref{thm:noiseThresholdsAreStrict} applied to $\cC_{\epsilon,\tau}$ with respect to the noisy gate $\land_\epsilon$ implies that $\cC_{\epsilon,\tau}$ does not support reliable computation for all $\epsilon\in[1/6,5/6]$.
Now, for all $\epsilon\in[1/6,5/6]$, Theorem~\ref{thm:noiseThresholdsAreStrict} applied to $\cC_{\epsilon,\tau}$ with respect to the noisy gate $\oplus_\tau$ implies that $\cC_{\epsilon,\tau}$ does not support reliable computation for all $\tau\in[0,1]$.
Thus, $\cC_{\eps, \tau}$ does not support reliable computation for all $\eps, \tau \in [1/6,5/6]\times[0,1]$, which proves the theorem. \begin{flushright}$\qed$\end{flushright}

\section{Proof of Theorem~\ref{thm:AmplifierEquivToFTCC_andOverhead}: Equivalence between reliable computation and amplification}\label{sec:pfEquiv}
In this section, we prove Theorem~\ref{thm:AmplifierEquivToFTCC_andOverhead}, which we restate below.

\begin{thm*}[Theorem \ref{thm:AmplifierEquivToFTCC_andOverhead}, restated]
Let $\mathcal{C}$ denote a circuit model closed under composition. Then $\conv\mathcal{C}$ supports reliable computation if and only if $\conv\mathcal{C}$ contains both an amplifier and a $\neg_\kappa$ gate for $\kappa < 1/2$.

Further,
given a circuit model $\mathcal{C}$ such that $\conv \mathcal{C}$ supports reliable computation, there exists a constant $s$ such that for any function $f:\{0,1\}^n \to \{0,1\}$ computable by a depth-$d$ circuit of noiseless $\NAND$ gates, $f$ can be computed by a depth-$(s \cdot d)$ circuit in $\conv \mathcal{C}$ with failure probability bounded away from $1/2$.
\end{thm*}

Note that Theorem~\ref{thm:AmplifierEquivToFTCC_andOverhead} does not apply to formulas, which are not closed under composition. This is the reason that cannot apply Theorem~\ref{thm:noiseThresholdsAreStrict} to tighten the threshold from Lemma~\ref{claim:mainClaim} directly, but must first move to the realm of circuits via Conjecture~\ref{conj:AppliesToCircuits}.

Before we prove Theorem~\ref{thm:AmplifierEquivToFTCC_andOverhead}, we state one definition which generalizes the amplification function to a multivariate polynomial.  That is, we substitute each variable with a possibly differently biased coin.
\begin{define}[Amplification function for multiple biases]\label{def:phi}
For a gate $g: \F_2^k \to \F_2$,
let
\[ \psi_g(p_1,\ldots, p_k) = \mathbb{P}_{X_i \sim \Ber(p_i)} \left[ g(X_1,\ldots,X_k) = 1\right]. \]
\end{define}

\begin{proof}[Proof of Theorem~\ref{thm:AmplifierEquivToFTCC_andOverhead}]

Fix $\gamma > 0$, and suppose that $\conv \mathcal{C}$ supports reliable computation with advantage $\gamma$. We wish to show that $\conv \mathcal{C}$ contains an amplifier and a $\neg_\kappa$ gate for $\kappa < 1/2$. Letting $\kappa = \frac{1}{2} (1 - \delta_0) < 1/2$, we see that $\neg_\kappa \in \mathcal{C} \subseteq \conv\mathcal{C}$ by Definition \ref{define:reliableComputationDefinition}.

Next we show that $\conv \mathcal{C}$ contains an amplifier. Since $\conv \mathcal{C}$ supports reliable computation with advantage $\gamma > 0$, for all odd $n$ there exists $c_n \in \conv \mathcal{C}$ such that for all $\mathbf{x} \in \F_2^n$,

\begin{equation}\label{eq:prMajComputedIsBoundedBelow}
\P_{c_n}[\text{Maj}_n(\mathbf{x}) = c_n(\mathbf{x})] \geq \frac{1}{2} + \gamma
\end{equation}
where the probability is taken over the stochastic behavior of $c_n$. Letting $|\mathbf{x}|$ denote the weight of bitstring $\mathbf{x}$, we may write
$$A_{c_n}(p) = \sum_{\mathbf{x} \in \mathbb{F}_2^n}{p^{|\mathbf{x}|}(1-p)^{n-|\mathbf{x}|}\P_{c_n}[c_n(\mathbf{x}) = 1]}$$
which implies that the derivative $A_{c_n}'(p)$ of $A_{c_n}(p)$ satisfies
\begin{align*}
A_{c_n}'(p) &= 
\sum_{\mathbf{x} \in \mathbb{F}_2^n} \left( |\mathbf{x}| p^{|\mathbf{x}| - 1} (1 - p)^{n - |\mathbf{x}|} - p^{|\mathbf{x}|} (n - |\mathbf{x}|) (1 - p)^{n - |\mathbf{x}| -1 }\right) \P_{c_n}\left[ c_n(\mathbf{x}) = 1 \right] \\
&= 
\sum_{\mathbf{x} \in \mathbb{F}_2^n} p^{|x|-1}(1-p)^{n-|x|-1}\left( |\mathbf{x}|(1  -p) + (n - |\mathbf{x}|) p \right) \P_{c_n}\left[ c_n(\mathbf{x}) = 1 \right] 
\end{align*}
and hence plugging in $p=1/2$,
\begin{equation}\label{eq:ampFuncDerivOnehalf}
A'_{c_n}(1/2) = \sum_{\mathbf{x} \in \mathbb{F}_2^n}{\frac{|\mathbf{x}| - n/2 }{2^{n - 2}}\P_{c_n}[c_n(\mathbf{x}) = 1]}.
\end{equation}

We split the sum over all bitstrings into those above and below weight $n/2$ and apply inequality \eqref{eq:prMajComputedIsBoundedBelow} to find
\begin{align*}
A_{c_n}'(1/2) &\geq  \sum_{k =0}^{ (n-1)/2}{\frac{k - n/2  }{2^{n - 2}}{n\choose k}\left(\frac{1}{2} -\gamma\right)} +  \sum_{k = (n+1)/2}^{n}{\frac{k - n/2 }{2^{n - 2}}{n\choose k}\left(\frac{1}{2} + \gamma\right)}\\
&\geq  \sum_{k =0}^{ (n-1)/2}{\frac{k - n/2  }{2^{n - 2}}{n\choose k}\left(\frac{1}{2} -\gamma\right)} -  \sum_{k = 0}^{(n-1)/2}{\frac{k - n/2 }{2^{n - 2}}{n\choose k}\left(\frac{1}{2} + \gamma\right)}\\
&=   \gamma \sum_{k=0}^{(n-1)/2}{{n\choose k}\frac{n/2 - k}{2^{n-1}}}.
\end{align*}
Using $k{n\choose k} = n{n-1\choose k -1}$, this reads
\begin{align*}
A_{c_n}'(1/2) &\geq  \frac{n\gamma }{2^n}\left[ 1 + \sum_{k=1}^{(n-1)/2}\left({n\choose k}  -  2{n-1 \choose k - 1}\right) \right]\\
&=  \frac{n\gamma }{2^n}\left[ 1 + \sum_{k=1}^{(n-1)/2}\left({n-1\choose k}  -  {n-1 \choose k - 1}\right) \right]\\
&=  \frac{n\gamma}{2^n} {n-1 \choose (n-1)/2 } .
\end{align*}
For large $n$, 
this lower bound is asymptotic to
\begin{equation}\label{eq:asymptoticLowerBd}
A_{c_n}'(1/2) \gtrsim \gamma \sqrt{\frac{n}{2\pi}}.
\end{equation}
To ensure that our amplifier is appropriately balanced, we define the mixture $b_n$ for each $n$ as:
\begin{equation}\label{eq:balancedMixtureExistsAmplifier}
b_n := \begin{cases}
\frac{1}{2A_{c_n}(1/2)} c_n + \left(1 - \frac{1}{2A_{c_n}(1/2)}\right) \mathbf{0} & 1/2 < A_{c_n}(1/2)  \leq 1 \\
c_n & A_{c_n}(1/2) = 1/2\\
\frac{1}{2(1-A_{c_n}(1/2))} c_n + \left(1 - \frac{1}{2(1-A_{c_n}(1/2))} \right) \mathbf{1} & 0 \leq A_{c_n}(1/2) < 1/2.
\end{cases}
\end{equation}
It is easy to show from this piecewise definition that for all $n$,
$$A_{b_n}(1/2) = \frac{1}{2}$$
as desired, and that the derivative of the amplification function of \eqref{eq:balancedMixtureExistsAmplifier} satisfies
$$\frac{1}{2}A_{c_n}'(1/2) \leq A_{b_n}'(1/2) \leq A_{c_n}'(1/2).$$
By equation \eqref{eq:asymptoticLowerBd}, since $\gamma$ is strictly greater than $0$, there exists finite $n$ such that $A_{c_n}'(1/2) > 2$, implying $A_{b_n}'(1/2) > 1$. Therefore there exists an amplifier away from $1/2$ in $\conv \mathcal{C}$.

To prove the other direction, we must show that given an amplifier away from $1/2$, as well as the  $\neg_\kappa$ gate for $\kappa < 1/2$, we may compute any Boolean function $f: \F_2^n \rightarrow \F_2$ with bounded error independent of $n$.  We will first introduce some notation. Recall from Definition~\ref{def:phi} that for a stochastic map $M$, $\psi_M:[0,1]^k \to [0,1]$ is defined by
\[ \psi_M(p_1, ..., p_k) = \P_{X_i \sim \Ber(p_i)}[M(X_1, ..., X_k) = 1] \]

Now consider a tree $T$ of $\NAND$ gates with leaves labeled by constant bits and variables (later we will use the fact that any Boolean function can be represented by such a $T$). We call $T$ a {\it NAND tree}. Our strategy will be to replace each $\NAND$ gate in $T$ with a stochastic map $\mathbf{N}$ that behaves like a $\NAND$ gate; Claim~\ref{claim:NMapExists} below guarantees that an appropriate map exists. 

\begin{claim}\label{claim:NMapExists}
Let $\mathcal{C}$ denote a circuit model closed under composition. Suppose $\mathcal{C}$ contains an amplifier and a $\neg_\kappa$ gate for some $\kappa < 1/2$. Then there exists $\beta \in (0, 1/2]$, an integer $m \geq 1$, and a map $\mathbf{N} \in \conv \mathcal{C}$ such that $\mathbf{N}$ takes $2m$ inputs and so that the following holds.  Letting
\begin{align}
I_- &= \left[\frac{1}{2}-\beta, \frac{1}{2}- \frac{\beta}{2}\right]\notag \\
I_+ &= \left[\frac{1}{2} + \frac{\beta}{2},  \frac{1}{2} + \beta\right], \notag
\end{align}
we have
\begin{align}
\psi_\mathbf{N}( (I_-)^m \times (I_-)^m) &\subseteq I_+ \label{eq:intervalMappingsNoisyNANDsimulator} \\
\psi_\mathbf{N}( (I_+)^m \times (I_-)^m) &\subseteq I_+  \notag \\
\psi_\mathbf{N}( (I_-)^m \times (I_+)^m) &\subseteq I_+  \notag \\
\psi_\mathbf{N}( (I_+)^m \times (I_+)^m) &\subseteq I_-. \notag
\end{align}
\end{claim}

We prove Claim~\ref{claim:NMapExists} in Appendix~\ref{appendix:NMapExists}, and for the rest of the current proof we will use this map $\mathbf{N}$ with the associated $\beta$ and $2m$ inputs, and take $I_+$ and $I_-$ as in the statement of the claim.

Let $f_T$ denote the Boolean function computed by $\NAND$ tree $T$. Using $\mathbf{N}$, we recursively define a transformation $\mathcal{J}$ that takes $T$ to a stochastic map $\mathcal{J}(T) \in \conv\mathcal{C}$.
For $\ell \in [0,1]$, let $\mathcal{N}_\ell$ denote the map
\begin{equation}\label{eq:noisegate}
 \mathcal{N}_\ell := \ell \mathbf{x} + \frac{1-\ell}{2}(\mathbf{0} + \mathbf{1}). 
\end{equation}
For a depth 0 tree $T$ we define $\mathcal{J}(T) = \mathcal{N}_{2\beta}$. (Notice that a depth 0 $\NAND$ tree has no $\NAND$ gates at all, and thus is either a constant $\mathbf{0}$ or $\mathbf{1}$, or is the literal Boolean variable $x$.)
Then we define $\mathcal{J}(T)$ recursively according to the process shown in 
Figure~\ref{fig:JellyfishSketch}.  That is, given a depth-$n$ $\NAND$ tree $T$, we write $T = \NAND( A, B )$, where $A, B$ are $\NAND$ trees of depth at most $n-1$.  Then we recursively define
\begin{equation}\label{eq:jellyfish}
 \mathcal{J}(T) = \mathbf{N}( \mathcal{J}(A), \mathcal{J}(A), \ldots, \mathcal{J}(A), \mathcal{J}(B), \ldots, \mathcal{J}(B) ) 
\end{equation}
where there are $m$ copies each of $\mathcal{J}(A)$ and $\mathcal{J}(B)$.

\begin{remark}[Depth complexity]\label{comment:depthFactorInrease}
Note that given $d(\mathbf{N})$, the maximum circuit depth of the map $\mathbf{N}$, it is clear that
$$d(\mathcal{J}(T)) = d(\mathbf{N}) d(T).$$
Therefore, for fixed $\kappa$, there is only a constant factor increase in depth complexity for the fault-tolerant circuit over the original $\NAND$ tree.  This establishes the second part of the theorem.
\end{remark}

\begin{figure}

\centering
\begin{tikzpicture}
\node[anchor=west](J) at (-1,0) {\Large $\mathcal{J} $};
\node[draw, semicircle](nand) at (2,1) {$\NAND$};
\node[draw,regular polygon, regular polygon sides=3](A) at (1,-1) {A};
\node[draw,regular polygon, regular polygon sides=3](B) at (3,-1) {B};
\draw (A.north) -- (nand);
\draw (B.north) -- (nand);
\draw[thick] (.2,2) to[out=180+70,in=180-70] (.2,-2);
\draw[thick] (3.8,2) to[out=-70,in=70] (3.8,-2);
\node at (5,0) {\Huge $=$};
\begin{scope}[xshift=9cm]
\node[draw,semicircle](N) at (0,1) {\Large $\mathbf{N}$ };
\node[draw](J1) at (-3,-1) {$\mathcal{J}(A)$};
\node[draw](J2) at (-1,-1) {$\mathcal{J}(A)$};
\node[draw](J3) at (1,-1) {$\mathcal{J}(B)$};
\node[draw](J4) at (3,-1) {$\mathcal{J}(B)$};
\node at (-2, -1) {$\cdots$};
\node at (2, -1) {$\cdots$};
\draw (N) to (J1);
\draw (N) to (J2);
\draw (N) to (J3);
\draw (N) to (J4);
\foreach \i in {1,1.2,1.4,1.6,1.8,2}{
\draw (N) to (-\i,-.5);
\draw (N) to (\i,-.5);
}

\end{scope}
\end{tikzpicture}
\caption{The $\mathcal{J}$ function is defined recursively, by replacing the top $\NAND$ gate in a $\NAND$ tree $T$ by a map $\mathbf{N}$. Since $\mathbf{N}$ takes $2m$ inputs, we must duplicate the input subtrees $m$ times and apply $\mathcal{J}$ to each copy.}
\label{fig:JellyfishSketch}
\end{figure}
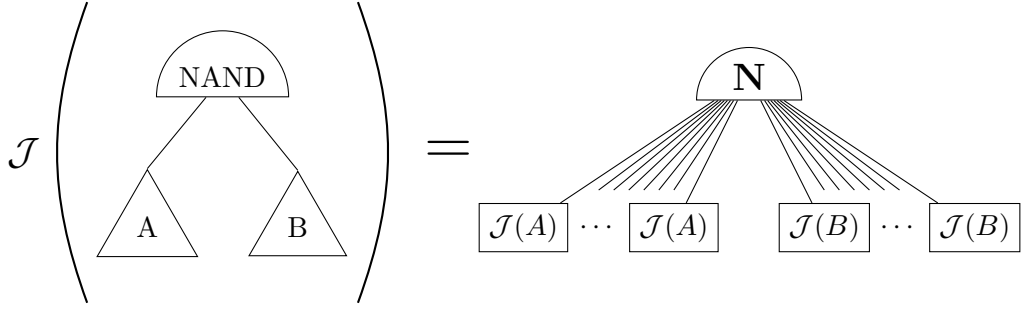

Now we prove the theorem by induction, with the inductive hypothesis that for all depth-$n$ $\NAND$ trees $T$,
\begin{equation}\label{eq:NANDTreeInductiveHypothesis}
\forall \mathbf{X} \in \F_2^{\text{number of inputs to $f_T$}} \quad  \P_{\mathcal{J}(T)}[\mathcal{J}(T)(\mathbf{X}) = f_T(\mathbf{X})] \in I_+.
\end{equation}

First we prove the base case, for depth $n=0$.
Using the base case of our recursive construction, we see that for any depth-0 $T$ that $\mathcal{J}(T) = \mathcal{N}_{2\beta}$.
It is not hard to see that
$$\psi_{\mathcal{N}_{2\beta}}(0) \in I_-$$
and
$$\psi_{\mathcal{N}_{2\beta}}(1) \in I_+$$
which establishes \eqref{eq:NANDTreeInductiveHypothesis} for $n=0$.

For the inductive step, assume that \eqref{eq:NANDTreeInductiveHypothesis} holds for all $\NAND$ trees of depth at most $n$.
Let $T$ be a depth-$(n+1)$ $\NAND$ tree, so that $T = \NAND(A,B)$ where $A$ and $B$ are both $\NAND$ trees.
By our definition \eqref{eq:jellyfish} of $\mathcal{J}$, 
\begin{align*}
\PR{ \mathcal{J}(T)(\mathbf{X}) = f_T(\mathbf{X}) }
&= \PR{ \mathbf{N}( \mathcal{J}(A), \ldots, \mathcal{J}(A), \mathcal{J}(B), \ldots, \mathcal{J}(B) ) = f_T(\mathbf{X})}.
\end{align*}

Suppose that 
\[f_T(\mathbf{X}) = \NAND( A(\mathbf{X}), B(\mathbf{X}) ) = 0,\] 
which means that $(A(\mathbf{X}), B(\mathbf{X})) = (1,1)$.
Then by the inductive hypothesis \eqref{eq:NANDTreeInductiveHypothesis},
\[ \PR{ \mathcal{J}(A)(\mathbf{X}) = 1 } \in I_+ \]
and the same for $B$, and so
\begin{align*}
1 - \PR{ \mathcal{J}(T)(\mathbf{X}) = f_T(\mathbf{X}) }
&= 1 - \PR{ \mathbf{N}( \mathcal{J}(A), \ldots, \mathcal{J}(A), \mathcal{J}(B), \ldots, \mathcal{J}(B) )(\mathbf{X}) = f_T(\mathbf{X})}  \\
&= 1 - \PR{ \mathbf{N}( \mathcal{J}(A), \ldots, \mathcal{J}(A), \mathcal{J}(B), \ldots, \mathcal{J}(B) )(\mathbf{X}) = 0 } \\
&= \PR{ \mathbf{N}( \mathcal{J}(A), \ldots, \mathcal{J}(A), \mathcal{J}(B), \ldots, \mathcal{J}(B) )(\mathbf{X}) = 1 } \\
&\in \psi_{\mathbf{N}}( (I_+)^m \times (I_+)^m ) \\
&\subseteq I_-  
\end{align*}
by the definition of $\mathbf{N}$ in Claim~\ref{claim:NMapExists}.
Aobve, we are using the fact that $T$ is a tree to say that each copy $\mathcal{J}(A)(\mathbf{X})$ and $\mathbf{J}(B)(\mathbf{X})$ are independent.  Notice that the only randomness here is over the noisy gates, and so it does not matter that the (deterministic) inputs $\mathbf{X}$ are the same for each copy.
This implies that in the case that $f_T(\mathbf{X}) = 0$, 
\[ \PR{ \mathcal{J}(T)(\mathbf{X}) = f_T(\mathbf{X}) } \in I_+. \]

On the other hand, suppose that
\[ f_T(\mathbf{X}) = 
 \NAND( A(\mathbf{X}), B(\mathbf{X}) ) = 1,\] 
which means that $(A(\mathbf{X}), B(\mathbf{X})) \in \{(0,0), (0,1), (1,0)\}$.
If it is, for example, $(0,1)$, then we have
\[ \PR{ \mathcal{J}(A)(\mathbf{X}) = 1 } \in I_- \qquad \text{and} \qquad \PR{ \mathcal{J}(B)(\mathbf{X}) = 1 } \in I_+,\]
so
\[ \PR{ \mathbf{N}( \mathcal{J}(A), \ldots, \mathcal{J}(A), \mathcal{J}(B), \ldots, \mathcal{J}(B) ) } 
\in \psi_{\mathbf{N}}( (I_-)^m \times (I_+)^m ) 
\subseteq I_+  \]
and a similar statement holds for $(1,0)$ or $(0,0)$, by the definition of $\mathbf{N}$.  So in this case as well we have
\[ \PR{ \mathcal{J}(A)(\mathbf{X}) = f_T(\mathbf{X}) } \in I_+. \]
This establishes the inductive hypothesis for $n+1$.

By induction, we conclude that \eqref{eq:NANDTreeInductiveHypothesis} holds for all $\NAND$ trees $T$ of any depth. 
But this immediately implies that $\conv \mathcal{C}$ supports reliable computation with advantage $\beta$. Moreover, by Remark~\ref{comment:depthFactorInrease}, it does so with constant factor overhead in depth complexity (for fixed $\kappa$).
This proves the theorem.
\end{proof}

\section{Proof of Theorem~\ref{thm:noiseThresholdsAreStrict}: The set that does not support reliable computation is closed}\label{section:proofNoiseThresholdsStrict}
In this section we prove Theorem~\ref{thm:noiseThresholdsAreStrict}, which we restate below.

\begin{thm*}[Theorem \ref{thm:noiseThresholdsAreStrict}, restated]
Let $\mathcal{C}_\epsilon$ denote a circuit model closed under composition on a gate set $\mathcal{G}$ which includes a noisy gate $g_\eps$.
Let $I \subseteq [0, 1]$ denote the set of $\epsilon$ for which $\conv \mathcal{C}_\epsilon$ does not support reliable computation (varying the noise on $g_\eps$ and keeping all other gates in $\mathcal{G}$ fixed). Then $I$ is closed.
\end{thm*}

The basic idea is to make use of Theorem~\ref{thm:AmplifierEquivToFTCC_andOverhead}.
In particular, if $\cC_\eps$ supports reliable computation, then $\conv \cC_\eps$ contains an amplifier and a noisy $\neg$ gate.  We will show that if $\eps$ is perturbed slightly to $\eps'$, then the amplifier remains an amplifier and the noisy $\neg$ gate remains a noisy $\neg$ gate.  We will conclude that $\cC_{\eps'}$ also supports reliable computation.

In order to make this intuition precise, we will need a few basic analytical lemmas, which we prove in Section~\ref{sec:analysis}.  Then we prove Theorem~\ref{thm:noiseThresholdsAreStrict} in Section~\ref{sec:thm2pf}.

\subsection{Analysis Lemmas}\label{sec:analysis}
\begin{lem}\label{lem:infimumIsContinuous}
Let $f: \R^2\rightarrow \R$ be Lipschitz continuous. Fix $a, b \in \R$. Then the function $g$ defined by
$$g(y) = \min_{x\in [a, b]}{f(x, y)}$$
is continuous.
\end{lem}
\begin{proof}
Since $f$ is Lipschitz continuous, there exists a constant $K > 0$ such that $\forall \Delta y, y \in \R$ and $x\in[a,b]$,
$$-  K|\Delta y| + f(x, y) \leq f(x, y + \Delta y)  \leq K|\Delta y| + f(x, y)$$
Therefore,
$$\left| \min_{x \in[a,b]}{f(x, y + \Delta y)}- \min_{x \in[a,b]}{f(x, y)} \right| \leq  K|\Delta y|$$
and
$$\left| g(y + \Delta y) - g(y) \right| \leq  K|\Delta y|$$
which implies the limit $\lim_{\Delta y\rightarrow 0}g(y + \Delta y) = g(y)$.
\end{proof}

\begin{lem}\label{lem:nudgeCrossing}
Let $f \in \R[x, y]$ be a bivariate polynomial such that for some $x_0 \in (a, b)$ and $y_0 \in \R$, $f(x_0, y_0) = 0$, and $(\partial_x f)(x_0, y_0) > 0$.
Then $\exists \delta > 0$ such that
for all $y'$ with $|y' - y_0| < \delta$, there is some $x' \in (a,b)$ so that
\[ f(x', y') = 0 \text{ and } (\partial_x f)(x', y') > 0. \]
\end{lem}
\begin{proof}
We will make use of the fact that $f$ and all of its derivatives are Lipschitz continuous.
Since $(\partial_x f)(x_0, y_0) > 0$, by continuity of $\partial_x f$ there exist $x_1 \in (a, x_0)$ and $x_2 \in (x_0, b)$ such that
$$f(x_1, y_0) < 0 < f(x_2, y_0)$$
and
$$\min_{x\in [x_1, x_2]} (\partial_x f)(x, y_0) > 0.$$
Then by continuity of $f$ there exists $\delta_1 > 0$ such that
for all $y'$ with $|y' - y_0| < \delta_1$,  $f(x_1, y) < 0 < f(x_2, y).$
The intermediate value theorem implies $\exists x' \in (x_1, x_2)$ such that $f(x', y') = 0$.
By Lemma~\ref{lem:infimumIsContinuous}, there exists $\delta_2 > 0$ such that
for all $y'$ with $|y' - y_0| < \delta_2$,  $\min_{x\in [x_1, x_2]} (\partial_x f)(x, y) > 0.$
Then setting $\delta = \min(\delta_1, \delta_2)$ we are done.
\end{proof}

\subsection{Proof of Theorem~\ref{thm:noiseThresholdsAreStrict}}\label{sec:thm2pf}

Before proving Theorem~\ref{thm:noiseThresholdsAreStrict}, we need one more observation.
\begin{observation}\label{prop:polynomial}
Let $\mathcal{C}_\eps$ and $g_\eps \in \mathcal{G}$ be as in the statement of Theorem~\ref{thm:noiseThresholdsAreStrict}.
Let $C = \sum_{j=1}^N p_j F_j$ be a mixture of circuits $C_j \in \cC_\eps$, so $C \in \conv \cC_\eps$.
Then viewing the amplification function $A_C(p)$ as a function of $\epsilon$ as well as $p$, $A_C$ is a polynomial in $\epsilon$ and $p$.
\end{observation}
\begin{proof}
First, we write
\begin{align*}
A_C(p) &= \mathbb{P}_{X_i \sim Ber(p)}[ C(X_1, \ldots, X_n) = 1 ] \\
&= \sum_{\mathbf{x} \in \{0,1\}^n} p^{|\mathbf{x}|}(1-p)^{n-|\mathbf{x}|} \mathbb{P}_C[ C(\mathbf{x}) = 1], 
\end{align*}
where $|\mathbf{x}|$ denotes the weight of $\mathbf{x}$,
to see that this is indeed a polynomial in $p$.  
Next, we claim that for any circuit $D \in \cC_\eps$ with at most $d$ $g_\eps$ gates, and for any fixed $\mathbf{x} \in \{0,1\}^n$, $\mathbb{P}[ D(\mathbf{x}) = 1 ]$ is a polynomial in $\eps$ of degree at most $d$.
Indeed, 
\[ \mathbb{P}[ D(\mathbf{x}) = 1 ] = \sum_{\mathbf{e} \in \{0,1\}^d } \epsilon^{|\mathbf{e}|} (1 - \epsilon)^{d - |\mathbf{e}|} \mathbb{P}[ F|_{\mathbf{e}}(\mathbf{x}) = 1 ], \]
where $F|_{\mathbf{e}}$ means the circuit $D$ where every $g_\eps$ gate has been replaced with either $g$ or $\neg g$ according to $\mathbf{e}$.
Returning to $C = \sum_{j=1}^N p_j F_j$, we have
 \[ \mathbb{P}_C[ C(\mathbf{x}) = 1] = \sum_{j=1}^N p_j \mathbb{P}[ C_j(\mathbf{x}) = 1] \]
which is again a polynomial in $\eps$ of degree at most $d$.  
Thus, $A_C(p)$ is a polynomial in $p$ of degree at most $n$ and a polynomial in $\eps$ of degree at most $d$.
\end{proof}

Finally, we prove Theorem~\ref{thm:noiseThresholdsAreStrict} using the lemmas above.
\begin{proof}[Proof of Theorem~\ref{thm:noiseThresholdsAreStrict}]

Let $\mathcal{C}_\eps$ and $g_\eps \in \mathcal{G}$ be as in the statement of the theorem.
Suppose that $\conv \mathcal{C}_{\eps_0}$ supports reliable computation.  We will show that for any $\eps$ sufficiently close to $\eps_0$, $\conv \cC_\eps$ also supports reliable computation.

By Theorem~\ref{thm:AmplifierEquivToFTCC_andOverhead}, $\conv \mathcal{C}_{\epsilon_0}$ supports reliable computation if and only if it contains an amplifier and a $\neg_{\kappa_0}$ gate for some $\kappa_0 < 1/2$.
Let $M_{\eps_0}, N_{\eps_0} \in \conv\mathcal{C}_{\epsilon_0}$ denote these mixtures of formulas, respectively, and define $M_\eps$ and $N_\eps$ in the natural way by replacing the $g_{\eps_0}$ gates in $M_{\eps_0}$ and $N_{\eps_0}$ with $g_\eps$ gates.
We will show that there exists $\delta > 0$ such that for all $\epsilon'$ satisfying $|\epsilon' - \epsilon_0| < \delta$, $M_{\eps'}$ remains an amplifier and $N_{\eps'}$ provides a noisy $\neg$ gate.

Since $M_{\eps_0}, N_{\eps_0}$ are mixtures of circuits, Observation~\ref{prop:polynomial} implies that the amplification functions $A_{M_\eps}(p)$ and $A_{N_\eps}(p)$ are polynomials of finite degree in $p, \epsilon$.
Thus, Lemma~\ref{lem:nudgeCrossing} will apply.

First, we show that $N_{\eps}$ remains a noisy $\neg$ gate for $\eps$ sufficiently close to $\eps_0$.
Notice that a stochastic map $f:\{0,1\} \to \{0,1\}$ represents $\neg_\kappa$ for $\kappa < 1/2$ if and only if
\begin{equation}\label{eq:equiv_to_not}
 A_f(1/2) = 1/2 \qquad \text{and} \qquad \frac{\partial}{\partial p} A_f(1/2) < 0. 
\end{equation}
Indeed, letting $p_0 := \mathbb{P}[f(0) = 1]$ and $p_1 := \mathbb{P}[f(1) = 1]$, a computation shows that
\begin{equation}\label{eq:equiv2}
 A_f(1/2) = \frac{1}{2}(p_0 + p_1) \qquad \text{and} \qquad \frac{\partial}{\partial p} A_f(1/2) = p_1 - p_0.
\end{equation}
What it means to represent $\neg_\kappa$ for $\kappa < 1/2$ is precisely that that $p_0 = 1 - \kappa$, $p_1 = \kappa$ for $\kappa < 1/2$, which given \eqref{eq:equiv2} is equivalent to \eqref{eq:equiv_to_not}.

By Lemma~\ref{lem:nudgeCrossing}, with $f(\eps, p)=-A_{N_\eps}(p) + 1/2$ (thought of as a bivariate polynomial in $p$ and $\eps$), there exists $\delta_1 > 0$ such that
for all $\eps'$ with $|\eps' - \eps_0| < \delta_1$, there is some $p' \in (0,1)$ so that
\[ A_{N_{\eps'}}(p') = 1/2 \text{ and } \frac{\partial}{\partial p} A_{N_{\eps'}}(p') < 0. \]
Choose $c \in \{0,1\}$ and $\lambda \in [0,1]$ so that
\[ p' = \frac{ 1 - (-1)^c \lambda}{2}, \]
and let $\mathbf{1}_c:\{0,1\} \to \{0,1\}$ denote the constant-$c$ function.
Now consider the mixture 
\[ N' = \lambda ( N_{\eps'} \circ \mathbf{1}_c ) + (1 - \lambda) N_{\eps'}, \]
where $\circ$ denotes composition.
That is, with probability $\lambda$, $N'$ behaves like a $N_{\eps'}$ gate with its input fixed to $c$, and with probability $1 - \lambda$, $N'$ behaves like a $N_{\eps'}$ gate.  It's not hard to see that
\[ A_{N'}(p) = A_{N_{\eps'}}(\lambda c + (1 - \lambda)p ), \]
which given our choice of $c$ and $\lambda$ implies that $A_{N'}(1/2) = A_{N_{\eps'}}(p')$, and in particular
\[ A_{N_{\eps'}}(1/2) = 1/2 \text{ and } \frac{\partial}{\partial p} A_{N_{\eps'}}(1/2) < 0. \]
Therefore by the equivalence above, $N'$ is equivalent to $\neg_{\kappa'}$ for some $\kappa' < 1/2$.

Next we show that $M_{\eps}$ remains an amplifier for $\eps$ close to $\eps_0$.
By definition, $M_\eps$ is an amplifier
if and only if there exists $p_0 \in (0, 1)$ such that
$$A_{M_{\eps}}(p_0) - p_0 = 0$$
and
$$ \frac{\partial}{\partial p} A_{M_\eps}(p_0) - 1 > 0. $$
By Lemma~\ref{lem:nudgeCrossing}, with $f=A_M -  p$, there exists $\delta_2 > 0$ such that
for all
$|\epsilon' - \epsilon_0| < \delta_2$, 
there is some $p' \in (0, 1)$ so that
$$ A_{M_{\eps'}}(p') = p' \text{ and } \frac{\partial}{\partial p} A_{M_\eps}(p') > 1.$$
Therefore, 
for any $\eps'$ so that $|\eps' - \eps_0| < \min(\delta_1,\delta_1)$,
$\conv \mathcal{C}_{\eps'}$ 
contains an amplifier and a gate $\neg_{\kappa'}$ for some $\kappa' < 1/2$.  By Theorem~\ref{thm:AmplifierEquivToFTCC_andOverhead} again, $\conv \cC_{\eps'}$
supports reliable computation for all $\epsilon'\in[0,1]$ such that $|\epsilon' -  \epsilon_0| < \min(\delta_1,  \delta_2)$.  This implies that the set of $\eps$ so that $\conv \mathcal{C}_{\eps}$ supports reliable computation is the intersection of an open set and the interval $[0,1]$, and hence the set of $\eps\in[0,1]$ so that $\conv \mathcal{C}_{\eps}$ does not support reliable computation is closed.
\end{proof}

\section{Proof of Theorem~\ref{thm:newNonlocalGame}: A game whose quantum value is the threshold for nontrivial communication complexity}\label{sec:newNLGame}

In this section, we prove Theorem~\ref{thm:newNonlocalGame}, which we restate below.

\begin{thm*}[Theorem \ref{thm:newNonlocalGame}, restated]
There exists a sequence of 2-player nonlocal games $G_k$ for $k \geq 1$ that satisfies properties (1-3) below, in which $S$ is any set of bipartite nonsignalling correlations closed under wirings and such that $S\supseteq Q$.
\begin{enumerate}
\item For all $k\geq 1$, $\omega_C(G_k) < \omega_Q(G_k) < 1.$
\item Fix any $k\geq 1$. If $\omega_S(G_k) > \omega_Q(G_k)$, then $S$ has trivial probabilistic communication complexity.
\item If $S$ has trivial probabilistic communication complexity, then there exists some $k\geq 1$ such that $\omega_S(G_k) > \omega_Q(G_k)$.
\end{enumerate}

\end{thm*}
First we will construct a game that satisfies properties  (2-3) but not (1). Then we will apply a technical manipulation to produce a game that satisfies (1-3).

\subsection{The Amplification Game}

We begin with a game that we can the \em Amplification Game, \em which 
satisfies properties (2-3) of Theorem~\ref{thm:newNonlocalGame}.

\begin{define}
Fix $k\geq 1$.
Let $n=2k+1$. Let $X=Y= \{0,1\}^{n}$. Let $A=B=\{0,1\}$. Let the decision predicate $D: X\times Y\times A\times B\rightarrow \{0,1\}$ be defined by $D(\mathbf{x},\mathbf{y},a,b) = \mathbf{1}\left[(a\oplus b) = \Maj_n(\mathbf{x} \oplus \mathbf{y})\right]$.
Let $\pi: X\times Y\rightarrow [0,1]$ denote the following probability distribution:
$$\pi(\mathbf{x}, \mathbf{y}) =\frac{|n-2|\mathbf{x}\oplus\mathbf{y}||}{2^{n+1}n{n-1\choose (n-1)/2} }.$$
Then the {\bf Amplification Game} denoted $\Amp_{k}$ is the nonlocal game $(A,B,X,Y,\pi,D)$.
\end{define}

Recall that in Definition~\ref{define:nonlocalGame} the supremum is restricted implicitly to correlations whose input and output alphabets match those of the game. We now introduce some new notation to make this explicit.
For a set $S$ of bipartite nonsignalling correlations, and for an integer $k\geq 1$, we will denote by $S_k$ the set of correlations in $S$ with input alphabets $X=Y=\{0,1\}^k$ and output alphabets $A=B=\{0,1\}$. We will denote by $\BBLMTU(S)_k$ the set of circuits in $\BBLMTU(S)$ which take $k$ input bits and produce one output bit.

Recall that for a set $S$ of bipartite nonsignalling correlations we have defined $\BBLMTU(S)$ as a circuit model:
$$\BBLMTU(S) = \conv\{\text{circuits from gates in }\{\BBLMTU(c) : c \in S \}\}.$$
This is done to ensure the applicability of Theorem~\ref{thm:AmplifierEquivToFTCC_andOverhead} as a black box.
Now in the case that $S\supseteq Q \supseteq C$, and $S$ is closed under wirings, by Proposition~\ref{proposition:closedUnderWiringsImpliesBBLMTUisElementwise}, $\BBLMTU(S)$ may be more simply expressed as
\begin{equation}\label{eq:BBLMTU_S_interchangeable}
\BBLMTU(S) = \{\BBLMTU(c) : c \in S\}.
\end{equation}
Further, it is a trivial consequence of Equation~\eqref{eq:BBLMTU_S_interchangeable} that in our new notation,
\begin{equation}\label{eq:BBLMTU_S_interchangeable_k}
\BBLMTU(S)_k = \{\BBLMTU(c) : c \in S_k\}.
\end{equation}
In the following proofs we will consider sets $S$ which satisfy these properties and hence, for which Equations~\eqref{eq:BBLMTU_S_interchangeable} and~\eqref{eq:BBLMTU_S_interchangeable_k} apply.

\begin{define}\label{define:cValueG}
Let $c=(A,B,X,Y,p)$ be a bipartite nonsignalling correlation, and let $G=(A,B,X,Y,\pi,D)$ be a nonlocal game. Define the $c$-value of $G$ as
$$\Omega_c(G) := \P_{x, y \sim \pi}[D(x, y, c(x, y)) = 1]$$
where the probability is also over the randomness of $c$.
\end{define}
Note that under Definition~\ref{define:cValueG}, given a game $G$ with input alphabets $X=Y=\{0,1\}^k$ and output alphabets $A=B=\{0,1\}$ we may rewrite the $S$-value (from Definition~\ref{define:nonlocalGame}) with our new notation as:
$$\omega_S(G) = \sup_{c\in S_k}{\Omega_c(G)}.$$

\begin{define}[Bernoulli Distribution over Distributed Bits]\label{define:bernoulliDistributed}
Let $\mathbf{x} \oplus \mathbf{y} \sim \text{Ber}(p)^n$ denote random variables $\mathbf{x}, \mathbf{y}\in\{0,1\}^n$ sampled from the following probability distribution:
$$\P[\mathbf{x}, \mathbf{y} | \mathbf{x} \oplus \mathbf{y} \sim \text{Ber}(p)^n] = \frac{1}{2^n}p^{|\mathbf{x}\oplus\mathbf{y}|}(1-p)^{n - |\mathbf{x}\oplus\mathbf{y}|}.$$
\end{define}

Notice that Definition~\ref{define:BBLMTU} and Definition~\ref{define:bernoulliDistributed} imply that for a bipartite nonsignalling correlation $c$ with binary outputs we have
$$A_{\BBLMTU(c)}(p) = \P_{\mathbf{z}\sim\Ber(p)^n}[\BBLMTU(c)(\mathbf{z}) = 1] = \P_{\mathbf{x}\oplus\mathbf{y}\sim\Ber(p)^n}[\XOR(c(\mathbf{x}, \mathbf{y})) = 1].$$

In the following Claim~\ref{claim:existsDistributedIdentityNotGates} and its Corollary~\ref{cor:translateLowerbdToSup}, we explain what will turn out to be an optimal classical strategy for playing $\Amp_k$.

\begin{claim}[Distributed Identity and $\neg$ Gates from Classical Correlations]\label{claim:existsDistributedIdentityNotGates}
There exist correlations $c, d\in C$ with input alphabets $X=Y=\{0,1\}$ and output alphabets $A=B=\{0,1\}$ such that
\begin{equation}\label{eq:distributedIdentityAmplificationDerivative}
A_{\BBLMTU(c)}'(1/2) = 1
\end{equation}
and 
\begin{equation}\label{eq:distributedNotGate}
A_{\BBLMTU(d)}(p) = 1-p
\end{equation}
\begin{proof}
Let $c \in S_1$ denote the local correlation in which Alice and Bob output their input bits $x$ and $y$. Then
\begin{align*}
A_{\BBLMTU(c)}(p) &= \P_{x\oplus y \sim \text{Ber}(p)}[\XOR(c(x, y)) = 1] \\
&= \P_{x_1\oplus y_1\sim\text{Ber}(p)}[\XOR(x_1, y_1)=1]\\
&= p.
\end{align*}
Therefore
$$A_{\BBLMTU(c)}'(p)= \frac{dp}{dp}= 1$$
and so $c$ satisfies Equation~\eqref{eq:distributedIdentityAmplificationDerivative}, and clearly $c\in S_1$.

Now consider the classical correlation $d$ in which Alice outputs $\neg x$ and Bob outputs $y$. Then
\begin{align*}
A_{\BBLMTU(d)}(p) &= \P_{\mathbf{x} \oplus \mathbf{y} \sim \text{Ber}(p)^n}[\XOR(c(\mathbf{x}, \mathbf{y})) = 1] \\
&= \P_{\mathbf{x} \oplus \mathbf{y} \sim \text{Ber}(p)^n}[x_1 \oplus y_1 \oplus 1 = 1] \\
&= 1-p,
\end{align*}
and so $d$ satisfies Equation~\eqref{eq:distributedNotGate}, and clearly $d\in S_1$. Note that this implies that $\BBLMTU(d)$ is a noise-free $\neg$ gate.
\end{proof}
\end{claim}
\begin{cor}\label{cor:translateLowerbdToSup}
Let $S$ be a set of bipartite nonsignalling correlations such that $S\supseteq C$. Fix any integer $k\geq 1$. Then
$$\sup_{u\in \BBLMTU(S)_k}{A_u'(1/2)} \geq 1.$$
\begin{proof}

Let $S$ be as above. Because any gate in a circuit model is itself a valid circuit, for any $c\in S$, $\BBLMTU(c) \in \BBLMTU(S)$. In particular we will fix the $c\in S$ such that $A_{\BBLMTU(c)}'(1/2) = 1$, which is guaranteed to exist by Claim~\ref{claim:existsDistributedIdentityNotGates} because $S\supseteq C$. Then since this $c$ has input and output alphabets $X=Y=A=B=\{0,1\}$, the gate $\BBLMTU(c)\in\BBLMTU(S)_1$. If $k=1$, we simply set $v= \BBLMTU(c) \in \BBLMTU(S)$, and noting that by definition of $\sup$, $A_v'(1/2)$ provides the desired lower bound:
$$\sup_{u\in \BBLMTU(S)_k}{A_u'(1/2)}\geq A_v'(1/2) = 1.$$
If $k> 1$, then since $\BBLMTU(S)$ is a circuit model we may build a circuit $b$ which takes input variables $X_1, ..., X_k$, throws away inputs $X_2, .. X_k$, and returns $\BBLMTU(c)(X_1)$. This yields
$$A_b'(1/2) = A_{\BBLMTU(c)}'(1/2),$$
and since $b \in \BBLMTU(S)_k$ this implies
$$\sup_{u\in \BBLMTU(S)_k}{A_u'(1/2)} \geq 1.$$
\end{proof}
\end{cor}

\begin{claim}\label{claim:informationTheoreticTwoPartyAmplification}
Let $S$ be a set of bipartite nonsignalling correlations closed under wirings and such that $S\supseteq C$. Then $S$ has TPCC if and only if there exists $c\in \BBLMTU(S)$ such that
\begin{equation}\label{eq:ampDistributedBits}
A_{c}'(1/2) > 1
\end{equation}
\begin{proof}

Recall that by Claim~\ref{claim:existsDistributedIdentityNotGates}, any such set $S\supseteq C$ satisfies $\neg \in \BBLMTU(S)$. Also recall that by definition, the circuit model $\BBLMTU(S)$ is closed under composition and also under convex combinations.

Now first suppose there exists $c\in \BBLMTU(S)$ satisfying Equation~\eqref{eq:ampDistributedBits}.
Then by Lemma~\ref{lem:ampViaSymmetrization}, since $\neg \in \BBLMTU(S)$, there also exists $m \in \BBLMTU(S)$ such that $m$ is an amplifier away from $1/2$. By Theorem~\ref{thm:AmplifierEquivToFTCC_andOverhead}, this implies that $\BBLMTU(S)$ supports reliable computation, which implies that $S$ has TPCC by Definition~\ref{define:TPCC}.

Conversely, suppose that $S$ has TPCC. By Definition~\ref{define:TPCC}, this implies that $\BBLMTU(S)$ supports reliable computation. By Theorem~\ref{thm:AmplifierEquivToFTCC_andOverhead}, this implies that there exists $u \in \BBLMTU(S)$ such that $u$ is an amplifier. By Lemma~\ref{lem:existsSelfDualAmplfiier}, since $\neg \in \BBLMTU(S)$, there also exists $c\in \BBLMTU(S)$ such that $c$ is an amplifier away from $1/2$, and so by definition $c$ satisfies Equation~\eqref{eq:ampDistributedBits}.
\end{proof}
\end{claim}

\begin{claim}\label{cl:relateToAmp}
Let $k\geq 1$, let $n=2k+1$, and let $c = (A,B,X,Y,p)$ be a bipartite nonsignalling correlation such that $A=B=\{0,1\}^n$ and $X=Y=\{0,1\}$.
Then the following holds:
\begin{equation}\label{eq:relationDerivativeToAmplificationGameValue}
A_{\BBLMTU(c)}'(1/2) = \frac{n{n-1 \choose (n-1)/2}}{2^{n-1}} \times[2 \Omega_c(\Amp_k) - 1].
\end{equation}
\end{claim}
Less formally, Claim~\ref{cl:relateToAmp} relates Alice and Bob's win probability for the game $\Amp_k$ played using correlation $c$ (on the right hand side of \eqref{eq:relationDerivativeToAmplificationGameValue}), to their ability to amplify $n$ distributed bits (away from $1/2$) by feeding them through $c$ (on the left hand side of \eqref{eq:relationDerivativeToAmplificationGameValue}).
\begin{proof}

We begin with the expression on the left:

\begin{align*}
A_{\BBLMTU(c)}'(1/2)&=\left.\frac{d}{dp}\P_{\mathbf{x} \oplus \mathbf{y} \sim \text{Ber}(p)^n}[\XOR(c(\mathbf{x}, \mathbf{y}))  = 1]\right|_{p=1/2} \\
&=\left.\frac{d}{dp} \sum_{\mathbf{x}, \mathbf{y}\in \{0,1\}^n}{\frac{1}{2^n}p^{|\mathbf{x}\oplus\mathbf{y}|}(1-p)^{n-|\mathbf{x}\oplus\mathbf{y}|} \P[\XOR(c(\mathbf{x}, \mathbf{y})) = 1]} \right|_{p=1/2}\\
&=\left.\frac{d}{dp} \sum_{\mathbf{z}\in \{0,1\}^n}{\frac{1}{2^n}p^{|\mathbf{z}|}(1-p)^{n-|\mathbf{z}|} \sum_{\mathbf{x}\in \{0,1\}^n} \P[\XOR(c(\mathbf{x}, \mathbf{z} \oplus \mathbf{x})) = 1]} \right|_{p=1/2}\\
&= \left.\sum_{\mathbf{z}\in \{0,1\}^n}{\frac{1}{2^n}(|\mathbf{z}|-pn)p^{|\mathbf{z}|-1}(1-p)^{n-|\mathbf{z}|-1} \sum_{\mathbf{x}\in \{0,1\}^n} \P[\XOR(c(\mathbf{x}, \mathbf{z} \oplus \mathbf{x})) = 1]} \right|_{p=1/2}\\
&= \sum_{\mathbf{z}\in \{0,1\}^n}{\frac{2|\mathbf{z}|-n}{2^{n-1}} \sum_{\mathbf{x}\in \{0,1\}^n}\frac{1}{2^n} \P[\XOR(c(\mathbf{x}, \mathbf{z} \oplus \mathbf{x})) = 1]}\\
&= \sum_{\mathbf{z}\in \{0,1\}^n}{\frac{2|\mathbf{z}|-n}{2^{n-1}} \P_{\mathbf{x}\sim\{0,1\}^n}[\XOR(c(\mathbf{x}, \mathbf{z} \oplus \mathbf{x})) = 1]}  \\
&= \sum_{\substack{\mathbf{z} \in \{0,1\}^n\\ |\mathbf{z}|<n/2}}{\frac{2|\mathbf{z}| - n }{2^{n - 1}}\left(1-\P_{\mathbf{x}\sim\{0,1\}^n}[\XOR(c(\mathbf{x}, \mathbf{z} \oplus \mathbf{x})) = \Maj_n(\mathbf{z})] \right)} \\
& \,\,\,\,\,\,\,\,\,\,\,\,\,\,\,\,\,\,\,\,\,\,\,\,\,\,\,\,\,\,\,\,\,\,\,\,\,\,\,\,\,\,\,\,\,\,\,\,+ \sum_{\substack{\mathbf{z} \in \{0,1\}^n\\ |\mathbf{z}|>n/2}}{\frac{2|\mathbf{z}| - n }{2^{n - 1}}\P_{\mathbf{x}\sim\{0,1\}^n}[\XOR(c(\mathbf{x}, \mathbf{z} \oplus \mathbf{x})) = \Maj_n(\mathbf{z})]} \\
&= \sum_{\substack{\mathbf{z} \in \{0,1\}^n\\ |\mathbf{z}|<n/2}}{\frac{2|\mathbf{z}| - n }{2^{n - 1}}} + \sum_{\mathbf{z}\in \{0,1\}^n}{\frac{|n-2|\mathbf{z}||}{2^{n-1}}\P_{\mathbf{x}\sim\{0,1\}^n}[\XOR(c(\mathbf{x}, \mathbf{z} \oplus \mathbf{x})) = \Maj_n(\mathbf{z})] }  \\
&= \frac{1}{2^{n - 1}} \left(\sum_{\mathbf{z}\in \{0,1\}^n}{|n-2|\mathbf{z}||\P_{\mathbf{x}\sim\{0,1\}^n}[\XOR(c(\mathbf{x}, \mathbf{z} \oplus \mathbf{x})) = \Maj_n(\mathbf{z})] } + \sum_{\substack{\mathbf{z} \in \{0,1\}^n\\ |\mathbf{z}|<n/2}}{(2|\mathbf{z}| - n )}\right),
\end{align*}
and noting that $\sum_{\substack{\mathbf{z} \in \{0,1\}^n\\ |\mathbf{z}|<n/2}}{(2|\mathbf{z}| - n )} = -n{n-1\choose (n-1)/2}$, this may be expressed as
\begin{align*}
&= \frac{n{n-1 \choose (n-1)/2}}{2^{n - 1}} \left( 2 \sum_{\mathbf{z}\in \{0,1\}^n}{\frac{|n-2|\mathbf{z}||}{2n{n-1\choose (n-1)/2} }\P_{\mathbf{x}\sim\{0,1\}^n}[\XOR(c(\mathbf{x}, \mathbf{z} \oplus \mathbf{x})) = \Maj_n(\mathbf{z})] } - 1 \right)\\
&= \frac{n{n-1 \choose (n-1)/2}}{2^{n - 1}} \left( 2 \sum_{\mathbf{x}, \mathbf{y}\in \{0,1\}^n}{\frac{|n-2|\mathbf{x}\oplus\mathbf{y}||}{2^{n+1}n{n-1\choose (n-1)/2} }\P[\XOR(c(\mathbf{x}, \mathbf{y})) = \Maj_n(\mathbf{x}\oplus\mathbf{y})] } - 1 \right)\\
&= \frac{n{n-1 \choose (n-1)/2}}{2^{n - 1}} \left( 2 \sum_{\mathbf{x}, \mathbf{y}\in \{0,1\}^n}{\pi(\mathbf{x},\mathbf{y})\P[D(\mathbf{x}, \mathbf{y}, c(\mathbf{x}, \mathbf{y}))=1]} - 1 \right)\\
&= \frac{n{n-1 \choose (n-1)/2}}{2^{n - 1}} \left( 2 \Omega_c(\Amp_k) - 1 \right).
\end{align*}
\end{proof}

\begin{claim}\label{claim:superquantumCollapse}
The game $\Amp_k$ satisfies conditions (2) and (3) of Theorem~\ref{thm:newNonlocalGame}.
\begin{proof}
Let $S$ be any set of bipartite nonsignalling correlations such that $S\supseteq Q$ and $S$ is closed under wirings. (Note that this includes the case where $S=Q$, and this case will be useful below.)

Now rearranging Equation~\eqref{eq:relationDerivativeToAmplificationGameValue}, we have that
\begin{align}
\omega_S(\Amp_k) = \sup_{c\in S_k}{\Omega_c(\Amp_k)} &= \frac{1}{2} + \frac{2^{n-2}}{n{n-1\choose (n-1)/2}} \sup_{c\in S_k}{A_{\BBLMTU(c)}'(1/2)}\nonumber \\
&= \frac{1}{2} + \frac{2^{n-2}}{n{n-1\choose (n-1)/2}} \sup_{u\in \BBLMTU(S)_k}{A_{u}'(1/2)}. \label{eq:supValueFromSupAmplificationDeriv}
\end{align}
in which we have used Equation~\eqref{eq:BBLMTU_S_interchangeable_k}.
Since $S\supseteq C$, Corollary~\ref{cor:translateLowerbdToSup} applies to give the lower bound
$$\sup_{u\in \BBLMTU(S)}A_u'(1/2) \geq 1,$$
which implies by \eqref{eq:supValueFromSupAmplificationDeriv} that
\begin{equation}\label{eq:lowerBdAllPostClassicalModelsAmpkValue}
\omega_S(\Amp_k) \geq \frac{1}{2} + \frac{2^{n-2}}{n{n-1\choose (n-1)/2}}.
\end{equation}

Now we would like to compute the quantum value $\omega_Q(\Amp_k)$.
The arguments above apply to the specific case where $S=Q$, and so \eqref{eq:lowerBdAllPostClassicalModelsAmpkValue} gives the lower bound
$$\omega_Q(\Amp_k) \geq \frac{1}{2} + \frac{2^{n-2}}{n{n-1\choose (n-1)/2}}.$$
On the other hand, suppose for contradiction that $\omega_Q(\Amp_k) > \frac{1}{2} + \frac{2^{n-2}}{n{n-1\choose (n-1)/2}}$. Then there exists $c\in Q$ such that $\Omega_c(\Amp_k) > \frac{1}{2} + \frac{2^{n-2}}{n{n-1\choose (n-1)/2}}.$ Substituting this $\Omega_c$ into Equation~\eqref{eq:relationDerivativeToAmplificationGameValue} immediately gives
$$A_{\BBLMTU(c)}'(1/2) > 1.$$
Since $\BBLMTU(c)\in \BBLMTU(Q)$, Claim~\ref{claim:informationTheoreticTwoPartyAmplification} implies that $Q$ has trivial probabilistic communication complexity. However, $Q$ does not have trivial probabilistic communication complexity.\footnote{This is well-known; \cite{innerProductCommComplexity} provides a specific proof for the inner product function.}
This provides a contradiction, so we must have the upper bound $\omega_Q(\Amp_k) \leq \frac{1}{2} + \frac{2^{n-2}}{n{n-1\choose (n-1)/2}}$. Combining this upper bound with our lower bound, we have computed the exact quantum value\footnote{Since $\Amp_k$ is a nonlocal computation game, we could have also computed $\omega_Q(\Amp_k)$ using the technique in \cite{noQAdvantage}, in which it was shown that quantum mechanics gives no advantage for nonlocal computation.}:
\begin{equation}\label{eq:quantumValueAmpk}
\omega_Q(\Amp_k) = \frac{1}{2} + \frac{2^{n-2}}{n{n-1\choose (n-1)/2}}.
\end{equation}

Now that we have computed the quantum value $\omega_Q(\Amp_k)$ we return to the case of a general set $S\supseteq Q$ which satisfies the conditions of Theorem~\ref{thm:newNonlocalGame}. We will show that the sequence of games $\Amp_k$ satisfies Property 2 and Property 3 of Theorem~\ref{thm:newNonlocalGame}.

{\bf Proof that $\Amp_k$ satisfies Property 2 of Theorem~\ref{thm:newNonlocalGame}:} Fix any $k\geq 1$, and let $n=2k+1$. Suppose $\omega_S(\Amp_k) > \omega_Q(\Amp_k)$. Then by Equation~\eqref{eq:quantumValueAmpk}, this implies
$$\omega_S(\Amp_k) > \frac{1}{2} + \frac{2^{n-2}}{n{n-1\choose (n-1)/2}}.$$
Then by definition there exists $c\in S$ such that
$$\Omega_c(\Amp_k) > \frac{1}{2} + \frac{2^{n-2}}{n{n-1\choose (n-1)/2}}.$$
Substituting this $\Omega_c$ into Equation~\eqref{eq:relationDerivativeToAmplificationGameValue} immediately gives
$$A_{\BBLMTU(c)}'(1/2) > 1.$$
Since $\BBLMTU(c)\in \BBLMTU(S)$, Claim~\ref{claim:informationTheoreticTwoPartyAmplification} implies that $S$ has trivial probabilistic communication complexity.

{\bf Proof that $\Amp_k$ satisfies Property 3 of Theorem~\ref{thm:newNonlocalGame}:} 
Conversely, suppose that $S$ has trivial probabilistic communication complexity. Then by Claim~\ref{claim:informationTheoreticTwoPartyAmplification}, there exists $k\geq 1$, $c\in \BBLMTU(S)_k$ such that 
\begin{equation}\label{eq:fixedcDerivGreater1fromTPCC}
A_c'(1/2) > 1.
\end{equation}
Fix this $k$ and $c$.
Equation~\eqref{eq:fixedcDerivGreater1fromTPCC} implies that
$$\sup_{c\in \BBLMTU(S)_k}{A_c'(1/2)} > 1,$$
which by \eqref{eq:supValueFromSupAmplificationDeriv} gives
$$\omega_S(\Amp_k) > \frac{1}{2} + \frac{2^{n-2}}{n{n-1\choose (n-1)/2}}.$$
Then by \eqref{eq:quantumValueAmpk}, for this $k$,
$$\omega_S(\Amp_k) > \omega_Q(\Amp_k).$$

\end{proof}
\end{claim}

It is not hard to see that the approach we have just taken to calculate $\omega_Q(\Amp_k)$ also works out for calculating $\omega_C(\Amp_k)$, and that it turns out to be the case that
$$\omega_C(\Amp_k) = \omega_Q(\Amp_k).$$
In fact, this equality also follows directly from a previous result by \cite{noQAdvantage}. Therefore it turns out that $\Amp_k$ does not satisfy property (1) of Theorem~\ref{thm:newNonlocalGame}, which is that there should be some quantum advantage over classical correlations for playing the game.

To complete the proof of Theorem~\ref{thm:newNonlocalGame}, we will use the following Lemma~\ref{lem:addQAdvantage}, which is proven in Section~\ref{ssection:addQuantumAdvantagetoAnyNonlocalGame}:
\begin{lem}\label{lem:addQAdvantage}
Let $G$ denote a 2-player nonlocal game. There exists another game $G'$ which satisfies the following:
\begin{enumerate}
\item For any set of correlations $S\supset Q$ closed under composition and restriction, 
\[ \omega_S(G') > \omega_Q(G') \Leftrightarrow \omega_{S}(G) > \omega_Q(G). \]
\item $\omega_Q(G') > \omega_C(G')$.
\end{enumerate}

\end{lem}
Less formally, this says that there is a game $G'$ with nonzero quantum advantage, such that correlations giving superquantum advantage at the game $G'$ can be used along with classical correlations to give superquantum advantage at playing the original game $G$.

We now complete the proof of Theorem~\ref{thm:newNonlocalGame} using Lemma~\ref{lem:addQAdvantage}.

\begin{proof}[Proof of Theorem~\ref{thm:newNonlocalGame}]

For $k\geq 1$, let $\text{MagicAmp}_k$ be the game defined by modifying $\Amp_k$ according to Lemma~\ref{lem:addQAdvantage}.
Let $S$ be any set of bipartite nonsignalling correlations such that $S\supseteq Q$ and $S$ is closed under wirings.

{\bf Proof of Property 1:} By property (2) of Lemma~\ref{lem:addQAdvantage}, for all $k\geq 1$, $\text{MagicAmp}_k$ satisfies
$$\omega_Q(\text{MagicAmp}_k) > \omega_C(\text{MagicAmp}_k).$$
Therefore the sequence of games $G_k = \text{MagicAmp}_k$ satisfies condition (1) of Theorem~\ref{thm:newNonlocalGame}.

{\bf Proof of Property 2:} Fix any $k\geq 1$. Suppose that $\omega_S(\text{MagicAmp}_k) > \omega_Q(\text{MagicAmp}_k)$. 
We will now use property (1) of Lemma~\ref{lem:addQAdvantage}, which implies that
\begin{equation}\label{eq:MagicAmpSuperQiffAmpSuperQ}
\omega_S(\text{MagicAmp}_k) > \omega_Q(\text{MagicAmp}_k) \Leftrightarrow \omega_S(\Amp_k) > \omega_Q(\Amp_k).
\end{equation}
By Equation~\eqref{eq:MagicAmpSuperQiffAmpSuperQ}, $\omega_S(\Amp_k) > \omega_Q(\Amp_k)$. By Claim~\ref{claim:superquantumCollapse}, $\Amp_k$ satisfies property (2) of Theorem~\ref{thm:newNonlocalGame} and therefore $S$ has trivial probabilistic communication complexity. Therefore the sequence of games $G_k = \text{MagicAmp}_k$ satisfies property (2) of Theorem~\ref{thm:newNonlocalGame}.

{\bf Proof of Property 3:} Finally, suppose that $S$ has trivial probabilistic communication complexity. By Claim~\ref{claim:superquantumCollapse}, $\Amp_k$ satisfies property (3) of Theorem~\ref{thm:newNonlocalGame} and therefore there exists some $k\geq 1$ such that $\omega_S(\Amp_k) > \omega_Q(\Amp_k)$. By Equation~\eqref{eq:MagicAmpSuperQiffAmpSuperQ}, for this $k$ we have that $\omega_S(\text{MagicAmp}_k) > \omega_Q(\text{MagicAmp}_k)$. Therefore the sequence of games $G_k = \text{MagicAmp}_k$ satisfies property (3) of Theorem~\ref{thm:newNonlocalGame}.

\end{proof}

\if{false}
\subsection{Miscellaneous Discouraging Calculations}
\begin{thm}\label{thm:QValueMajNL}
\begin{equation}\label{eq:QValueMajNL}
\omega_Q(\text{MajNL}_{2k+1}) = \frac{1}{2} + \frac{1}{2^{2k+1}}{2k\choose k}
\end{equation}
\end{thm}
\begin{thm}\label{thm:MinMajorityAmpValue}
A $\Maj_{2k+1,\eps}$ gate serves as an amplifier away from $1/2$ for all
\begin{equation}\label{eq:MinMajorityAmpValue}
\epsilon < \frac{1}{2} - \frac{2^{2k-1}}{(2k+1){2k\choose k}}.
\end{equation}
\end{thm}
Notice that \eqref{eq:QValueMajNL} and \eqref{eq:MinMajorityAmpValue} are separated by a nonzero gap for all finite $k$. The approach of \cite{Brassard} is to build a circuit of $\text{PR}_\eps$ boxes which implements distributed $\Maj_{\delta(\eps)}$. The value $\delta(\eps)$ is below the maximum error tolerable for amplification \eqref{eq:MinMajorityAmpValue} for all $\eps < \frac{1}{2}-\frac{1}{\sqrt{6}}$.
This exact argument can not possibly work out at the quantum threshold.
Any circuit (wiring) of $\text{PR}_{\eps}$ boxes which implements a $\Maj_{\delta(\eps)}$ gate will necessarily have by Theorem~\ref{thm:QValueMajNL} $\delta\left(\frac{1}{2} - \frac{1}{\sqrt{8}} \right) > \frac{1}{2} - \frac{2^{2k-1}}{(2k+1){2k\choose k}}$.
By the finite degree in $\epsilon$ of the amplification function, this strict inequality can not be reversed for all arbitrarily small shifts of $\epsilon$
\fi

\subsection{Adding Quantum Advantage to a Nonlocal Game: The Magic Amplification Game}\label{ssection:addQuantumAdvantagetoAnyNonlocalGame}
By Claim~\ref{claim:superquantumCollapse}, the Amplification Game provides a tight limit on quantum nonlocality in any world in which communication complexity is nontrivial, satisfying a key motivation of \cite{implausibleConsequencesSuperstrongNonlocality,Brassard}.
On the other hand, part of the motivation that these authors had for investigating the CHSH game was that this game has some quantum advantage but no perfect quantum strategy. Intuitively, why would nature be strictly more nonlocal than classical strategies, but stop short of still greater nonlocality? In this respect, the Amplification Game falls short of the CHSH game because it has no quantum advantage. 
This turns out to be easy to rectify by building a new game from $\Amp_k$, the trivial game $G_T$ (Definition~\ref{define:trivialNLGame}), and a {pseudo-telepathy} game $M$.
A pseudo-telepathy game is a nonlocal game for which there is no perfect classical strategy, but there exists a perfect quantum strategy.
The idea is that if we require Alice and Bob to play both $M$ and $\Amp_k$ at the same time, then the resulting game could inherit some quantum advantage from $M$ and the special properties from $\Amp_k$ pertaining to communication complexity. To make this idea precise requires some fine-tuning because we can merely bound the value of the new nonlocal game.

There are many pseudo-telepathy games, and \cite{quantumPseudoTelepathy} provide a nice catalogue.
Although for our purposes it will suffice that pseudo-telepathy games exist, for completeness we describe one pseudo-telepathy game in detail here: the Mermin-Peres magic square game, \cite{merminMagicSquare, peresMagicSquare}\footnote{also see \cite{mermin1993hidden, aravind2002, aravind2002simple, aravind2004quantum}}, which we denote $M$. The inputs are $x, y \in \{1, 2, 3\}$. Alice and Bob must return bit vectors $\mathbf{a}$ and $\mathbf{b}$ respectively with $\mathbf{a}, \mathbf{b}\in \{0,1\}^2$. The bits returned must satisfy a linear system of equations mod 2, which depends upon the inputs $x, y$. In particular, we say that a matrix $A \in \F_2^{3\times 3}$ is a ``magic square" if each column sums to $1$ and each row sums to $0$. A magic square cannot exist, since summing all the entries gives $0=1$. Alice and Bob must return the first two bits of row $x$ and column $y$, respectively, such that the verifier is convinced that these bits came from a magic square. More precisely, there must exist some completion of the partial row and column $\mathbf{a}$ and $\mathbf{b}$ such that the intersecting bit agrees and the parity constraints are satisfied. Since no magic square exists, it holds that $\omega_C(M) = 8/9$. Surprisingly, $\omega_Q(M) = 1$, making $M$ a pseudo-telepathy game.

We begin with a two basic observations about the $S$-value of conjunctions and mixtures of two games $G_1, G_2$, in the case that $S$ is closed under wirings. Note that there exists a wiring of any correlation used to play $G_1\land G_2$ which can be used to play $G_1$ with at least the same win probability, giving an upper bound. Further, we may use two correlations which play $G_1$ and $G_2$ independently to play the conjunction $G_1\land G_2$. This gives Observation~\ref{observation:conjunctionValue}:
\begin{observation}\label{observation:conjunctionValue}
Let $S\supseteq C$ be a set of bipartite nonsignalling correlations closed under wirings. Then
$$\omega_S(G_1) \omega_S(G_2) \leq \omega_S(G_1\land G_2) \leq \min\{ \omega_S(G_1), \omega_S(G_2) \},$$
in which $G_1\land G_2$ is the conjunction (Definition~\ref{define:NLConjunction}).
\end{observation}
Next consider the mixture $(1-q)G_1 + qG_2$. Recall that in the mixture, the players are told each round which game they must play.
Any correlation which plays the mixture may thus be restricted to play either $G_1$ or $G_2$, allowing the supremum in Definition~\ref{define:nonlocalGame} to be rewritten to yield Observation~\ref{observation:mixtureValue}:
\begin{observation}\label{observation:mixtureValue}
Let $S\supseteq C$ be a set of bipartite nonsignalling correlations closed under wirings. Then
$$\omega_S((1-q)G_1 + qG_2) = (1-q)\omega_S(G_1) + q\omega_S(G_2),$$
in which $(1-q)G_1 + qG_2$ is the mixture (Definition~\ref{define:NLMixture}).
\end{observation}

Now we may prove Lemma~\ref{lem:addQAdvantage}.
\begin{lem*}[Lemma~\ref{lem:addQAdvantage}, restated]
Let $G$ denote a 2-player nonlocal game. There exists another game $G'$ which satisfies the following:
\begin{enumerate}
\item For any set of correlations $S\supset Q$ closed under composition and restriction, 
\[ \omega_S(G') > \omega_Q(G') \Leftrightarrow \omega_{S}(G) > \omega_Q(G). \]
\item $\omega_Q(G') > \omega_C(G')$.
\end{enumerate}

\begin{proof}
Fix $G = (X,Y,A,B,\pi,D)$.
Let $M = (A_M, B_M, X_M, Y_M, \pi_M, D_M)$ be a pseudo-telepathy game, which by definition has $\omega_C(M) <1$ and $\omega_Q(M) = 1$.
For $q\in (0,1]$ define the new game:
$$G_q := (q G + (1-q) G_T) \land M.$$
Let $\ell = (A\sqcup\{\perp\}, B\sqcup\{\perp\},(A \sqcup \{\perp\})  \times A_M, (B \sqcup \{\perp\}) \times B_M,f) \in C$ be the local correlation which will simply discard the outputs for the game $M$:
$$f(a,b, (a_q, a_M),(b_q, b_M)) = \begin{cases}
1 & a=a_q\text{ and }b=b_q\\
0 & \text{otherwise}
\end{cases}.$$

We will now demonstrate that for any $q\in (0,1]$, $G'=G_q$ along satisfies condition (1) of Lemma~\ref{lem:addQAdvantage}.
Fix any $q\in (0,1]$ and let $S$ be any set of correlations such that $S\supseteq Q$, $S$ is closed under wirings, and $\omega_S(G') > \omega_Q(G')$. Then since $S\supseteq Q$, $\omega_S(M) \geq \omega_Q(M) = 1$ and the upper and lower bounds from Observation~\ref{observation:conjunctionValue} both become tight, giving
\begin{equation}\label{eq:expandObsv1}
\omega_S(G_q) = \omega_S(qG+(1-q)G_T).
\end{equation}
By Observation~\ref{observation:mixtureValue} and since $\omega_S(G_T) =1$ the right side may be expressed as
\begin{equation}\label{eq:expandObsv2}
q\omega_S(G) + (1-q)\omega_S(G_T) = q\omega_S(G) + 1-q.
\end{equation}
Combining \eqref{eq:expandObsv1} and \eqref{eq:expandObsv2} we have
$$\omega_S(G_q) = q\omega_S(G) + 1-q.$$
Since $Q$ is itself closed under wirings, this applies to $Q$ as well, and we have 
\begin{equation}\label{eq:qValueGqGeneral}
\omega_Q(G_q) = q\omega_Q(G) + 1-q.
\end{equation}
For all $q>0$, this implies the following equivalence:
$$\omega_S(G_q) > \omega_Q(G_q) \Leftrightarrow q\omega_S(G) + 1-q > q\omega_Q(G) + 1-q \Leftrightarrow \omega_S(G) > \omega_Q(G),$$
which demonstrates that $G_q$ satisfies condition (1) of Lemma~\ref{lem:addQAdvantage}.

Next we show that there exists a choice of $q\in (0,1]$ so that $G_q$ satisfies condition (2) of Lemma~\ref{lem:addQAdvantage}, which requires that
$$\omega_C(G_q) < \omega_Q(G_q).$$
The upper bound in Observation~\ref{observation:conjunctionValue} gives
$$\omega_C(G_q) \leq \omega_C(M).$$
There are two cases:
either $\omega_C(G) > \omega_C(M)$ or $\omega_C(G)\leq \omega_C(M)$.

{\bf Case 1:} If $\omega_C(G) > \omega_C(M)$, then we set $q=1$ and find that
$$\omega_C(G_{q=1}) \leq \min\{\omega_C(G), \omega_C(M)\} = \omega_C(M) <  \omega_C(G) \leq \omega_Q(G) = \omega_Q(G_{q=1}),$$
in which we have used the upper bound of Observation~\ref{observation:conjunctionValue}, then the fact that $C\subset Q$, and finally Equation~\eqref{eq:qValueGqGeneral} for the case of $q=1$. Therefore
$$\omega_C(G_{q=1}) < \omega_Q(G_{q=1})$$
as desired.

{\bf Case 2:} Otherwise $\omega_C(G)\leq \omega_C(M)$. Then we may define the constant
$$q_0 := \frac{1-\omega_C(M)}{1-\omega_C(G)},$$
and because $\omega_C(G) \leq \omega_C(M) < 1$, $q_0 \in (0,1]$.
Notice that for any $q \in (0, q_0)$, by Observation~\ref{observation:mixtureValue} and since $\omega_C(G_T)=1$,
$$\omega_C(qG + (1-q)G_T) = 1 - q(1-\omega_C(G)) > 1 - q_0(1-\omega_C(G)) = \omega_C(M).$$
Combining with the upper bound from Observation~\ref{observation:conjunctionValue} we have
\begin{equation}\label{eq:upperBdOmegaCGq}
\omega_C(G_q) \leq \min\{ \omega_C(M), \omega_C(qG + (1-q)G_T) \} = \omega_C(M).
\end{equation}
Meanwhile we can lower bound the quantum value using Observation~\ref{observation:conjunctionValue}  combined with Observation~\ref{observation:mixtureValue} and the fact that $\omega_Q(M)=1$:
$$\omega_Q(G_q) \geq \omega_Q(qG + (1-q)G_T) = 1 - q(1-\omega_Q(G)) > 1- q_0 (1-\omega_Q(G)).$$
Since $Q\supseteq C$, $1- \omega_Q(G)  \leq 1 -\omega_C(G)$ which implies
$$\frac{1-\omega_Q(G)}{1-\omega_C(G)} \leq 1$$
$$\Rightarrow q_0 (1-\omega_Q(G)) =  \frac{1-\omega_C(M)}{1-\omega_C(G)}(1-\omega_Q(G)) \leq 1-\omega_C(M)$$
$$\Rightarrow 1-q_0(1-\omega_Q(G)) \geq 1 - (1-\omega_C(M)) = \omega_C(M).$$
which we combine with Equation~\eqref{eq:upperBdOmegaCGq}
to give
$$\omega_Q(G_q) > \omega_C(G_q)$$
as desired.

Therefore there exists a choice of $q\in (0,1]$ so that $G_q$ satisfies condition (2) of Lemma~\ref{lem:addQAdvantage}. Recall that for any such $q$, $G_q$ also satisfies condition (1). Therefore Lemma~\ref{lem:addQAdvantage} is proven with $G' = G_q$.

\end{proof}
\end{lem*}

\section{Conclusion and Future Work}\label{sec:conclusion}
We investigated the extent to which the axiom ``communication complexity is nontrivial'' can explain the quantum value of nonlocal games, along the way developing new results about reliable classical computation with noisy gates.
On the quantum side, we have shown that there is a game $G$ so that $\omega_Q(G)$ is precisely explained by the axiom ``communication complexity is nontrivial''; and we have provided evidence that the approach of \cite{Brassard} cannot show a similar statement for the CHSH game.  On the reliable computation side, we have shown that the class $\mathcal{F}_{\eps,\tau}$ of formulas made from $\land_\eps$ and $\oplus_\tau$ gates does not support reliable computation for any $\eps \in (1/6, 5/6)$ and $\tau\in (0,1)$.  Assuming  Conjecture~\ref{conj:AppliesToCircuits}, and combined with previous work of \cite{Brassard}, this implies that the noise threshold for $\cC_\eps$ is exactly $1/6$. To prove our results, we have developed new tools for reasoning about fault-tolerant computation with asymmetric noise, including formalizing the tight relationship between amplifiers and fault-tolerant computation.
 
We conclude with a few open questions and directions for future work.
\begin{enumerate}
\item \textbf{Establishing that ``communication complexity is nontrivial'' is \em not \em enough to explain $\omega_Q(CHSH)$.}  We have shown that the approach of Brassard \etal in \cite{Brassard} likely cannot be pushed further.
However, this does not rule out all approaches; in particular, it could be that there is a way to use the CHSH correlation in way other than to create noisy AND gates.
It would be interesting to rule out any approach (or to find an approach that works!)

\item \textbf{An analogous result for circuits.}
As with previous results about formulas (eg, \cite{maxGateNoise,FalkUngerMinorBoundImprovements}),
we conjecture (Conjecture~\ref{conj:AppliesToCircuits}) that the same threshold of $\eps=1/6$ that we have proved for formulas also holds for circuits.  The assumption of formulas only comes in in the proof of Lemma~\ref{claim:mainClaim}, where we use the fact that the noise in the subtrees beneath two different inputs is independent.  It would be interesting to see if this assumption could be relaxed by investigating the nature of the dependencies which arise in general circuits.

\item \textbf{Results for general asymmetric gate noise.}  We have studied the gate set $\{\land_\eps, \oplus_\tau\}$ for the case that $\tau$ is arbitrarily small or $\tau = 0$. 
However, it remains open for  general $\tau > 0$.  The parameter regime where $(\eps,\tau) \in \left(\frac{3 - \sqrt{7}}{4}, 1/6\right) \times \left(0, \frac{3 - \sqrt{7}}{4}\right)$ is of particular interest.  Indeed, as shown in Figure~\ref{fig:params}, we understand what happens on the boundaries of this region, but do not know what happens in the interior.  

\item \textbf{Relationship to quantum fault-tolerant computation.}
Theorem~\ref{thm:AmplifierEquivToFTCC_andOverhead} may be viewed as an upper bound on the overhead required for reliable computation: regardless of the noise rate, there will only ever be a constant blow-up in the depth of the circuit when using noisy gates to compute reliably.

It is interesting to consider the analogous question in quantum computation, where realistic gate implementations will have significant gate noise necessitating fault-tolerance techniques in order to scale. Despite the resulting enormous amount of work on fault-tolerant quantum computation, it is not known whether a corresponding statement about constant blow-up in depth applies in the quantum setting.

It is possible that there are multiple distinct thresholds in the quantum case: one noise threshold below which quantum circuits can reliably compute with minimal overhead, and a higher noise threshold below which quantum circuits can reliably compute at all. Indeed, a trivial version of the statement is almost certainly true; one limit of maximally asymmetric gate noise simply turns a quantum computer into a noiseless classical computer. Such a computer could simulate quantum computation but only with exponential overhead as far as we know. More interestingly, the possibility of multiple thresholds is supported, for example, by the work of \cite{virmani2005classical} which shows that circuits of sufficiently noisy quantum gates are efficiently simulatable by a classical circuit. It is also consistent with the best known constructions for fault-tolerant quantum computing.\footnote{The usual model allows noiseless classical computation on the side to perform syndrome calculations for error correction.} Specifically, in order to obtain fault-tolerance with a constant factor overhead in quantum circuit depth, the best current construction has a threshold that is orders of magnitude worse than thresholds from proposals with super-constant overhead~\cite{constantOverheadQuantum,DBLP:journals/qic/Gottesman14} .

It may well be that a quantum version of Theorem~\ref{thm:AmplifierEquivToFTCC_andOverhead}  exists, meaning there remain major improvements to be found in quantum fault tolerance that will achieve constant depth overhead at high noise rates. That would be an exciting and likely technologically important discovery. On the other hand, the story may simply be more complicated in the quantum setting, with multiple thresholds depending on the scaling of the overhead cost, which would be a sharp contrast to what happens for reliable classical computation.
\end{enumerate}

\section*{Acknowledgements} We thank Li-Yang Tan for helpful discussions. We thank the Stanford Research Computing Center and Google for providing computing resources.  We thank anonymous reviewers for helpful comments and suggestions. We also thank Ryuhei Mori for helping correct an error in a previous version of this manuscript by pointing out that formulas do not remain closed under composition.

\begin{appendices}
\section{Proof of Proposition~\ref{prop:thepoint}: Connection between nontrivial communication complexity and reliable computation}\label{app:qm}
In this section we outline the proof of Proposition~\ref{prop:thepoint}, which we repeat below.
\begin{proposition*}[Proposition \ref{prop:thepoint}, restated]
Suppose that $C \subseteq S \subseteq NS$ and that $S$ is closed under wirings.  Then
 $S$ causes probabilistic communication complexity to become trivial (in the sense described in Section~\ref{sec:overview_rel}) if and only if $\BBLMTU(S)$ supports reliable computation. 
\end{proposition*}
\begin{proof}
We begin by explaining why, if $\BBLMTU(S)$ supports reliable computation, then $S$ renders communication complexity trivial.  This direction follows the reasoning of \cite{Brassard}.

Suppose that Alice and Bob would like to compute $f(\mathbf{u}, \mathbf{v})$. 
It is not hard to see 
that using only shared randomness, Alice and Bob can always come up with bits $x$ and $y$ respectively so that the marginals of each of $x$ and $y$ are uniform, and so
that 
\[ \mathbb{P}[ x \oplus y = f(\mathbf{u},\mathbf{v}) ] \geq \frac{1}{2} + \gamma_n,\] where $\gamma_n > 0$ may depend on $n$.  
Indeed, suppose that Alice and Bob flip $n$ shared random coins to get $\mathbf{r}$.  Bob assumes that $\mathbf{u} = \mathbf{r}$, and computes $y = f(\mathbf{r}, \mathbf{v})$.  Alice produces a single bit $x$ which is $0$ if indeed $\mathbf{u} = \mathbf{r}$, and otherwise is uniformly random.  
Then $\mathbb{P}[x \oplus y = f(\mathbf{u}, \mathbf{v})] \geq \frac{1}{2} + \frac{1}{2^n}.$

However, Alice and Bob are after a success probability of $1/2 + \eps$ for some constant $\eps > 0$.  Thus, they would like to amplify their success probability.  If $\BBLMTU(S)$ supports reliable computation, then 
in particular $\BBLMTU(S)$ contains a circuit $\mathrm{Amp}: \{0,1\}^t \to \{0,1\}$ that acts as an \em amplifier \em (Definition~\ref{define:amplifier}).  That is, given independent random bits $z_1, \ldots, z_t$ with bias $p>1/2$ (resp. $p<1/2$), $\mathrm{Amp}(z_1, \ldots, z_t)$ outputs a bit $a$ that is very likely to be $1$ (resp. $0$).  Indeed, $\mathrm{Amp}$ is simply the circuit that implements the Majority function.

Alice and Bob repeat the procedure above $t$ times independently to obtain $x_1, \ldots, x_t$ and $y_1, \ldots, y_t$. Then they can use protocol using $S$ that corresponds to the amplifier $\mathrm{Amp} \in \BBLMTU(S)$ to obtain final bits $x$ and $y$ so that 
\[ x \oplus y = \mathrm{Amp}( x_1 \oplus y_1, \ldots, x_t \oplus y_t ) = \mathrm{Amp}( z_1, \ldots, z_t ). \]
Finally, Alice sends the single bit $x$ to Bob, who outputs $x \oplus y$.  By construction, Bob's output is very likely to be equal to $f(\mathbf{u}, \mathbf{v})$.

For the other direction, suppose that Alice and Bob can use $S$ to compute any function with $t$ bits of communication each and with probability at least $1/2 + \eps$ for some constants $\eps > 0$ and $t \geq 0$.  We claim that, without loss of generality, the communication can come in the form of a single bit $a$ that Alice sends to Bob at the end of the computation, and moreover that Bob outputs $a \oplus b$ for some bit $b$ that he has computed locally.

Indeed, suppose that there is a protocol $\Pi$ for Alice and Bob to compute $f(\mathbf{u}, \mathbf{v})$ with the guarantees above, where Bob outputs the final answer.  Suppose that, in $\Pi$, Alice would send the bits $a_1, \ldots, a_t$ to Bob, and Bob would send the bits $b_1, \ldots, b_t$ to Alice.  Then consider the following modification of $\Pi$.  Alice and Bob use shared randomness to obtain random bits $r_1, \ldots, r_t, q_1, \ldots, q_t$.  Alice assumes that $\mathbf{b} = \mathbf{q}$ and computes her responses $\mathbf{a}$ accordingly.  If $\mathbf{a} = \mathbf{r}$, then Alice sends Bob the bit $a = 0$; otherwise she sends a uniformly random bit $a$.  Meanwhile, Bob assumes that $\mathbf{a} = \mathbf{r}$ and computes his responses $\mathbf{b}$ and the outcome $\Pi_B(\mathbf{r}, \mathbf{b}, \mathbf{v})$ of running the protocol $\Pi$ on Alice's assumed responses and his own input and responses.  Then Bob also computes a bit $y$ which is $0$ if $\mathbf{b} = \mathbf{q}$ and uniformly random otherwise and sets $b = y \oplus \Pi_B(\mathbf{r}, \mathbf{b}, \mathbf{v})$.  Finally, Bob outputs $a \oplus b$.  If $\Pi$ correctly computed $f(\mathbf{u},\mathbf{v})$ with probability at least $1/2 + \eps$, then this new protocol computes $f(\mathbf{u}, \mathbf{v})$ with probability at least $1/2 + \eps \cdot 2^{-2t}$.  Since $t$ is a constant, Alice and Bob still compute $f$ with a constant advantage.

Now let $g: \{0,1\}^n \to \{0,1\}$ be any function, and define $f(\mathbf{u}, \mathbf{v}) := g(\mathbf{u} \oplus \mathbf{v})$, where $\oplus$ is defined coordinate-wise.  Then there is a strategy $\Pi$ for Alice and Bob to compute $f$ using the correlations in $S$; we assume that $\Pi$ has the form described above.  Since $S$ is closed under wirings and the only communication in $\Pi$ is in the form of a single bit at the end of the protocol, there is some correlation $c \in S$ so that $\Pi$ consists of using $c(\mathbf{u}, \mathbf{v})$ to obtain bits $a,b$ for Alice and Bob respectively; then Alice sends $a$ to Bob and Bob outputs $a \oplus b$. 

Then the single gate $\BBLMTU(c)$ reliably computes the function $g$.  Indeed, we have
\[ \mathbb{P}[\BBLMTU(c)(\mathbf{z}) = g(\mathbf{z})] = \mathbb{P}[ \BBLMTU(c)(\mathbf{u}\oplus \mathbf{v}) = g(\mathbf{u} \oplus \mathbf{v}) ] = \mathbb{P}[c(\mathbf{u},\mathbf{v}) = f(\mathbf{u}, \mathbf{v})] \geq 1/2 + \eps. \]

Since $g$ was arbitrary, $\BBLMTU(S)$ can compute any function with constant advantage, meaning that $\BBLMTU(S)$ supports reliable computation.
\end{proof}

\section{Proof of Claim~\ref{claim:NMapExists}: Existence of the map $\mathbf{N}$}\label{appendix:NMapExists}
\subsection{Theory of Amplification}

In this section we collect a few useful definitions and preliminary lemmas to reason about amplifiers. The full proof of Claim~\ref{claim:NMapExists}, given in Section~\ref{ssection:proofNMapExists}, will make use of these definitions and generalizations of the lemmas.

We begin with a useful definition, which defines the \em dual \em of a stochastic map $c$.
Below, for stochastic maps $f:\{0,1\}^k \to \{0,1\}$ and $g:\{0,1\}^m \to \{0,1\}$, we use the notation $f \circ g$ to mean the function $f\circ g: \{0,1\}^{mk} \to \{0,1\}$ given by
\[ (f\circ g)(\mathbf{x}^{(1)}, \ldots, \mathbf{x}^{(k)}) = f( g(\mathbf{x}^{(1)}), \ldots, g(\mathbf{x}^{(k)}) ) \]
where each $\mathbf{x}^{(j)} \in \{0,1\}^m$.
\begin{define}[Dual]
The {\bf dual} of a stochastic map $c$ is defined by $\dual(c) = \neg \circ c \circ \neg$.
\end{define}
Note that the dual has amplification function
$$A_{\dual(c)}(p) = A_{\neg}(A_c(A_{\neg}(p))) = 1 - A_c(1-p).$$

Our first lemma shows that if there is any stochastic $c \in \conv\cC$ with $A'_c(1/2) > 1$, then there is some other stochastic map in $\conv \cC$ that amplifies away from $1/2$.
\begin{lem}\label{lem:ampViaSymmetrization}
If there exists a circuit $c \in \conv\mathcal{C}$ such that $A_c'(1/2)> 1$, and $\neg \in \conv\mathcal{C}$, then $\exists f \in \conv \mathcal{C}$ such that $A_{f}'(1/2)  >1, A_{f}(1/2) = 1/2$.
\end{lem}
\begin{proof}
Defining the stochastic map $f$ as the uniform distribution over $c$ and $\dual(c)$, we have
$$A_f(p) = \frac{1}{2}(A_{c}(p) + 1 - A_{c}(1-p))$$
which implies that
$$A_f(1/2) = \frac{1}{2},\,\, A_f'(1/2) = A_{c}'(1/2) > 1$$
and therefore $f \in \conv\mathcal{C}$ is an amplifier away from $1/2$.
\end{proof}

Next we show how to convert an amplifier away from a point $p_0$ to one which amplifies away from $1/2$.
\begin{lem}\label{lem:amplifierEquivalence}
Let $\mathcal{C}$ denote a set of circuits closed under composition and including the constant functions $0, 1$. 
Let $p_0 \in (0,1)$ and suppose that there is some $c \in \cC$ so that
\[ A_{c}'(p_0) > 1, A_{c}(p_0) = p_0. \]
Then there exists $f \in \conv\cC$ so that
\[ A_{f}'(1/2) > 1, A_{f}(1/2) = 1/2. \]
\end{lem}
\begin{proof}
Suppose we have a stochastic map $c$ taking $n$ inputs that amplifies away from a point $p_0 \in (0, 1) \setminus \{ 1/2 \}$. We will show how to construct a stochastic map $f$ that amplifies away from $1/2$. 
If $p_0 < 1/2$, we may instead choose the stochastic map $\dual(c) \in \conv\mathcal{C}$ that amplifies away from $1-p_0$. Hence, without loss of generality we assume $p_0 > 1/2$. For $r < 1$, let
\[ m_r = r \cdot \mathbf{x} + (1 - r) \cdot \mathbf{1} \]
denote the stochastic map on a single input bit $x$ which returns $1$ with probability $1-r$ and $x$ with probability $r$. The amplification function of $m_r$ is
$$A_{m_r}(p) = rp + 1-r.$$
It is easy to see that $A_{c\circ m_r} = A_c \circ A_{m_r}.$
Choosing $r = 2(1-p_0) < 1$,  we have that
$$A_{c\circ m_r}(1/2) = A_c(p_0) = p_0$$
and
$$A_{c\circ m_r}'(1/2) = 2(1-p_0)A_c'(p_0).$$
Since $2(1-p_0) \in (0, 1)$, as long as $A_c'(p_0) > \frac{1}{2(1-p_0)}$ we will have $A_{c\circ m_r}'(1/2) > 1$. With this in mind, we wish to construct an amplifier $b\in \conv\mathcal{C}$ away from $p_0$ so that $A_b'(p_0) > \frac{1}{2(1 - p_0)}$.

Define $c^{\circ k}$ to be the stochastic map 
$$c^{\circ k} := \underbrace{c \circ c \circ ... \circ c }_{k \text{ times}},$$
so $c^{\circ k}$ takes $n^k$ inputs.
The amplification function of $c^{\circ k}$ is given by 
$$A_{c^{\circ k}}(p) = \underbrace{A_c \circ A_c \circ ... \circ A_c }_{k \text{ times}}.$$
This implies that that for all $k \geq 1$, the value and derivative at $p_0$ obey
\begin{align}
A_{c^{\circ k}}(p_0) &= p_0, \notag \\
A_{c^{\circ k}}'(p_0) &= \left[ A_{c}'(p_0) \right]^k.\label{eq:derivativeAtFixedPointKfoldSelfComposition}
\end{align}
In particular, since $A_{c}'(p_0) > 1$, \eqref{eq:derivativeAtFixedPointKfoldSelfComposition} implies that there is some $k'$ so that
$$A_{c^{\circ k'}}'(p_0) > \frac{1}{2(1-p_0)}.$$
Choosing $b = c^{\circ k'}$ for this $k'$, we have
$$A_{b \circ m_r}'(1/2) = 2(1-p_0)A_c'(p_0)^k > 1.$$
Then by Lemma \ref{lem:ampViaSymmetrization}, $\conv\mathcal{C}$ contains an amplifier $f$ away from $1/2$.

\end{proof}

\subsection{Self-Dual Amplifier Lemma}
We will first prove Lemma~\ref{lem:existsSelfDualAmplfiier}, stated below, which extends Lemma \ref{lem:ampViaSymmetrization} to show that if $\conv \mathcal{C}$ contains an amplifier and a $\neg_\kappa$ gate then there is a {\it self-dual} amplifier away from $1/2$ in $\conv \mathcal{C}$. Then we will prove Claim~\ref{claim:NMapExists} using Lemma~\ref{lem:existsSelfDualAmplfiier}.
\begin{lem}\label{lem:existsSelfDualAmplfiier}
If $\mathcal{C}$ contains an amplifier and a $\neg_\kappa$ gate for some $\kappa < 1/2$, then there exists $c \in  \conv \mathcal{C}$ such that $c$ is an amplifier away from $1/2$ and $A_c(p) = 1- A_c(1-p)$.
\begin{proof}

Suppose there exists and amplifier $\Amp \in \conv \mathcal{C}$ such that $\Amp$ amplifies away from $p_0$. We wish to construct $c \in \conv \mathcal{C}$ such that $c$ amplifies away from $1/2$ and $A_{c}(p) = 1-A_{c}(1-p)$. We will reuse some ideas from the proof of Lemma~\ref{lem:amplifierEquivalence}.  Without loss of generality, assume $p_0 > 1/2$. Let $\mathbf{0}, \mathbf{1}$ denote the constant $0$ and $1$ function respectively. Let $\mathbf{x}$ denote the identity function on one bit. First define the ``noise gate'' $\mathcal{N}_\ell$, as the stochastic map
\begin{equation}\label{eq:NoiseGateDefinition}
\mathcal{N}_\ell := \ell \mathbf{x} + \frac{1-\ell}{2}(\mathbf{0} + \mathbf{1}),
\end{equation}
which has amplification function
\begin{equation}\label{eq:AmpFuncNell}
A_{\mathcal{N}_\ell}(p) = \ell p + \frac{1-\ell}{2}.
\end{equation}
We will reuse the $m_r$ map defined in the proof of Lemma~\ref{lem:amplifierEquivalence}, 
$$m_r := r\mathbf{x} + (1-r)\mathbf{1} \in \conv \mathcal{C}.$$
As in the proof of Lemma~\ref{lem:amplifierEquivalence} we will set $r = 2(1-p_0)$ and make use of the stochastic map
$$\Amp^{\circ k} \circ m_{2(1-p_0)}.$$
Recall that this map satisfies:
\begin{align}
&A_{\Amp^{\circ k} \circ m_{2(1-p_0)}} (1/2) = p_0 \label{eq:oneHalfMappedDilatedAmp} \\
&A_{\Amp^{\circ k} \circ m_{2(1-p_0)}}' (1/2) = 2(1-p_0)A_{\Amp}'(p_0)^k. \label{eq:oneHalfMappedDilatedAmpDeriv}
\end{align}

We will use the following mixture as our self-dual amplifier away from $1/2$:
\begin{equation}\label{eq:selfDualAmplifierDefn}
c = \frac{1}{2}\left[ \neg_\kappa \circ \left(\Amp^{\circ k} \circ m_{2(1-p_0)} \right) \circ \neg_\kappa + \mathcal{N}_{1-2\kappa} \circ \left(\Amp^{\circ k} \circ m_{2(1-p_0)} \right) \circ \mathcal{N}_{1-2\kappa}  \right].
\end{equation}

Note that $1-2\kappa > 0$ because $\kappa < 1/2$.
We will now show that for any $\kappa < 1/2$, we may choose sufficiently large $k$ such that $c$ is an amplifier away from $1/2$ with $A_c(p) = 1-A_c(1-p)$. The amplification function of $\neg_\kappa$ is
\begin{equation}\label{eq:AmpFuncNOTeps}
A_{\neg_\kappa}(p) = (1-\kappa)(1-p) + \kappa p = 1-\kappa - (1 - 2\kappa)p.
\end{equation}

{\bf Check that $1/2$ is a Fixed Point of $A_c(p)$:}
First, we verify that $A_c(1/2) = 1/2$.
Clearly,
$$A_{\neg_\kappa}(1/2) = A_{\mathcal{N}_\ell}(1/2) = \frac{1}{2},$$
and by \eqref{eq:oneHalfMappedDilatedAmp}, plugging in \eqref{eq:AmpFuncNOTeps} and \eqref{eq:AmpFuncNell}, we find that
\begin{align}
A_c(1/2) &= \frac{1}{2}\left[ A_{\neg_\kappa}(p_0) + A_{\mathcal{N}_{1-2\kappa}}(p_0) \right]\notag \\
&= \frac{1}{2}\left[ 1 - \kappa + p_0(2\kappa - 1) + \kappa + (1-2\kappa)p_0 \right] \notag \\
&= \frac{1}{2}\notag
\end{align}
and so this condition is satisfied for any choice of $k$.

{\bf Check that $c$ is Self-Dual:}
Next we show that
$$A_c(p) = 1-A_c(1-p).$$
This condition holds as long as
\begin{equation}\label{eq:selfDualAmpDef}
A_c = A_{\dual c} = A_{\neg \circ c \circ \neg}.
\end{equation}
To verify this condition, we will first compute the composition of $\neg$ with the noise gate and noisy $\neg$ gate, enabling us to compute $A_{\dual c}$. The composition of $\neg$ with the noise gate $\mathcal{N}_{1-2\kappa}$ is
$$A_{\neg \circ \mathcal{N}_{1-2\kappa}}(p) = (1-2\kappa)(1-p) + \frac{2\kappa}{2} = 1-\kappa -  (1-2\kappa) p = A_{\neg_\kappa},$$
which implies further that
$$A_{\neg \circ \neg_\kappa} = A_{\neg \circ \neg \circ \mathcal{N}_{1-2\kappa}} = A_{\mathcal{N}_{1-2\kappa}},$$
giving us the $\neg$ composition with the noisy $\neg$ gate.
Using these composition relations, it is easy to verify equation \eqref{eq:selfDualAmpDef}. Specifically, we have that
\begin{align}
A_{\dual c}  &= A_{\frac{1}{2}\neg \circ \left[ \neg_\kappa \circ \left(\Amp^{\circ k} \circ m_{2(1-p_0)} \right) \circ \neg_\kappa + \mathcal{N}_{1-2\kappa} \circ \left(\Amp^{\circ k} \circ m_{2(1-p_0)} \right) \circ \mathcal{N}_{1-2\kappa}  \right]\circ \neg}  \\
&=A_{\frac{1}{2}\left[\neg \circ \neg_\kappa \circ \left(\Amp^{\circ k} \circ m_{2(1-p_0)} \right) \circ \neg_\kappa \circ \neg +\neg \circ \mathcal{N}_{1-2\kappa} \circ \left(\Amp^{\circ k} \circ m_{2(1-p_0)} \right) \circ \mathcal{N}_{1-2\kappa} \circ \neg \right]} \notag \\
&=A_{\frac{1}{2}\left[\mathcal{N}_{1-2\kappa} \circ \left(\Amp^{\circ k} \circ m_{2(1-p_0)} \right) \circ \mathcal{N}_{1-2\kappa} +\neg_\kappa \circ \left(\Amp^{\circ k} \circ m_{2(1-p_0)} \right) \circ \neg_\kappa \right]} \notag \\
&= A_c \notag
\end{align}
and therefore $c$ is self-dual for any choice of $k$.

{\bf Check that $c$ satisfies $A_c'(1/2) > 1$:}
Finally, we must show that we may choose $k$ such that $A_c'(1/2) > 1$.
We compute the derivative of the amplification function of one term in equation \eqref{eq:selfDualAmplifierDefn},
$$A_{\mathcal{N}_{1-2\kappa}\circ \left(\Amp^{\circ k} \circ m_{2(1-p_0)} \right)\circ \mathcal{N}_{1-2\kappa}}(p) = \kappa + (1-2\kappa) A_{\Amp^{\circ k} \circ m_{2(1-p_0)} }(\kappa + (1-2\kappa)p)$$
which gives
\begin{align}
A_{\mathcal{N}_{1-2\kappa}\circ \left(\Amp^{\circ k} \circ m_{2(1-p_0)} \right)\circ \mathcal{N}_{1-2\kappa}}'(p) &=  (1-2\kappa)^2 A_{\Amp^{\circ k} \circ m_{2(1-p_0)} }'(\kappa + (1-2\kappa)p) \notag \\
A_{\mathcal{N}_{1-2\kappa}\circ \left(\Amp^{\circ k} \circ m_{2(1-p_0)} \right)\circ \mathcal{N}_{1-2\kappa}}'(1/2) &=  (1-2\kappa)^2 A_{\Amp^{\circ k} \circ m_{2(1-p_0)} }'(1/2), \notag 
\end{align}
which by equation \eqref{eq:oneHalfMappedDilatedAmpDeriv} becomes
$$A_{\mathcal{N}_{1-2\kappa}\circ \left(\Amp^{\circ k} \circ m_{2(1-p_0)} \right)\circ \mathcal{N}_{1-2\kappa}}'(1/2) =  2 (1-2\kappa)^2 (1-p_0)A_{\Amp}'(p_0)^k.$$
Since $(1-2\kappa)^2, (1-p_0) > 0$ and $A_{\Amp}'(p_0) >1$, there exists finite $k_0 $ such that for all $k \geq k_0$,
$$A_{\mathcal{N}_{1-2\kappa}\circ \left(\Amp^{\circ k} \circ m_{2(1-p_0)} \right)\circ \mathcal{N}_{1-2\kappa}}'(1/2) > 1.$$
The derivative at $1/2$ is the same for the amplification function of both terms in \eqref{eq:selfDualAmplifierDefn} because they are dual to each other.
Therefore we choose $k=k_0$ and have that $c \in \conv \mathcal{C}$ is a self-dual amplifier away from $1/2$.
\end{proof}
\end{lem}

\subsection{Proof of Claim~\ref{claim:NMapExists}}\label{ssection:proofNMapExists}
Recall Claim~\ref{claim:NMapExists}:

\begin{claim*}[Claim~\ref{claim:NMapExists}, restated]
Let $\mathcal{C}$ denote a circuit model closed under composition. Suppose $\mathcal{C}$ contains an amplifier and a $\neg_\kappa$ gate for some $\kappa < 1/2$. Then there exists $\beta \in (0, 1/2]$, an integer $m \geq 1$, and a map $\mathbf{N} \in \conv \mathcal{C}$ such that $\mathbf{N}$ takes $2m$ inputs and so that the following holds.  Letting
\begin{align}
I_- &= \left[\frac{1}{2}-\beta, \frac{1}{2}- \frac{\beta}{2}\right]\notag \\
I_+ &= \left[\frac{1}{2} + \frac{\beta}{2},  \frac{1}{2} + \beta\right], \notag
\end{align}
we have
\begin{align}
\psi_\mathbf{N}( (I_-)^m \times (I_-)^m) &\subseteq I_+ \notag \\
\psi_\mathbf{N}( (I_+)^m \times (I_-)^m) &\subseteq I_+  \notag \\
\psi_\mathbf{N}( (I_-)^m \times (I_+)^m) &\subseteq I_+  \notag \\
\psi_\mathbf{N}( (I_+)^m \times (I_+)^m) &\subseteq I_-. \notag
\end{align}
\end{claim*}

\begin{proof}[Proof of Claim~\ref{claim:NMapExists}]
First, by Lemma \ref{lem:existsSelfDualAmplfiier}, we can construct a self-dual amplifier $c\in \conv \mathcal{C}$ away from $1/2$
using $\Amp$ and $\neg_\kappa$.
Let $p_0, p_1 \in (0, 1) \setminus \{ 1/2 \}$ denote the two fixed points of $A_{c}(p)$
adjacent to $p=1/2$, with $p_0 < 1/2 < p_1$. 
Notice that these exist because $c$ is an amplifier, and further 
that by the self-duality of $c$, $p_0 = 1-p_1$.

To construct $\mathbf{N}$, we choose any $\beta$ such that $0 < \beta < \min\left\{p_1 - \frac{1}{2}, \frac{1}{2} - p_0 \right\} = p_1 -\frac{1}{2}$. Next we will need some ingredients. The first ingredient is the noise gate $\mathcal{N}_\ell$, which was defined in equation \eqref{eq:NoiseGateDefinition}, restated here:
$$\mathcal{N}_\ell := \ell \mathbf{x} + \frac{1-\ell}{2}(\mathbf{0} + \mathbf{1}).$$
The noise gate will be useful for ensuring that inputs have {\it sufficient} noise to occupy $I_- \cup I_+$.

The second ingredient is $D_{r,s}$, which accepts two input bits $x_1, x_2$, and is similar to an $\land$ gate. 
For $j \in \{1,2\}$,
Let $\mathbf{x}_j$ denote the map accepting $2$ input bits $(x_1, x_2)$ and outputting the $j$th input $x_j$. Then we define the stochastic  map $D_{r, s}:\F_2^2 \to \F_2$ as the mixture
$$D_{r,s} := r\mathbf{0} + s\mathbf{1} + \frac{1-r - s}{2}\left( \mathbf{x}_1 +  \mathbf{x}_2 \right).$$

The final ingredient is a reordering map $\mathcal{R}_n$, which we need just for notational purposes.  
For a positive integer $n$, the map
$$\mathcal{R}_n : \F_2^{2n} \rightarrow \F_2^{2n}$$
is defined by
$$\left(\mathcal{R}_n(\mathbf{x}) \right)_{k} = \begin{cases}
\mathbf{x}_{k/2} & k \text{ even} \\
\mathbf{x}_{n + (k+1)/2} & k \text{ odd} \\
\end{cases}$$
for $k = 1, \ldots, 2n$.
For example, $\mathcal{R}_4(11110000) = 01010101$. 
We will use $\mathcal{R}_n$ to ensure that the inputs are in the right order to satisfy equation \eqref{eq:intervalMappingsNoisyNANDsimulator}. 
As noted at the beginning of this appendix, we have been using the shorthand $c \circ c'$ to mean $c \circ (c')^{\otimes n}$ when $c: \F_2^n \to \F_2$ and $c': \F_2^m \to \F_2$; 
below, for $g: \F_2^{2n} \to \F_2$, we will use $g \circ \mathcal{R}_n$ to mean composition in the usual sense (not using our shorthand).

Now we can construct $\mathbf{N}$. We will use the following map, parametrized by $\ell_0, \ell_1, r, s, k$:
\begin{equation}\label{eq:defNmapParametrized}
\mathbf{N}_{r,s,\ell_0,\ell_1, k} = \left[\mathcal{N}_{\ell_1} \circ c^{\circ k} \circ \mathcal{N}_{\ell_0} \circ\neg_\kappa \circ D_{r,s} \right] \circ \mathcal{R}_{r^k}.
\end{equation}
This satisfies $\mathbf{N}_{r,s,\ell_0,\ell_1} \in \conv \mathcal{C}$. Let $c$ take $r$ inputs. Then $\mathbf{N}_{r,s,\ell_0,\ell_1}$ takes $2r^k$ inputs. Therefore, the following Claim~\ref{claim:NMapParametersExists} implies  Claim~\ref{claim:NMapExists}:

\begin{claim}\label{claim:NMapParametersExists}
For all $\beta \in (0, p_1 - \frac{1}{2})$ and $\kappa < 1/2$, there exist $\ell_0, \ell_1 \in (0, 1]$, and $r, s \in [0, 1]$ with $r+s < 1$, and $k \geq 1$ such that $\mathbf{N}_{r,s,\ell_0,\ell_1, k}$ as defined in equation \eqref{eq:defNmapParametrized} satisfies equation \eqref{eq:intervalMappingsNoisyNANDsimulator}.
\end{claim}

\begin{proof}
Analyzing the $D_{r,s}$ map, we see that for inputs $x_j \sim \Ber(\frac{1}{2} + \kappa_k)$ for $j \in \{1,2\}$ and $r, s$, we have 
\[ D_{r,s} \sim \Ber(q) \qquad \text{for} \qquad q = \frac{1 + \kappa_1 + \kappa_2}{2} (1-r-s) + s. \]
We would like to choose $r$ and $s$ so that $1/2 < q $ if and only if $\kappa_1, \kappa_2 > 0$, while $  q < 1/2$ otherwise.
For all inputs $(\frac{1}{2} + \kappa_1, \frac{1}{2} + \kappa_2) \in  (I_- \cup I_+)^2$ such that there is some $j \in \{1,2\}$ such that $\kappa_j < 0$, we have

$$\kappa_1 + \kappa_2 \leq  - \frac{\beta}{2}  + \beta = \frac{\beta}{2}$$
and we require
$$q \leq \frac{1  + \frac{\beta}{2}}{2} (1-r-s) + s < 1/2$$
which happens if and only if
\begin{equation}\label{eq:sUpperBd}
s  < \frac{1 -  (1  + \frac{\beta}{2})(1-r)}{1 - \frac{\beta}{2}}.
\end{equation}
For all inputs $(\frac{1}{2} + \kappa_1, \frac{1}{2} + \kappa_2) \in  (I_- \cup I_+)^2$ such that $\kappa_1, \kappa_2 > 0$, we have
$$\kappa_1 + \kappa_2 \geq \beta$$
and we require
$$q \geq \frac{1 + \beta}{2}(1-r-s) + s > 1/2$$
which happens if and only if
\begin{equation}\label{eq:sLowerBd}
s  > \frac{1 - (1 + \beta)(1-r)}{1-\beta}.
\end{equation}

If our choice of $r, s$ satisfies both \eqref{eq:sUpperBd} and \eqref{eq:sLowerBd}, then the output bit will have positive bias if and only if both input bits have positive bias, so $D_{r, s}$ will function effectively similar to an $\land$ gate. We will choose $r = 1/4$ and show that we may always choose $s$ (depending on $\beta$) so that \eqref{eq:sUpperBd} and \eqref{eq:sLowerBd} are satisfied. With $r=1/4$, our requirements on $s$ become
$$s > \frac{1 - 3 \, \beta }{4 \, {\left(1 - \beta \right)}}$$
and
$$s < \frac{3 \, \beta - 2}{4 \, {\left(\beta - 2\right)}}.$$
Since we must have $s \in [0, 1-r]$, there exists a suitable choice of $s$ to satisfy both \eqref{eq:sUpperBd} and \eqref{eq:sLowerBd} for each $\beta \in (0, 1/2]$ if and only if the following three inequalities are satisfied:
\begin{align*}
\frac{1 - 3 \, \beta }{4 \, {\left(1 - \beta \right)}} & \leq \frac{3}{4}, \\ 
 0 &\leq \frac{3 \, \beta - 2}{4 \, {\left(\beta - 2\right)}} ,\\ 
\frac{1 - 3 \, \beta }{4 \, {\left(1 - \beta \right)}} & \leq \frac{3 \, \beta - 2}{4 \, {\left(\beta - 2\right)}}. 
\end{align*}
We note that $\beta \leq p_1 - 1/2 \leq 1/2$, and it is not hard to see that the above are satisfied for any $\beta \leq 1/2$.

Thus, for our $\beta$, there exists $ r,s \in [0, 1]$ such that $r + s \leq 1$, and denoting $\mathbf{D} := D_{r,s}$ with this choice of $r, s$, $\mathbf{D}$ satisfies the following:
\begin{align}
\psi_{ \mathbf{D}}\left((I_-\times I_-) \cup (I_-\times I_+) \cup (I_+\times I_-) \right) &\subseteq [0, 1/2) \label{eq:intervalMappingsDrs} \\
\psi_{ \mathbf{D}}\left(I_+\times I_+\right) &\subseteq (1/2, 1]. \notag
\end{align}

It is also straightforward to show that for $\kappa < 1/2$, the following are satisfied:
\begin{align}
\psi_{\neg_\kappa}([0, 1/2)) &\subseteq (1/2, 1] \label{eq:intervalMappingsNOTeps}\\
\psi_{\neg_\kappa}((1/2, 1]) &\subseteq [0, 1/2). \notag
\end{align}
Therefore $\neg_\kappa \circ \mathbf{D}$ satisfies the following:
\begin{align}
\psi_{ \neg_\kappa \circ \mathbf{D}} \left((I_-\times I_-) \cup (I_-\times I_+) \cup (I_+\times I_-) \right) &\subseteq (1/2, 1] \\
\psi_{ \neg_\kappa \circ \mathbf{D}} \left(I_+\times I_+\right) &\subseteq [0, 1/2). \notag
\end{align}

Fixing any $\ell_0 \in (0, \beta)$, we have that $\mathcal{N}_{\ell_0} \circ \neg_\kappa \circ \mathbf{D}$ satisfies
\begin{align}
\psi_{ \mathcal{N}_{\ell_0} \circ\neg_\kappa \circ \mathbf{D}} \left((I_-\times I_-) \cup (I_-\times I_+) \cup (I_+\times I_-) \right) &\subseteq (1/2, 1/2 + \beta/2] \\
\psi_{ \mathcal{N}_{\ell_0} \circ\neg_\kappa \circ \mathbf{D}} \left(I_+\times I_+\right)  &\subseteq [1/2 - \beta/2, 1/2)\notag
\end{align}
With all our outputs as Bernoulli random variables occupying $(p_0, 1/2) \cup (1/2, p_1)$, we can amplify using $c^{\circ k}$. Let $\Delta := p_1 - \frac{1}{2} = \frac{1}{2} - p_0$. In particular, there exists sufficiently large $k$ such that, denoting
\begin{align*}
S_- &= (I_-\times I_-) \cup (I_-\times I_+) \cup (I_+\times I_-)\\
S_+ &= I_+\times I_+
\end{align*}
we have
\begin{align}
\psi_{c^{\circ k} \circ \mathcal{N}_{\ell_0} \circ\neg_\kappa \circ \mathbf{D}} \left( S_-^{r^k} \right) &\subseteq (1/2 + \Delta / 2, 1/2 + \Delta]  \\
\psi_{c^{\circ k} \circ \mathcal{N}_{\ell_0} \circ\neg_\kappa \circ \mathbf{D}}\left( S_+^{r^k}\right) &\subseteq [1/2 - \Delta, 1/2-\Delta/2)\notag
\end{align}
and choosing $\ell_1 := \frac{\beta}{\Delta}$ we have that
\begin{align}
\psi_{\mathcal{N}_{\ell_1} \circ c^{\circ k} \circ \mathcal{N}_{\ell_0} \circ\neg_\kappa \circ \mathbf{D}}(S_-^{r^k}) &\subseteq I_+ \label{eq:semiFinalNconstruction} \\
\psi_{\mathcal{N}_{\ell_1} \circ c^{\circ k} \circ \mathcal{N}_{\ell_0} \circ\neg_\kappa \circ \mathbf{D}}(S_+^{r^k}) &\subseteq I_-.\notag
\end{align}
Inserting the reordering map prior to evaluation changes the ordering of the domains specified in equation \eqref{eq:semiFinalNconstruction}, so that
\begin{align*}
\psi_{\mathbf{N}_{r,s,\ell_0,\ell_1, k}}\left(((I_-)^{r^k}\times (I_-)^{r^k}) \cup ((I_-)^{r^k}\times (I_+)^{r^k}) \cup ((I_+)^{r^k}\times (I_-)^{r^k})\right) &\subseteq I_+ \\
\psi_{\mathbf{N}_{r,s,\ell_0,\ell_1, k}}\left((I_+)^{r^k} \times (I_+)^{r^k}\right) &\subseteq I_-
\end{align*}
as desired.

Therefore letting $\mathbf{N} = \left[\mathcal{N}_{\ell_1} \circ c^{\circ k} \circ \mathcal{N}_{\ell_0} \circ\neg_\kappa \circ\mathbf{D} \right] \circ \mathcal{R}_{r^k}$, we have the map $\mathbf{N} \in \conv \mathcal{C}$ accepting $2r^k$ inputs and satisfying relations \eqref{eq:intervalMappingsNoisyNANDsimulator}.
This proves Claim~\ref{claim:NMapParametersExists}, which completes the proof of
 Claim~\ref{claim:NMapExists}.
\end{proof}

\end{proof}

\end{appendices}

\bibliographystyle{alpha}
\bibliography{refs}

\end{document}